\documentclass[12pt]{article}
\usepackage{booktabs}
\usepackage{threeparttable}
\usepackage[colorlinks,bookmarksopen,bookmarksnumbered,citecolor=blue,urlcolor=blue,linkcolor=blue]{hyperref}
\usepackage{float}
\usepackage{amsmath}

\usepackage{graphicx}
\graphicspath{{Pictures/}}

\newcommand{\rank}{\mathrm{rank}}

\usepackage{xcolor}
\usepackage[draft,inline,nomargin,index]{fixme}
\fxsetup{theme=color,mode=multiuser}
\FXRegisterAuthor{yf}{ayf}{\color{magenta}}

\def\spacingset#1{\renewcommand{\baselinestretch}%
	{#1}\small\normalsize} \spacingset{1}

\usepackage{amsthm,amsmath,enumerate,amsbsy,amsfonts,amssymb,mathabx,amscd,graphicx,algorithm, algpseudocode,indentfirst,makecell}
\usepackage{natbib}
\usepackage{geometry}
\usepackage{authblk, caption, subcaption}
\usepackage{mathrsfs}
\usepackage{epsfig}

\newtheorem{theorem}{Theorem}
\newtheorem{corollary}{Corollary}
\newtheorem{lemma}{Lemma}

\newtheorem{proposition}{Proposition}
\newtheorem{remark}{Remark}

\newtheorem{condition}{Condition}

\usepackage{multirow}

\newcommand{\ols}{{\rm ols}}

\newcommand{\bzero}{\boldsymbol 0}
\newcommand{\bLambda}{\boldsymbol \Lambda}

\newcommand{\bOmega}{\boldsymbol \Omega}
\newcommand{\bSigma}{\boldsymbol \Sigma}

\newcommand{\bGamma}{\boldsymbol \Gamma}
\newcommand{\bTheta}{\boldsymbol \Theta}
\newcommand{\bPi}{\boldsymbol \Pi}

\newcommand{\bXi}{\boldsymbol \Xi}
\newcommand{\bpsi}{\boldsymbol \psi}
\newcommand{\bPsi}{\boldsymbol \Psi}
\newcommand{\bphi}{\boldsymbol \phi}
\newcommand{\bPhi}{\boldsymbol \Phi}

\newcommand{\bxi}{\boldsymbol \xi}

\newcommand{\bnu}{\boldsymbol \nu}

\newcommand{\bvarepsilon}{\boldsymbol \varepsilon}

\newcommand{\be}{{\mathbf e}}

\newcommand{\bx}{{\mathbf x}}
\newcommand{\by}{{\mathbf y}}
\newcommand{\bu}{{\mathbf u}}

\newcommand{\br}{{\mathbf r}}

\newcommand{\bs}{{\mathbf s}}
\newcommand{\bz}{{\mathbf z}}
\newcommand{\bbb}{{\mathbf b}}
\newcommand{\bc}{{\mathbf c}}

\newcommand{\bW}{{\bf W}}
\newcommand{\bD}{{\bf D}}
\newcommand{\bA}{{\bf A}}
\newcommand{\bB}{{\bf B}}
\newcommand{\bC}{{\bf C}}
\newcommand{\bE}{{\bf E}}
\newcommand{\bI}{{\bf I}}
\newcommand{\bK}{{\bf K}}
\newcommand{\bP}{{\bf P}}
\newcommand{\bS}{{\bf S}}
\newcommand{\bT}{{\bf T}}
\newcommand{\bX}{{\bf X}}
\newcommand{\bY}{{\bf Y}}
\newcommand{\bZ}{{\bf Z}}
\newcommand{\bR}{{\bf R}}
\newcommand{\bU}{{\bf U}}
\newcommand{\bV}{{\bf V}}
\newcommand{\bQ}{{\bf Q}}

\newcommand{\bM}{{\bf M}}

\newcommand{\bH}{{\bf H}}
\newcommand{\bG}{{\bf G}}
\newcommand{\bN}{{\bf N}}
\newcommand{\bL}{{\bf L}}

\newcommand{\cC}{{\cal C}}
\newcommand{\cD}{{\cal D}}

\newcommand{\cN}{{\cal N}}

\newcommand{\eZ}{\mathbb{Z}}
\newcommand{\eR}{\mathbb{R}}

\newcommand{\eE}{\mathbb{E}}

\newcommand{\eP}{\mathbb{P}}

\newcommand{\cov}{\text{Cov}}

\newcommand{\bic}{\text{BIC}}
\newcommand{\tr}{\mbox{tr}}
\newcommand{\diag}{\mbox{diag}}

\def\F{{ \mathrm{\scriptstyle F} }}
\def\IC{{ \mathrm{\scriptscriptstyle IC} }}

\makeatletter
\newcommand*{\rom}[1]{\expandafter\@slowromancap\romannumeral #1@}
\makeatother

\makeatletter
\DeclareRobustCommand\widecheck[1]{{\mathpalette\@widecheck{#1}}}
\def\@widecheck#1#2{%
	\setbox\z@\hbox{\m@th$#1#2$}%
	\setbox\tw@\hbox{\m@th$#1%
		\widehat{%
			\vrule\@width\z@\@height\ht\z@
			\vrule\@height\z@\@width\wd\z@}$}%
	\dp\tw@-\ht\z@
	\@tempdima\ht\z@ \advance\@tempdima2\ht\tw@ \divide\@tempdima\thr@@
	\setbox\tw@\hbox{%
		\raise\@tempdima\hbox{\scalebox{1}[-1]{\lower\@tempdima\box
				\tw@}}}%
	{\ooalign{\box\tw@ \cr \box\z@}}}
\makeatother
\RequirePackage[colorlinks,citecolor=blue,urlcolor=blue]{hyperref}

\def\spacingset#1{\renewcommand{\baselinestretch}%
	{#1}\small\normalsize} \spacingset{1}

\providecommand{\keywords}[1]{\textbf{\textit{Keywords: }} #1}
\newcommand{\blind}{1}

\setcitestyle{aysep={,},yysep={;}}

\begin{document}
	\if1\blind
	{
    	\spacingset{1.25}
		\title{\bf \Large Weight-calibrated estimation for factor models of high-dimensional time series\footnote{The authors' names are sorted alphabetically.}
            \hspace{.2cm}\\
		}
		\author[1]{Xinghao Qiao}
		\author[2]{Zihan Wang}
		\author[3]{Qiwei Yao}
            \author[4]{Bo Zhang}
            \affil[1]{\it Faculty of Business and Economics, The University of Hong Kong, Hong Kong}
	    \affil[2]{\it Department of Statistics and Data Science, Tsinghua University, Beijing, China}
		\affil[3]{\it Department of Statistics, London School of Economics, London, U.K.}
            \affil[4]{\it School of Management, University of Science and Technology of China, Hefei, China}
		\setcounter{Maxaffil}{0}
		
		\renewcommand\Affilfont{\itshape\small}
		\date{\vspace{-5ex}}
		\maketitle
	} \fi
	\if0\blind
	{   \spacingset{1.7}
            \bigskip
		\bigskip
		\bigskip
		\begin{center}
			{\Large\bf Weight-calibrated estimation for factor models of high-dimensional time series}
		\end{center}
		\medskip
	} \fi

\setcounter{Maxaffil}{0}
\renewcommand\Affilfont{\itshape\small}
\spacingset{1.25}

\begin{abstract}
    The factor modeling for high-dimensional time series is powerful in discovering latent common components for dimension reduction and information extraction. Most available estimation methods can be divided into two categories: the covariance-based under asymptotically-identifiable assumption and the  autocovariance-based with white idiosyncratic noise. This paper follows the autocovariance-based framework and develops a novel weight-calibrated method to improve the estimation performance. It adopts a linear projection to tackle high-dimensionality, and employs a reduced-rank autoregression formulation. The asymptotic theory of the proposed method is established, relaxing the assumption on white noise. Additionally, we make the first attempt in the literature by providing a systematic theoretical comparison among the covariance-based, the standard autocovariance-based, and our proposed weight-calibrated autocovariance-based methods in the presence of factors with different strengths. Extensive simulations are conducted to showcase the superior finite-sample performance of our proposed method, as well as to validate the newly established theory. The superiority of our proposal is further illustrated through the analysis of one financial and one macroeconomic data sets.
\end{abstract}


\newpage
\spacingset{1.7}

\section{Introduction}
\label{sec.intro}
High-dimensional time series have become indispensable in various fields such as economics, finance, climatology, medical research and others. However, conventional methods such as joint multivariate or componentwise univariate time series modeling often suffer from over-parametrization or omission of cross-sectional correlations. 
To overcome these difficulties, high-dimensional factor models have emerged as one of the most powerful approaches for dimension reduction and information extraction. Consider the factor model for a stationary $p$-vector time series $\{\by_t\}_{t\in\eZ}$:
\begin{equation}
    \label{eq.model}
    \by_t=\bA\bx_t+\be_t,~~t=1,\dots,n,
\end{equation}
where $\bx_t$ is a latent stationary $r_0$-vector factors, $\bA$ is a $p\times r_0$ full-ranked factor loading matrix, $\be_t$ is a stationary $p$-vector  idiosyncratic component, and the number of factors $r_0$ is unknown but assumed finite. Denote $\bOmega_{y}(k)=\cov(\by_t,\by_{t-k})$ for $k\in\eZ$ as the (auto)covariance matrix of $\by_t$\footnote{We refer to lag-$0$ and nonzero-lagged autocovariances as covariance and autocovariance, respectively.} with the corresponding sample estimate $\widehat \bOmega_y(k)$. For simplicity, we use $\bOmega_{y}$ to denote the covariance matrix $\bOmega_{y}(0)$. Similar definitions also apply to $\bOmega_{x}(k)$ and $\bOmega_{e}(k)$. 

The literature has mainly focused on two types of model assumptions on \eqref{eq.model}. 
The first type assumes that the $r_0$ eigenvalues of $\bA\bOmega_{x}\bA^{\top}$ diverge at rate $O(p)$ whereas all the eigenvalues of $\bOmega_e$ are bounded as $p \rightarrow \infty.$ This asymptotic identifiability assumption holds when 
the common factors can influence a non-vanishing proportion of components in $\by_t$ and the idiosyncratic components exhibit weak cross-sectional correlations.
The seminal work \cite{bai2002determining} proposes to estimate factors and loadings by solving a constrained least squares minimization problem, whose solution is equivalent to the principal component analysis (PCA) based on the sample covariance estimator $\widehat \bOmega_y$, thereby introducing a covariance-based approach. 
Other relevant literature includes, e.g., \cite{chamberlain1983arbitrage,stock2002forecasting,bai2003inferential,onatski2012asymptotics,bai2012statistical,connor2012efficient,fan2013large,bai2016maximum,bai2023approximate}, and extensions to matrix factor models \citep{yu2022projected,chen2023statistical}, functional factor models \citep{li2025factor} and tensor factor models \citep{chen2024rank,chen2026estimation}. 

The second type of models assumes that the dynamic information in $\by_t$ is entirely captured by the common factors, while the idiosyncratic component $\be_t$ is a white noise sequence with ${\mathbb E}(\be_t)={\bf 0}$ and $\bOmega_e(k)=0$ for any $k \neq 0,$ and is allowed to exhibit any strength of cross-sectional correlations.
Given that the autocovariance of $\by_t$ can filter out the impact of $\be_t$ automatically, \cite{lam2011estimation} proposed an autocovariance-based method to estimate the factor loading space and the number of factors $r_0$ through eigenanlaysis of $\sum_{k=1}^{m}\widehat\bOmega_{y}(k)\widehat\bOmega_{y}(k)^{\top},$ where $m$ is some prescribed positive integer. 
Other relevant literature includes, e.g., \cite{pena1987identifying,pan2008modelling,lam2012factor,gao2022modeling}, and extensions to matrix factor models \citep{wang2019,chen2020constrained,chang2023modelling}, functional factor models \citep{guo2026factor} and tensor factor models \citep{chen2022factor,han2024cp,han2024tensor}. In this paper, we follow the autocovariance-based estimation in our methodological development, but systematically compare it with the covariance-based method in theory.

The static factor model \eqref{eq.model} considered in this paper differs from the dynamic factor model \cite[]{forni2000generalized}, which allows $\by_t$ to depend on $\bx_t$ and its lagged values. Their approach adopts the frequency domain analysis based on the PCA for spectral density matrices. See also \cite{barigozzi2024dynamic} and the references therein.

Without the loss of generality, we assume
that the columns of $\bA$ are orthonormal, i.e. $\bA^{\top}\bA=\bI_{r_0},$ as $(\bA, \bx_t)$ in \eqref{eq.model} can be replaced by $(\widecheck{\bA}, \bV \bx_t)$, where, e.g., $\bA = \widecheck{\bA} \bV$ is the QR decomposition of $\bA$. Even with this orthogonality constraint, $\bA$ can still be replaced by $\bA \bU$ in \eqref{eq.model} for any $r_0 \times r_0$ orthogonal matrix $\bU$. However the linear space spanned by the columns of $\bA$, denoted by ${\cal C}(\bA)$, can be uniquely determined by \eqref{eq.model} when $\be_t$ is white noise and none of the linear combinations of $\bx_t$ is white noise. Then the standard autocovariance-based method is based on the fact that ${\cal C}(\bA)$ can be spanned by the $r_0$ leading eigenvectors, i.e., they correspond to the $r_0$ largest eigenvalues, of $\bOmega_{y}(k)\bOmega_{y}(k)^{\top}$ for any $k\ne 0$.
Hence the columns of an estimator for ${\bA}$ can be taken as the $r_0$ leading orthonormal eigenvectors of $\widehat{\bM}_1=\sum_{k=1}^{m}\widehat\bOmega_{y}(k)\widehat\bOmega_{y}(k)^{\top}.$
However, such method using the eigenvalue-ratio principle \citep{lam2012factor} often suffers from underestimating the number of factors in practice due to its less effectiveness in  separating the common and idiosyncratic components in finite samples, not mentioning the existence of the common factors with different strengths in practice. Alternative methods for estimating $r_0$ include the covariance-based eigenvalue-ratio principle \citep{ahn2013eigenvalue}, information criteria \citep{bai2002determining}, and testing-based approaches \citep{onatski2009testing,trapani2018randomized}. Despite these advances, determining the number of factors remains challenging in practice and may yield unsatisfactory results.

In this paper, we propose a novel weight-calibrated method by introducing a $p\times p$ non-negative definite weight matrix $\widehat\bW$ to enhance the separation, thereby improving the performance of the standard autocovariance-based method.
Specifically, we take the $r_0$ leading orthonormal eigenvectors of matrix
\begin{equation}
    \label{eq.hat_M}
    \widehat \bM = \sum_{k=1}^{m}\widehat\bOmega_{y}(k)\widehat\bW\,\widehat\bOmega_{y}(k)^{\top}
\end{equation}
as columns of the estimated loading matrix. 
By casting this weight-calibrated autocovariance-based estimation in a reduced-rank autoregression framework, 
the weight matrix turns out to be of the form
\begin{equation}
    \label{eq.matrix_W}
    \widehat{\bW}=\bQ\big(\bQ^{\top}\widehat{\bOmega}_{y}\bQ\big)^{-1}\bQ^{\top},
\end{equation}
where $\bQ$ is  a $p \times q$ projection matrix and its columns are the $q$ leading orthonormal eigenvectors of $\widehat\bOmega_y$, and $r_0< q\le \min(p,n)$. 
This weight matrix induces a rescaling of eigenvalues of $\widehat{\bOmega}_y(k)\widehat{\bOmega}_y(k)^{\top},$ often resulting in an enhanced separation between the common and idiosyncratic components, and, more precisely, the relative decrease from the $r_0$-th to the $(r_0+1)$-th largest eigenvalues of $\widehat{\bOmega}_{y}(k)\widehat{\bW}\widehat{\bOmega}_{y}(k)^{\top}$ is of higher order than that of $\widehat{\bOmega}_{y}(k)\widehat{\bOmega}_{y}(k)^{\top}.$
In practice when $r_0$ is unknown, we may choose $q$ to be sufficiently large so as to satisfy the constraint.
Additionally, we introduce a data-driven criterion for selecting $q.$ 
Our method is tailored to time series data with  serial dependence, but is not applicable to independent data.

The main contribution of our paper is four-fold.
Firstly, this paper represents the first effort in the literature to establish the connection between the autocovariance-based eigenanalysis and constrained least squares in a reduced-rank autoregression formulation. 
Our proposal involves a novel weight matrix to enhance the separation of common and idiosyncratic components, leading to the  improvement in estimation. The proposed weight-calibrated estimation method is of independent interest as it can also be applied to other factor models, such as those for matrix or tensor time series. See Section~\ref{sec.disscuss}.
Secondly, under model~\eqref{eq.model}, we investigate the asymptotic properties of the relevant estimated quantities using our weight-calibrated autocovariance-based method. In particular, we relax the commonly imposed white noise assumption in the autocovariance-based literature to allow for weak serial correlations in the idiosyncratic components.

Thirdly, this paper makes the first attempt in the literature to systematically compare the covariance-based and  autocovariance-based methods in the presence of the factors with different strengths,
highlighting their respective applicable scenarios through both theoretical analysis and simulation studies. Last but not least, by presenting asymptotic properties of the corresponding eigenvalues and eigenvectors that arise from the covariance-based, the standard autocovariance-based and our newly proposed weight-calibrated autocovariance-based methods, we demonstrate the theoretical superiority of our approach in effectively distinguishing strong factors from weak factors and idiosyncratic components, even when the two competing methods fail.

The remainder of the paper is organized as follows. In Section~\ref{sec.method}, we specify the weight matrix from a reduced-rank autoregression formulation. 
Section~\ref{sec.theory_basic} presents the asymptotic properties of the proposed estimators. 
Section~\ref{sec.theory_weak} conducts the theoretical analysis of the covariance-based, the standard autocovariance-based, and the weight-calibrated autocovariance-based methods in the presence of the factors with different strengths. In Sections~\ref{sec.sim} and \ref{sec.real}, we demonstrate the superior finite sample performance of our proposed method over the competitors through extensive simulations and the analysis of two real datasets, respectively. Section~\ref{sec.disscuss} discusses several future extensions. All technical proofs are relegated to the supplementary material.

{\it Notation}. For any matrix $\bB=(B_{ij})_{p \times q},$ we let $\Vert \bB\Vert_{\min}=\lambda_{\min}^{1/2}(\bB^{\top}\bB),\Vert \bB\Vert=\lambda_{\max}^{1/2}(\bB^{\top}\bB)$ and denote its Frobenius norm by $\Vert \bB\Vert_{\F}=(\sum_{i,j}B_{ij}^2)^{1/2}$. Let $\lambda_{i}(\bB)$ and $\sigma_{i}(\bB)$ be the $i$-th largest eigenvalue (if exists) and singular value of $\bB$, respectively. Let $\rank(\bB)$ be the rank of $\bB$. Let $\cC(\bB)=\cC(\{\bbb_i\}_{i=1}^q)$ be the space spanned by the columns of $\bB=(\bbb_1, \dots, \bbb_q).$ For a positive integer $m,$ write $[m]=\{1, \dots, m\}$ and denote by $\bI_m$ the identity matrix of size $m \times m.$ For $x,y \in {\mathbb R},$ we use $x \wedge y = \min(x,y)$, $x\vee y=\max(x,y)$, and $\lfloor x\rfloor$ as the floor function of $x$. For two positive sequences $\{a_n\}$ and $\{b_n\}$, we write $a_n\lesssim b_n$ or $a_n=O(b_n)$ or $b_n\gtrsim a_n$ if there exists a positive constant $c$ such that $a_n/b_n\le c$, and $a_n\ll b_n$ or $a_n=o(b_n)$ if $a_n/b_n\to0$. We write $a_n\asymp b_n$ if and only if $a_n\lesssim b_n$ and $a_n\gtrsim b_n$ hold simultaneously.

\section{Methodology}
\label{sec.method}

\subsection{Reduced-rank autoregression formulation}
\label{subsec.reduced}

To specify the weight matrix in \eqref{eq.matrix_W}, we cast the autocovariance-based method in a reduced-rank autoregression framework. This is motivated by the equivalence between the standard PCA and a reduced-rank regression \cite[]{reinsel2022multivariate}, which also involves a weight matrix implicitly. Specifically, consider regressing a $p$-vector $\by_t$ on itself with a rank-reduced coefficient matrix: $\by_t = \bL \by_t + \be_t$ for $t\in [n]$, where $\bL$ is  $p\times p$ with rank $r_0<p$. Let $p\le n$. 
Denoted by $\widehat \bL$ the constrained least squares estimator for this reduced-rank regression. Then $\cC(\widehat\bL)$ is spanned by the $r_0$ leading eigenvectors of $\bY^{\top}\bY(\bY^{\top}\bY)^{-1}\bY^{\top}\bY:=n\widehat{\bOmega}_y\widetilde{\bW}\widehat{\bOmega}_y^{\top}=n\widehat{\bOmega}_y,$ where $\bY=(\by_1,\dots,\by_n)^{\top}\in\eR^{n\times p}.$ Note the sandwiched form with the weight matrix $\widetilde{\bW}=\widehat{\bOmega}_y^{-1}$. This effectively amounts to regressing $\by_t$ on its $r_0$ leading principal components.

Inspired by this observation, we aim to implement the autocovariance-based eigenanalysis based on \eqref{eq.matrix_W} through regressing $\by_t$ on its lagged values $\by_{t-k}$, for $k \ge 1$, using latent factor $\bx_t$ as an intermediary. To
tackle high-dimensionality including the case $p>n,$ we project $\by_t$ to $\widetilde{\by}_t=\bQ^{\top}\by_t,$ which confines our analysis to a subspace that retains the most information of common components. 
Here $\bQ$ is a $p \times q$ projection matrix and its columns are the $q$ leading orthonormal eigenvectors of $\widehat\bOmega_y$, and $q$, satisfying $r_0 < q\le p\wedge n$, is a tuning parameter. Now by assuming the latent factor is of the form $\bx_t=\bH_k^{\top}\widetilde{\by}_{t-k}+\widetilde{\be}_{tk}$, where $\bH_k$ is  $q\times r_0$ and full-ranked, model \eqref{eq.model} becomes a reduced-rank autoregression
\begin{equation}
    \label{eq.rrr}
    \by_t=\bA\bH_k^{\top}\widetilde{\by}_{t-k}+\be_{tk},~~t=k+1,\dots,n,
\end{equation}
where $\be_{tk}=\bA\widetilde{\be}_{tk}+\be_t$.
The estimators for both $\bA$ and $\bH_k$ are then the solution of the following constrained least squares problem
\begin{equation}
    \label{eq.minimize}
    \min_{\bH_k,\,\bA^{\top}\bA=\bI_{r_0}}
    \sum_{t=k+1}^n \| \by_t - \bA \bH_k^{\top}\, \widetilde{\by}_{t-k}\|^2.
\end{equation}
According to Section~\ref{supsubsec.solution} of the supplementary material,  the columns of the resulting estimator for $\bA$ can be taken as the $r_0$ leading orthonormal eigenvectors of matrix 
$$
\widehat{\bOmega}_{y}(k)\bQ\big(\bQ^{\top}\widehat{\bOmega}_{y}\bQ\big)^{-1}\bQ^{\top}\widehat{\bOmega}_{y}(k)^{\top}
= \widehat{\bOmega}_{y}(k) \widehat{\bW} \, \widehat{\bOmega}_{y}(k)^{\top},
$$
where the weight matrix $\widehat{\bW}$ is defined in \eqref{eq.matrix_W}. To accumulate information across different lags, we propose to estimate the columns of $\bA$ by the $r_0$ leading orthonormal eigenvectors of matrix $\widehat \bM$ defined in \eqref{eq.hat_M}.

The estimation procedure developed for model \eqref{eq.model} assumes that the number of factors $r_0$ is known or can be correctly identified. In practice, $r_0$ is unknown and needs to be estimated.
Let $\hat{\lambda}_{k1} \ge \cdots \ge \hat{\lambda}_{kq}$ be the eigenvalues of $\widehat{\bOmega}_y(k)\widehat{\bW}\widehat{\bOmega}_y(k)^{\top}$ for $k\in[m]$.
Motivated by \cite{zhang2024factor}, we then estimate $r_0$ by maximizing the ratios of the adjacent cumulative weighted eigenvalues as 
\begin{equation}
    \label{eq.deter_r}
    \hat{r}_{0}=\arg\max_{j\in[q-1]}R_{j},~~{\rm with~~ }R_{j}=\frac{\sum_{k=1}^{m}(1-k/n)\hat{\lambda}_{kj}+\vartheta_{n}}{\sum_{k=1}^{m}(1-k/n)\hat{\lambda}_{k,j+1}+\vartheta_{n}},
\end{equation}
where heavier weights are placed on eigenvalues with smaller $k$. Here $\vartheta_n>0$ provides a lower bound correction to $\hat{\lambda}_{kj}$ for $j=r_0+1,\dots,q$ and satisfies the conditions in Proposition~\ref{propos.est_r}  below. In practice, we may set $\vartheta_n=0.1n^{-1}p.$ It is worth noting that the standard method based on the ratios of the adjacent eigenvalues of $\widehat{\bM}$ \cite[]{lam2012factor} may suffer from inconsistent estimation, as discussed in \cite{zhang2024factor}. We finally take the $\hat{r}_{0}$ leading eigenvectors of $\widehat{\bM}$ as the columns of our loading matrix estimator $\widehat{\bA}.$ 

\begin{remark}
\label{rmk.calibrating}
(i) We provide an explanation for why the weight matrix $\widehat{\bW}$ can improve the estimation performance for model~\eqref{eq.model}. Since $\widehat{\bW}=\bQ(\bQ^{\top}\widehat{\bOmega}_{y}\bQ)^{-1}\bQ^{\top}$ acts as a rank-$q$ pseudo-inverse of $\widehat{\bOmega}_{y}$, it
induces a rescaling of eigenvalues of $\widehat{\bOmega}_y(k)\widehat{\bOmega}_y(k)^{\top}.$ This often enhances the separation between the common and idiosyncratic components, which in turn improves the performance of the ratio-based method in \eqref{eq.deter_r}, as empirically evidenced by simulation results in Section~\ref{sim.bfm} below. Detailed intuitive illustrations are presented in Section~\ref{supsubsec.calibrating} of the supplementary material. \\
(ii) The weight matrix can also enhance the capability to distinguish strong factors from weak factors. Consider the model $\by_t=\bA\bx_t+\bB\bz_t+\be_t,$ where the strong factors $\bx_t\in\eR^{r_0}$ and the weak factors $\bz_t\in\eR^{r_1}$ are asymptotically identifiable through the difference in the respective strengths. 
Let $\hat{\mu}_{k1} \ge \cdots \ge \hat{\mu}_{kp}$ be the
eigenvalues of $\widehat{\bOmega}_y(k)\widehat{\bOmega}_y(k)^{\top}$,  and $\hat{\lambda}_{k1} \ge \cdots \ge \hat{\lambda}_{kq}$ be the eigenvalues of $\widehat{\bOmega}_y(k)\widehat{\bW}\widehat{\bOmega}_y(k)^{\top}$. When the strong and weak factors exhibit comparable cross-sectional correlations, it can be shown that $(\hat{\mu}_{kr_0}/\hat{\mu}_{k,r_0+1})/(\hat{\mu}_{k,r_0+r_1}/\hat{\mu}_{k,r_0+r_1+1})=o_p\big\{(\hat{\lambda}_{kr_0}/\hat{\lambda}_{k,r_0+1})/(\hat{\lambda}_{k,r_0+r_1}/\hat{\lambda}_{k,r_0+r_1+1})\big\}$ under mild conditions. This implies that our proposed method can distinguish $\bx_t$ from $\bz_t$ more effectively. A rigorous theoretical justification is provided in Section~\ref{sec.theory_weak} below.
\end{remark}

\begin{remark}
    \label{rmk.limitation}
    Our proposed method builds on the autocovariance-based framework of \cite{lam2012factor}, which is tailored to time series data with sufficiently strong serial dependence.
    However, when data are independent, the population counterpart of matrix $\widehat{\bM}$ becomes degenerate, and the proposed method is not applicable. In practice, one may first conduct a ``pre-test'' for the serial dependence (e.g., based on sample autocorrelations) before applying the proposed method.
\end{remark}

\subsection{Determining the tuning parameter}
\label{subsec.IC}
The practical implementation requires choosing a suitable value for the tuning parameter $q$ (i.e., the number of columns in $\bQ$). 
Intuitively, if $q$ is too small, the bias increases because the $q$ leading eigenvectors of $\widehat{\bOmega}_y$ may fail to preserve the most information of common components.
Conversely, if $q$ is too large, the variance increases due to a higher proportion of invalid eigenvectors being included. The bias-variance tradeoff underscores the necessity of optimally selecting $q$. Inspired by \cite{wang2011consistent} and \cite{fan2013tuning}, we propose a generalized Bayesian information criterion (BIC):
\begin{equation}
    \label{eq.bic}
    \bic_k(q)=pn\log L_k(q)+Cd_k(q)\log(pn),~~k\in[m],
\end{equation}
where $L_k(q)=(pn)^{-1}\sum_{t=k+1}^n\|\by_t-\widehat{\bA}_k \widehat{\bH}_k^{\top}\, \widetilde{\by}_{t-k}\|^2$
is the sum of squared residuals (divided by $pn$) from the regression in \eqref{eq.rrr}, $(\widehat{\bH}_k,\widehat{\bA}_k)$ is the solution of the constrained minimization problem in \eqref{eq.minimize} using $\hat r_0,$ $\hat r_0$ is obtained by our proposed ratio-based estimator in \eqref{eq.deter_r}, and $C>0$ is some constant. 
Additionally, $d_k(q)=(p+q)\hat{r}_0-\hat{r}_0(\hat{r}_0+1)/2$ represents the corresponding degrees of freedom, calculated as the total number of parameters minus the number of constraints in \eqref{eq.minimize}. 
Then, the optimal value of $q$ can be determined as:
\begin{equation}
    \label{eq.deter_q}
    \hat{q}=\arg\min_{\bar{r}_0<q\le q_0}\sum_{k=1}^{m}\bic_k(q),
\end{equation}
where $q_0$ is a prespecified positive integer that does not need to be very large for computational efficiency, and $\bar{r}_0$ is obtained by \eqref{eq.deter_r} using $q_0.$ In our empirical analysis in Sections~\ref{sec.sim} and \ref{sec.real}, we set $C=0.2$ and $q_0=15,$ which consistently yield good finite-sample performance. We also conduct additional simulations in Section~\ref{supsubsec.sensitivity} of the supplementary material to show the robustness of our proposed calibrating weight to the choice of $q$ in improving the performance of the standard autocovariance-based method.

\section{Theoretical results with uniform factor strength}
\label{sec.theory_basic}
In this section, we investigate the asymptotic properties of the relevant estimated quantities under the condition that all the factors in model~\eqref{eq.model} are of the same strength. 
Before presenting the theoretical results, we impose some regularity conditions.

\begin{condition}
    \label{cond.A}
    $\bA^{\top}\bA=\bI_{r_0}$.
\end{condition}

\begin{condition}
    \label{cond.x}
    (i) $\{\bx_t\}_{t\in\eZ}$ and $\{\be_t\}_{t\in\eZ}$ are uncorrelated; (ii) $\{\by_t\}_{t\in\eZ}$ and $\{\bx_t\}_{t\in\eZ}$ are strictly stationary with finite fourth moments; (iii) $\{\bx_t\}_{t\in\eZ}$ is $\psi$-mixing with the mixing coefficients satisfying $\sum_{t \ge 1}t\psi(t)^{1/2}<\infty$.
\end{condition}

\begin{condition}
    \label{cond.cov}
    (i) There exists some constant $\delta_0 \in (0,1]$ such that $\|\bOmega_{x}\|\asymp \|\bOmega_{x}\|_{\min} \asymp p^{\delta_0}$; (ii) $n=O(p)$ and $p^{1-\delta_0}=o(n)$.
\end{condition}

\begin{condition}
    \label{cond.E}
    For model~\eqref{eq.model}, let $\bE=(\be_1,\dots,\be_n)^{\top}=\bG\bGamma \bR,$ where there exists some constant $c \in (0,1]$ such that $\bG \in {\mathbb R}^{n \times n}$ with $0<\sigma_{\lfloor cn\rfloor}(\bG) \le \cdots \le \sigma_{1}(\bG)<\infty$ and $\bR \in {\mathbb R}^{p \times p}$ with $0<\sigma_{\lfloor cp\rfloor}(\bR) \le \cdots \le \sigma_{1}(\bR)<\infty.$ Moreover, $\bGamma=(\Gamma_{tj})_{n\times p},$ whose entries are i.i.d. across $t\in[n]$ and $j\in[p]$ with zero-mean, unit-variance and finite fourth moment.
\end{condition}

\begin{condition}
    \label{cond.auto}
    $\|\bOmega_{x}(k)\|\asymp \|\bOmega_{x}(k)\|_{\min} \asymp p^{\delta_{0}}$ for $k\in[m]$, where $\delta_{0}$ is specified in Condition~\ref{cond.cov}.
\end{condition}

\begin{remark}
    \label{rmk.cond}
    Condition~\ref{cond.A} can always be satisfied as discussed in Section~\ref{sec.intro}. 
    Condition~\ref{cond.x}(i) is conventional in the factor modeling literature \citep{fan2013large,wang2019,chen2022factor}. However, it can be relaxed to allow weak correlations between $\{\bx_t\}$ and $\{\be_t\},$ which do not affect the asymptotic results. For simplicity, we assume  uncorrelatedness to facilitate the presentation.
    Conditions~\ref{cond.x}(ii) and (iii) are also standard in the literature; see, e.g., \cite{lam2012factor} and \cite{zhang2024factor}.
    The parameter $\delta_0$ in Conditions~\ref{cond.cov}(i) and \ref{cond.auto} can be regarded as the strength of cross-sectional contemporaneous and serial correlations in model~\eqref{eq.model}, with larger values yielding stronger factors.
    When $\delta_0=1,$ Condition~\ref{cond.cov}(i) corresponds to the pervasiveness assumption in \cite{fan2013large}.
    Condition~\ref{cond.cov}(ii) is standard in the literature, see e.g., Theorem 2 of \cite{lam2012factor} and \cite{zhang2024factor}. If it is replaced by $p^{1-\delta_0}=o(n^{1/2})$, as in Theorem 1 of \cite{lam2012factor} and \cite{wang2019}, the convergence rates established in Theorem~\ref{theoforbasic}(ii) and Proposition~\ref{propos.U_A} will be slower.
    Condition~\ref{cond.E} implies that $\tr (\bE\bE^{\top}) \asymp np $, and the singular values of $\bE$ satisfy $\sigma_1(\bE) \asymp\sigma_{\lfloor c(p\wedge n)\rfloor}(\bE) \asymp n^{1/2}+p^{1/2}$ with probability tending to 1. These restrictions on idiosyncratic errors $\{\be_t\}$ are weaker than the commonly-imposed white noise assumption in the autocovariance-based factor modeling literature \cite[]{lam2012factor,wang2019,chen2022factor}. Condition~\ref{cond.auto} indicates comparable serial correlations in $\{\bx_t\}$, allowing the autocovariance matrices of $\by_t$ to still retain the useful information of $\cC(\bA)$; see \cite{lam2012factor,zhang2024factor}. 
    While the same asymptotics can be achieved if Condition~\ref{cond.auto} is relaxed to hold for at least one lag $k \in [m],$ we adopt Condition~\ref{cond.auto} to conduct the asymptotic analysis for each $k \in [m]$ and to facilitate the theoretical comparison of competing methods.
\end{remark}

Let $\widehat{\bphi}_{k1}, \dots, \widehat{\bphi}_{kq}$ be the eigenvectors of $\widehat{\bOmega}_{y}(k)\widehat{\bW} \widehat{\bOmega}_{y}(k)^{\top}$ corresponding to the eigenvalues $\hat{\lambda}_{k1}\ge\cdots \ge\hat{\lambda}_{kq} \ge 0$, and $\widehat{\bPhi}_{k,A}=(\widehat{\bphi}_{k1},\dots,\widehat{\bphi}_{kr_0}) \in {\mathbb R}^{p \times r_0}.$ 
Note that $\cC(\widehat{\bPhi}_{k,A})$ serves as the $k$-th individual estimate of $\cC(\bA)$ for $k\in[m].$ The following theorem presents the asymptotic properties of the eigenvalues and eigenvectors of $\widehat{\bOmega}_{y}(k)\widehat{\bW} \widehat{\bOmega}_{y}(k)^{\top}$. 

\begin{theorem}
    \label{theoforbasic}
    Let Conditions~\ref{cond.A}--\ref{cond.auto} hold, and $r_0<q\le c(p\wedge n)$, where $c$ is specified in Condition~\ref{cond.E}. For each $k\in[m],$ the following assertions hold:\\
    (i)    $\hat{\lambda}_{k1}  \asymp \hat{\lambda}_{kr_0} \asymp p^{\delta_0}$
    with probability tending to 1, and
    $\hat{\lambda}_{k,r_0+1}=O_p(n^{-1}p);$\\
    (ii)    $\big\|\widehat{\bPhi}_{k,A}\widehat{\bPhi}_{k,A}^{\top}-\bA\bA^{\top}\big\|=O_p(n^{-1/2}p^{(1-\delta_{0})/2}).$
\end{theorem}

\begin{remark}
\label{rmk.thm1}
(i) Theorem~\ref{theoforbasic}(i) is essential in establishing the consistency of our ratio-based estimator $\hat r_0$ for $r_0$ proposed in \eqref{eq.deter_r}; see Proposition~\ref{propos.est_r} below. Specifically, combined with the lower bound correction $\vartheta_n\asymp n^{-1}p,$ it yields that, with probability tending to 1, $(\hat{\lambda}_{r_0}+\vartheta_n)/(\hat{\lambda}_{r_0+1}+\vartheta_n)\asymp np^{\delta_0-1},$ which goes to infinity by Condition~\ref{cond.cov}(ii). This demonstrates that our method can distinguish $\bx_t$ from $\be_t$ with high probability.\\ 
(ii) Theorem~\ref{theoforbasic}(ii) establishes the convergence rate of $\cC(\widehat{\bPhi}_{k,A})$ for each $k\in[m].$ It, together with Proposition~\ref{propos.est_r}, facilitates the asymptotic analysis of the estimate $\widehat{\bA}$ of $\bA,$ whose columns are the $\hat{r}_0$ leading eigenvectors of $\widehat{\bM}$ in \eqref{eq.hat_M}, aggregating dynamic information across different lags. 
As demonstrated in Proposition~\ref{propos.U_A} below, $\cC(\widehat{\bA})$ achieves the same rate as $\cC(\widehat{\bPhi}_{k,A})$. \\
(iii) Similar arguments as Remarks~\ref{rmk.thm1}(i) and (ii) apply to the theoretical results for model~\eqref{eq.model_2} with different factor strengths in Section~\ref{sec.theory_weak}. 
Therefore, we will only present the asymptotic  properties of the eigenvalues and eigenvectors of $\widehat\bOmega_y,$ $\widehat{\bOmega}_{y}(k)\widehat{\bOmega}_{y}(k)^{\top}$ and $\widehat{\bOmega}_{y}(k)\widehat{\bW} \widehat{\bOmega}_{y}(k)^{\top}$ for each $k \in [m]$ in Sections~\ref{subsec.cov}, \ref{subsec.auto} and \ref{subsec.wauto}, respectively. 
\end{remark}

\begin{remark}
\label{rmk.dis}
Since only the loading space $\cC(\bA)$ is uniquely determined, we measure the estimation error in terms of its uniquely defined projection matrix $\bA\bA^{\top}$ under the operator norm \citep{chen2022factor,zhang2024factor}, as presented in Theorem~\ref{theoforbasic}(ii). Furthermore, it can be shown that for each $k \in [m],$ there exists an orthogonal matrix $\bU_k\in\eR^{r_0\times r_0}$ such that $\Vert\widehat{\bPhi}_{k,A}-\bA\bU_k\Vert=O_p(n^{-1/2}p^{(1-\delta_{0})/2}),$ which is consistent with Theorem 2(i) of \cite{lam2011estimation}. 
To quantify the accuracy in estimating $\cC(\bA)$ directly, we can use the measure of the distance between two column spaces \cite[]{wang2019}. Specifically, for two column orthogonal matrices $\bK_1 \in {\mathbb R}^{p\times q_1}$ and $\bK_2  \in {\mathbb R}^{p\times q_2}$, define
    \begin{equation}
    \label{def.distance}
    \cD\big\{\cC(\bK_1),\cC(\bK_2)\big\}=\Big\{1-\frac{1}{q_1\vee q_2}\tr(\bK_1\bK_1^{\top}\bK_2\bK_2^{\top})\Big\}^{1/2},
    \end{equation}
ranging from 0 to 1. The distance is 0 if and only if $\cC(\bK_1)=\cC(\bK_2)$, and 1 if and only if $\bK_1$ and $\bK_2$ are orthogonal. By Theorem~\ref{theoforbasic}(ii), we can further establish that $\cD\{\cC(\widehat{\bPhi}_{k,A}),\cC(\bA)\}=O_p(n^{-1/2}p^{(1-\delta_{0})/2})$ for each $k \in [m],$ achieving the same rate as in Theorem~\ref{theoforbasic}(ii).
Similar arguments apply to our proposed estimator $\widehat{\bA},$ as shown in  Proposition~\ref{propos.U_A} below.
\end{remark}

\begin{proposition}
    \label{propos.est_r}
    Let the conditions for Theorem~\ref{theoforbasic} hold. Assume that $\vartheta_{n}p^{-\delta_0}\to0$ and $\vartheta_{n}\gtrsim n^{-1}p,$ where $\vartheta_n$ is specified in \eqref{eq.deter_r}. Then, $\eP(\hat{r}_{0}=r_0)\to1$ as $p,n\to\infty,$  
\end{proposition}

\begin{proposition}
    \label{propos.U_A}
    Under the conditions of Proposition~\ref{propos.est_r}, 
    $\big\|\widehat{\bA}\widehat{\bA}^{\top}-\bA\bA^{\top}\big\|=O_p(n^{-1/2}p^{(1-\delta_{0})/2})$ and 
    $\cD\big\{\cC(\widehat{\bA}),\cC(\bA)\big\}=O_p(n^{-1/2}p^{(1-\delta_{0})/2}).$
\end{proposition}

\section{Theoretical results with different factor strengths}
\label{sec.theory_weak}
In this section, we consider a more common scenario in practice that the factors are of different strengths. 
To highlight the added complication in comparison to the case of uniform factor strength presented in Section~\ref{sec.theory_basic}, we consider the case that the components of $\bx_t$ in model \eqref{eq.model} exhibit two different levels of strength.
More specifically, we rewrite \eqref{eq.model} as follows: 
\begin{eqnarray}
    \label{eq.model_2}
    \by_t=\bA\bx_t + \bB\bz_t + \be_t,~~t\in[n],
\end{eqnarray}
where $\bx_t$ is an $r_0$-vector of strong factors, $\bz_t$ is an $r_1$-vector of weak factors, $\bA \in {\mathbb R}^{p \times r_0}$ and $\bB \in {\mathbb R}^{p \times r_1}$ are the corresponding full-ranked factor loading matrices, and $\be_t$ is a $p$-vector of idiosyncratic errors. 
The numbers of strong factors $r_0$ and weak factors $r_1$ are unknown but assumed finite.
The strengths of $\bx_t$ and $\bz_t$ are reflected by constants $\delta_0$ and $\delta_1,$ $\delta_{1k}$'s respectively, which are specified in Conditions \ref{cond.cov_a}, \ref{cond.auto_a00}, \ref{cond.auto_a} below.
Model~\eqref{eq.model_2} and its variants have been extensively studied in the literature; see, e.g., \cite{pena2006nonstationary,chudik2011weak,lam2012factor,ando2017clustering,anatolyev2022factor} and \cite{zhang2024factor}. Note that \cite{zhang2024factor} studied~\eqref{eq.model_2} under a special assumption on the strengths of $\bx_t$ and $\bz_t$. In this section, we relax their factor strength assumption; see details in Remark \ref{rmk.cond_a} below.

In model~\eqref{eq.model_2}, our main focus is on the estimation of the number of strong factors $r_0$ and the corresponding loading matrix $\bA$ for two main reasons. Firstly, the strong factors $\{\bx_t\},$ which exhibit strong serial correlations, play a crucial role in forecasting the time series $\{\by_t\}.$ 
Once the strong factors are correctly identified, the weak factors and their loadings can then be estimated using a two-step method \cite[]{lam2012factor}.
Secondly, treating $\bu_t=\bB\bz_t+\be_t$ as the idiosyncratic errors, which may exhibit moderate cross-sectional and serial correlations due to an unobserved factor structure \cite[]{anatolyev2022factor}, can degenerate model~\eqref{eq.model_2} to model~\eqref{eq.model}, where $r_0$ and $\bA$ are the primary quantities to be estimated. 

To identify the scenarios in which $\bx_t$ can be distinguished from $\bz_t$ and $\be_t$, we investigate the theoretical properties of the covariance-based, the standard autocovariance-based, and our weight-calibrated autocovariance-based methods in Sections~\ref{subsec.cov}, \ref{subsec.auto} and \ref{subsec.wauto}, respectively.  Specifically, we present the corresponding asymptotic results for the eigenvalues and eigenvectors of $\widehat{\bOmega}_{y}$, $\widehat{\bOmega}_{y}(k)\widehat{\bOmega}_{y}(k)^{\top}$, and $\widehat{\bOmega}_{y}(k)\widehat{\bW}\widehat{\bOmega}_{y}(k)^{\top}$ for each $k \in [m].$ The number of strong factors $r_0$ can be determined using the ratios of adjacent eigenvalues for the covariance-based method \cite[]{ahn2013eigenvalue} and, analogously to \eqref{eq.deter_r}, the ratios of adjacent cumulative weighted eigenvalues for the autocovariance-based methods. The asymptotic properties of the corresponding estimators for $r_0$ and $\bA$ will be presented in a similar fashion to Theorem~\ref{theoforbasic}(i) and (ii).

\subsection{Covariance-based method}
\label{subsec.cov}

We impose the following regularity conditions for model~\eqref{eq.model_2}, which correspondingly generalize Conditions~\ref{cond.A}--\ref{cond.E} for model~\eqref{eq.model}  to incorporate the additional weak factors $\{\bz_t\}.$ 

\begin{condition}
    \label{cond.A_a}
    $\bA^{\top}\bA=\bI_{r_0},\bB^{\top}\bB=\bI_{r_1},$ $\|\bA^{\top}\bB\|<1$.
\end{condition}

\begin{condition}
    \label{cond.x_a}
    (i) $\{\bx_t\}_{t\in\eZ},\{\bz_t\}_{t\in\eZ}$, and $\{\be_t\}_{t\in\eZ}$ are mutually uncorrelated; (ii) $\{\by_t\}_{t\in\eZ},\{\bx_t\}_{t\in\eZ}$ and $\{\bz_t\}_{t\in\eZ}$ are strictly stationary with the finite fourth moments; (iii) both $\{\bx_t\}_{t\in\eZ}$ and $\{\bz_t\}_{t\in\eZ}$ are $\psi$-mixing with the mixing coefficients satisfying $\sum_{t \ge 1}t\psi(t)^{1/2}<\infty$.
\end{condition}

\begin{condition}
    \label{cond.cov_a}
    (i) There exist some constants $\delta_0$ and $\delta_1$ with $0<\delta_1 \le \delta_0 \le 1$ such that $\|\bOmega_{x}\|\asymp \|\bOmega_{x}\|_{\min} \asymp p^{\delta_0}$ and $\|\bOmega_{z}\|\asymp \|\bOmega_{z}\|_{\min} \asymp p^{\delta_1}$; (ii) $n=O(p)$ and $p^{1-\delta_1}=o(n)$.
\end{condition}

\begin{condition}
    \label{cond.E_a}
    For model~\eqref{eq.model_2}, let $\bE=(\be_1, \dots, \be_n)^{\top}= \bG\bGamma \bR,$ where there exists some constant $c' \in (0,1]$ such that $\bG \in {\mathbb R}^{n \times n}$ with $0<\sigma_{\lfloor c'n\rfloor}(\bG) \le \cdots \le \sigma_{1}(\bG)<\infty$ and $\bR \in {\mathbb R}^{p \times p}$ with $0<\sigma_{\lfloor c'p\rfloor}(\bR) \le \cdots \le \sigma_{1}(\bR)<\infty.$ Moreover, $\bGamma=(\Gamma_{tj})_{n\times p},$ whose entries are i.i.d. across $t\in[n]$ and $j\in[p]$ with zero-mean, unit-variance and finite fourth moment.
\end{condition}

Let $\hat{\theta}_{1}\ge\cdots \ge\hat{\theta}_{p} \ge 0$ be the eigenvalues of $\widehat{\bOmega}_{y}$ with the corresponding eigenvectors $\widehat{\bxi}_{1}, \dots,  \widehat{\bxi}_{p},$ and $\widehat{\bXi}_{A}=(\widehat{\bxi}_{1},\dots,\widehat{\bxi}_{r_0}) \in {\mathbb R}^{p \times r_0}.$ 
Denote $r=r_0+r_1.$
The following theorem establishes the asymptotic  properties of the eigenvalues and eigenvectors of $\widehat{\bOmega}_{y}$. 

\begin{theorem}\label{advsam1}
Let Conditions~\ref{cond.A_a}--\ref{cond.E_a} hold, and $\delta_0>\delta_1$. 
The following assertions hold:\\
(i) $\hat{\theta}_{1} \asymp \hat{\theta}_{r_0} \asymp p^{\delta_{0}},\hat{\theta}_{r_0+1} \asymp \hat{\theta}_{r} \asymp p^{\delta_{1}},$ and $\hat{\theta}_{r+1} \asymp \hat{\theta}_{\lfloor c'(p\wedge n)\rfloor} \asymp n^{-1}p$
with probability tending to 1;\\
(ii) $\big\|\widehat{\bXi}_{A}\widehat{\bXi}_{A}^{\top}-\bA\bA^{\top}\big\|=O_p(n^{-1/2}p^{(1-\delta_0)/2}+p^{\delta_1-\delta_0}).$
\end{theorem}

\begin{remark}
    \label{rmk.cov}
    By Theorem~\ref{advsam1}(i), we obtain that $\hat{\theta}_{r_0}/\hat{\theta}_{r_0+1}\asymp p^{\delta_0-\delta_1}$ and $\hat{\theta}_{r}/\hat{\theta}_{r+1}\asymp np^{\delta_1-1}$ with probability tending to 1, which together demonstrate that the covariance-based method can distinguish $\bx_t$ from $\bz_t$ with high probability as long as $\delta_0>\delta_1$ and $n=o( p^{1+\delta_0-2\delta_1})$. The condition $\delta_0>\delta_1$ can be relaxed to $\delta_0\ge\delta_1$ for two autocovariance-based methods, as presented in Theorems~\ref{advauto} and \ref{advour} below.
\end{remark}

\subsection{Standard autocovariance-based method}
\label{subsec.auto}

In addition to Conditions~\ref{cond.A_a}--\ref{cond.E_a}, we give the following regularity conditions for model~\eqref{eq.model_2}, which characterize the strengths of serial correlations for strong and weak factors.
\begin{condition}
    \label{cond.auto_a00}
    $\|\bOmega_{x}(k)\|\asymp \|\bOmega_{x}(k)\|_{\min} \asymp p^{\delta_{0}}$ for $k\in[m]$, where $\delta_{0}$ is specified in Condition~\ref{cond.cov_a}.
\end{condition}

\begin{condition}
    \label{cond.auto_a}
    (i) There exists some constant $\delta_{1k}$ with $0 < \delta_{1k}<\delta_{0}$ such that $\|\bOmega_{z}(k)\|\asymp \|\bOmega_{z}(k)\|_{\min} \asymp p^{\delta_{1k}}$ for $k\in[m]$, where $\delta_0$ is specified in Condition~\ref{cond.cov_a}; (ii) $p^{1+\delta_1-2\delta_{1k}}=o(n)$.
\end{condition}

\begin{remark}
    \label{rmk.cond_a}
    Condition~\ref{cond.auto_a00} serves as the counterpart of Condition~\ref{cond.auto} for strong factors in model~\eqref{eq.model_2}. 
    Condition~\ref{cond.auto_a}(i) implies that $\{\bz_t\}$ has weaker serial correlations compared to $\{\bx_t\}$, where $\delta_{1k}<\delta_0$ is satisfied  if $\delta_{1k}<\delta_1$ or $\delta_1<\delta_0$. A similar assumption is imposed in \cite{zhang2024factor} as $\delta_{1k}=\delta_1<\delta_0=1,$ which is a special case of Condition~\ref{cond.auto_a}(i).
    Condition~\ref{cond.auto_a}(ii) is required to establish the asymptotics of the standard and weight-calibrated autocovariance-based methods, which is stronger than Condition~\ref{cond.cov_a}(ii) if $\delta_{1k}<\delta_1$.
\end{remark}

Let $\widehat{\bpsi}_{k1}, \dots,  \widehat{\bpsi}_{kp}$ be the eigenvectors of $\widehat{\bOmega}_{y}(k)\widehat{\bOmega}_{y}(k)^{\top}$ corresponding to the eigenvalues $\hat{\mu}_{k1} \ge \cdots \ge \hat{\mu}_{kp} \ge 0$, and $\widehat{\bPsi}_{k,A}=(\widehat{\bpsi}_{k1},\dots,\widehat{\bpsi}_{kr_0}) \in {\mathbb R}^{p \times r_0}.$
The following theorem gives the asymptotic properties of the eigenvalues and eigenvectors of $\widehat{\bOmega}_{y}(k)\widehat{\bOmega}_{y}(k)^{\top}$.

\begin{theorem}\label{advauto}
Let Conditions~\ref{cond.A_a}--\ref{cond.auto_a} hold. For each $k\in[m]$, the following assertions hold:\\
    (i) $\hat{\mu}_{k1} \asymp \hat{\mu}_{kr_0} \asymp p^{2\delta_{0}},\hat{\mu}_{k,r_0+1} \asymp \hat{\mu}_{kr} \asymp p^{2\delta_{1k}},$ and $\hat{\mu}_{k,r+1} \asymp \hat{\mu}_{k,\lfloor c'(p\wedge n)\rfloor} \asymp n^{-2}p^{2}$
    with probability tending to 1;\\
    (ii) $\big\|\widehat{\bPsi}_{k,A}\widehat{\bPsi}_{k,A}^{\top}-\bA\bA^{\top}\big\|=O_p(n^{-1/2}p^{(1-\delta_{0})/2}+p^{\delta_{1k}-\delta_{0}}).$
\end{theorem}

\begin{remark}
    \label{rmk.auto}
    Theorem~\ref{advauto}(i) entails that $\hat{\mu}_{kr_0}/\hat{\mu}_{k,r_0+1}\asymp p^{2\delta_0-2\delta_{1k}}$ and $\hat{\mu}_{kr}/\hat{\mu}_{k,r+1}\asymp n^2p^{2\delta_{1k}-2}$ with probability tending to 1. These results show that the standard autocovariance-based method can distinguish $\bx_t$ from $\bz_t$ with high probability when $\delta_0>\delta_{1k}$ and $n=o(p^{1+\delta_0-2\delta_{1k}})$.
\end{remark}

\begin{remark}
    \label{rmk.MA_pro}
    Condition~\ref{cond.auto_a}(ii) may not always hold, especially when there are very weak serial correlations in $\{\bz_t\}$ (i.e., very small $\delta_{1k}$). For example, when $\{\bz_t\}$ is from a vector MA$(l)$ model, then for each $k>l$, $\bOmega_{z}(k)=\bzero,$ resulting in $\delta_{1k}=0$. Conditions~\ref{cond.cov_a}(ii) and \ref{cond.auto_a}(ii) cannot be satisfied simultaneously.
    To address this issue, we can replace Condition~\ref{cond.auto_a} with a new condition $\|\bOmega_{z}(k)\|=O(n^{-1/2}p^{(1+\delta_1)/2})$, which includes the case that $\{\bz_t\}$ follows a vector MA($l$) process with $k>l$, and establish the asymptotic results for $\widehat{\bOmega}_{y}(k)\widehat{\bOmega}_{y}(k)^{\top}$ in Corollary~\ref{coro.auto} below. With a lower bound correction $\vartheta_n\asymp n^{-1}p^{1+\delta_1}$, Corollary~\ref{coro.auto}(i) implies that  $(\hat{\mu}_{kr_0}+\vartheta_n)/(\hat{\mu}_{k,r_0+1}+\vartheta_n)$ goes to infinity with probability tending to 1, which ensures the consistency of our ratio-based estimator for $r_0$. Moreover, the convergence rate of $\widehat{\bPsi}_{k,A}$ in Corollary~\ref{coro.auto}(ii) is consistent with that in Theorem~\ref{theoforbasic}(ii) for model~\eqref{eq.model}. 
\end{remark}

\begin{corollary}
    \label{coro.auto}
    Let Conditions~\ref{cond.A_a}--\ref{cond.auto_a00} hold and $\|\bOmega_{z}(k)\|=O(n^{-1/2}p^{(1+\delta_1)/2})$ for each $k\in[m].$ \\
    (i) $\hat{\mu}_{k1} \asymp \hat{\mu}_{kr_0} \asymp p^{2\delta_{0}}$ with probability tending to 1, and $\hat{\mu}_{k,r_0+1}=O_p(n^{-1}p^{1+\delta_1});$\\ 
    (ii) $\big\|\widehat{\bPsi}_{k,A}\widehat{\bPsi}_{k,A}^{\top}-\bA\bA^{\top}\big\|=O_p(n^{-1/2}p^{(1-\delta_{0})/2}).$
\end{corollary}

\subsection{Weight-calibrated autocovariance-based method}
\label{subsec.wauto}
The following theorem presents the asymptotic theory for the eigenvalues and eigenvectors of $\widehat{\bOmega}_{y}(k)\widehat{\bW}\widehat{\bOmega}_{y}(k)^{\top}$ for model~\eqref{eq.model_2}, compared to Theorem~\ref{theoforbasic} for model~\eqref{eq.model}. 

\begin{theorem}\label{advour}
Let Conditions~\ref{cond.A_a}--\ref{cond.auto_a} hold, $q=O(n^{-1}p^{2-\delta_0})$, and $r<q\le c'(p\wedge n)$, where $c'$ is specified in Condition~\ref{cond.E_a}. For each $k\in[m]$, the following assertions hold:\\
    (i) $\hat{\lambda}_{k1} \asymp  \hat{\lambda}_{kr_0} \asymp p^{\delta_0}$ and $\hat{\lambda}_{k,r_0+1} \asymp \hat{\lambda}_{kr}  \asymp p^{2\delta_{1k}-\delta_1}$
    with probability tending to 1, and $\hat{\lambda}_{k,r+1}=O_p(n^{-1}p)$;\\
    (ii) $\big\|\widehat{\bPhi}_{k,A}\widehat{\bPhi}_{k,A}^{\top}-\bA\bA^{\top}\big\|=O_p(n^{-1/2}p^{(1-\delta_{0})/2}+p^{\delta_{1k}-(\delta_{0}+\delta_{1})/2}).$
\end{theorem}

The requirements on $q$ are satisfied for fixed $q$ under Condition~\ref{cond.cov_a}, and also allow $q$ to diverge to infinity if $\delta_0<1$ or $n=o(p)$.
Similar to the arguments in Remark~\ref{rmk.thm1}, Theorems~\ref{advauto}(i) and \ref{advour}(i) can guarantee the consistency of the corresponding ratio-based estimators for $r_0,$ which, combined with Theorems~\ref{advauto}(ii) and \ref{advour}(ii), can ensure that the respective loading space estimates achieve the same convergence rates. Therefore, we rely on the results in Theorems~\ref{advauto} and \ref{advour} to compare the theoretical properties of the standard and weight-calibrated autocovariance-based methods.

\begin{remark}
\label{remark.eigengap}
    Comparing the asymptotic properties of eigenvalues in Theorems~\ref{advauto}(i) and \ref{advour}(i) yields that $\hat{\lambda}_{kj}\asymp\hat{\mu}_{ki}^{1/2}$ for $j\in[r_0]$ with probability tending to 1 and $\hat{\lambda}_{kj}=O_p(\hat{\mu}_{kl}^{1/2})$ for $j>r.$ Moreover $\hat{\lambda}_{kj}=o_p(\hat{\mu}_{kj}^{1/2})$ for $j=r_0+1,\dots,r$ if $\delta_{1}>\delta_{1k}.$ 
    Therefore, compared to $\widehat{\bOmega}_{y}(k)\widehat{\bOmega}_{y}(k)^{\top},$ our weight-calibrated version $\widehat{\bOmega}_{y}(k)\widehat{\bW}\widehat{\bOmega}_{y}(k)^{\top}$ 
    results in a higher order of relative decrease from the $r_0$-th to the $(r_0+1)$-th eigenvalues and a smaller order of relative decrease from the $r$-th and $(r+1)$-th eigenvalues,
    which enhance its capability to separate $\bx_t$ from $\bz_t$.
    Specifically, by Theorem~\ref{advour}(i) with a lower bound correction $\vartheta_{n}\asymp n^{-1}p$, we have $(\hat{\lambda}_{kr_0}+\vartheta_{n})/(\hat{\lambda}_{k,r_0+1}+\vartheta_{n})\asymp p^{\delta_0+\delta_1-2\delta_{1k}}$ and $(\hat{\lambda}_{kr}+\vartheta_{n})/(\hat{\lambda}_{k,r+1}+\vartheta_{n})\asymp np^{2\delta_{1k}-\delta_1-1}$ with probability tending to 1. 
    If $n=o(p^{\delta_0+2\delta_1-4\delta_{1k}+1})$, then $(\hat{\lambda}_{kr}+\vartheta_{n})/(\hat{\lambda}_{k,r+1}+\vartheta_{n})=o_p\{(\hat{\lambda}_{kr_0}+\vartheta_{n})/(\hat{\lambda}_{k,r_0+1}+\vartheta_{n})\},$ which indicates that the proposed method can distinguish $\bx_t$ from $\bz_t$ with high probability. Combined with Remark~\ref{rmk.auto}, there exists an overlap, $p^{\delta_0-2\delta_{1k}+1}\ll n\ll p^{\delta_0+2\delta_1-4\delta_{1k}+1}$ when $\delta_1>\delta_{1k}$, such that our weight-calibrated autocovariance-based method can distinguish $\bx_t$ from $\bz_t$ with high probability, whereas the standard autocovariance-based method cannot.
\end{remark}

\begin{remark}
    For the loading space estimation of $\bA,$ despite the rate of $\widehat{\bPhi}_{k,A}$ being not faster than that of $\widehat{\bPsi}_{k,A}$ according to Theorems~\ref{advauto}(ii) and \ref{advour}(ii), when $\delta_0=\delta_1$ or Condition~\ref{cond.auto_a} is replaced by $\|\bOmega_{z}(k)\|=O(n^{-1/2}p^{(1+\delta_1)/2})$ (see Corollaries~\ref{coro.auto} and \ref{coro.new}), the standard and  weight-calibrated autocovariance-based methods achieve the same rate. Our simulations provide empirical evidence of improved estimation performance for $\bA$ using the proposed method.
\end{remark}

Following an analogy to Remark~\ref{rmk.MA_pro}, we present the following corollary for our method. 

\begin{corollary}
    \label{coro.new}
    Let Conditions~\ref{cond.A_a}--\ref{cond.auto_a00} hold, $q=O(n^{-1}p^{2-\delta_0})$, $r_0<q\le c'(p\wedge n)$, and $\|\bOmega_{z}(k)\|=O(n^{-1/2}p^{(1+\delta_1)/2})$ for each $k\in[m].$ \\
    (i) $\hat{\lambda}_{k1} \asymp  \hat{\lambda}_{kr_0} \asymp p^{\delta_0}$ with probability tending to 1, and $\hat{\lambda}_{k,r_0+1}=O_p(n^{-1}p);$ \\
    (ii)$
    \big\|\widehat{\bPhi}_{k,A}\widehat{\bPhi}_{k,A}^{\top}-\bA\bA^{\top}\big\|=O_p(n^{-1/2}p^{(1-\delta_{0})/2}).$
\end{corollary}

\subsection{Summary of theoretical results}
\label{subsec.summary}
We compare the effectiveness of three competing methods in distinguishing $\bx_t$ from $\bz_t$ and $\be_t$, as well as in estimating $\cC(\bA)$, based on the results in Theorems~\ref{advsam1}--\ref{advour} and Corollaries~\ref{coro.auto}--\ref{coro.new} as follows:

\begin{itemize}
\item Theorem~\ref{advsam1}(i) reveals that when $\delta_0=\delta_1$, the covariance-based method cannot separate $\bx_t$ from $\bz_t$ with high probability. In contrast, both the standard and weight-calibrated autocovariance-based methods can achieve this separation even when $\delta_0=\delta_1$, as demonstrated in Theorems~\ref{advauto}(i) and \ref{advour}(i). 

\item Theorems~\ref{advauto}(i) and \ref{advour}(i) respectively show that, with probability tending to 1, 
$$
\frac{\hat{\mu}_{kr_0}/\hat{\mu}_{k,r_0+1}}{\hat{\mu}_{kr}/\hat{\mu}_{k,r+1}}\asymp n^{-2}p^{2+2\delta_0-4\delta_{1k}},~~\frac{(\hat{\lambda}_{kr_0}+\vartheta_n)/(\hat{\lambda}_{k,r_0+1}+\vartheta_n)}{(\hat{\lambda}_{kr}+\vartheta_n)/(\hat{\lambda}_{k,r+1}+\vartheta_n)}\asymp n^{-1}p^{1+\delta_0+2\delta_1-4\delta_{1k}},
$$ 
where $\vartheta_n\asymp n^{-1}p$ provides a lower bound correction to $\hat{\lambda}_{kj}$ for $j>r$. When $p^{1+\delta_0-2\delta_1}=o(n)$, we have 
$\hat{\mu}_{kr_0}\hat{\mu}_{k,r+1}/\hat{\mu}_{k,r_0+1}\hat{\mu}_{kr}=o_p\{(\hat{\lambda}_{kr_0}+\vartheta_n)(\hat{\lambda}_{k,r+1}+\vartheta_n)/(\hat{\lambda}_{k,r_0+1}+\vartheta_n)(\hat{\lambda}_{kr}+\vartheta_n)\}.$ This demonstrates that our proposed method can distinguish $\bx_t$ from $\bz_t$ more effectively, highlighting its theoretical superiority over the standard autocovariance-based method. It is worth noting that, when $\delta_0=\delta_1$, the requirement $p^{1+\delta_0-2\delta_1}=o(n)$ is automatically satisfied under Condition~\ref{cond.cov_a}(ii).

\item In estimating $\cC(\bA),$ Theorems~\ref{advauto}(ii) and \ref{advour}(ii) show that our method does not converge faster than the standard autocovariance-based method. When $\delta_0=\delta_1$ or Condition~\ref{cond.auto_a} is replaced by $\|\bOmega_{z}(k)\|=O(n^{-1/2}p^{(1+\delta_1)/2})$ (see Corollaries~\ref{coro.auto} and \ref{coro.new}), both methods achieve the same convergence rate. Together with Theorem~\ref{advsam1}(ii), these results indicate that the rate of the covariance-based method is not uniformly faster or slower than those of the two autocovariance-based methods, but depends on the relationship among $\delta_0$, $\delta_1$, and $\delta_{1k}$. When $\delta_0=\delta_1$, however, the covariance-based method yields inconsistent estimation, whereas the two autocovariance-based methods remain consistent.
\end{itemize}

\section{Simulations}
\label{sec.sim}

\subsection{The model with uniform factor strength}
\label{sim.bfm}
Consider the factor model
\begin{equation}
    \label{eq.model_ex1}
    \by_t=\widetilde{\bA}\widetilde{\bx}_t+\be_t,~~t\in[n],
\end{equation}
where the entries of $\widetilde{\bA} \in \eR^{p \times r_0}$ are sampled independently from the uniform distribution on $(-p^{-(1-\delta_0)/2},p^{-(1-\delta_0)/2})$, following a similar sampling procedure as in \cite{wang2019}. Each component series of $\widetilde{\bx}_t$ follows an AR(1) process with independent ${\cal N}(0,1)$ innovations, and the autoregressive coefficient drawn from the uniform distribution on $\{(-0.95,-0.7)\cup (0.7,0.95)\}.$ The $p$-vector of the idiosyncratic error $\be_t$ is generated by a vector MA(1) model $\be_t=\bvarepsilon_{t}+\bPi \bvarepsilon_{t-1},$ where the coefficient matrix $\bPi=(0.6^{|i-j|})_{p\times p},$ and each component series of the $p$-vector $\bvarepsilon_t$ follows a MA(1) process with independent ${\cal N}(0, 1)$ innovations, and the moving-average coefficient drawn from the uniform distribution on $\{(-0.15,-0.05)\cup(0.05,0.15)\}$.
To incorporate model~\eqref{eq.model_ex1} into our framework, we let the columns of $\bA\in\eR^{p\times r_0}$ consist of the $r_0$ leading left singular vectors of $\widetilde{\bA}$, which ensures  $\cC(\bA)=\cC(\widetilde{\bA})$ and $\bA^{\top}\bA=\bI_{r_0},$ and results in the decomposition  $\widetilde{\bA}=\bA\bC.$ By letting $\bx_t=\bC\widetilde{\bx}_t$, it is easy to check that $\{\bx_t\}$ fulfills the factor strength requirements in Conditions~\ref{cond.cov}(i) and \ref{cond.auto}. Then, model~\eqref{eq.model_ex1} can be rewritten as $\by_t=\bA\bx_t+\be_t,$ aligning with model~\eqref{eq.model} with uniform factor strength. We set in model~\eqref{eq.model_ex1} $n=300,$ $p=50,100,200$ and $400,$ and $r_0=3$ with factor strength ranging from relatively weak to strong, specifically $\delta_0=0.75$ and $1.$  

We compare the proposed weight-calibrated autocovariance-based method, referred to as WAuto for simplicity, with two competing methods: the covariance-based method (denoted as Cov) and the standard autocovariance-based method (denoted as Auto). We evaluate the sample performance of these three methods in estimating the number of factors $r_0$ via the corresponding eigenvalue-ratio-based methods and the factor loading space $\cC(\bA),$ which are respectively measured by the relative frequency estimate for $\eP(\hat{r}_{0}=r_0)$ and the estimation error using the distance $\cD\big\{\cC(\widehat{\bA}), \cC(\bA)\big\},$ defined according to \eqref{def.distance}. 
Since our experimental results are insensitive to the choice of $m$ when implementing Auto and WAuto, we set $m=2.$ We ran each simulation 1000 times.

\begin{table}[ht]
    \small
    \centering
    \caption{The relative frequency estimate of $\eP(\hat{r}_{0}=r_0)$ and the average of $\hat{r}_{0}$ (in parentheses) for model~\eqref{eq.model_ex1} over 1000 simulation runs.}
    \label{tab.num}{
    \begin{tabular}{l|ccc|ccc}
            \hline
            & \multicolumn{3}{c|}{$\delta_0=0.75$} & \multicolumn{3}{c}{$\delta_0=1$} \\
          & Cov & Auto & WAuto & Cov & Auto & WAuto \\
          \hline
          $p=50$ & 0.098  & 0.647  & 0.899  & 0.651  & 0.961  & 0.998  \\
          & (1.635) & (2.430) & (2.855) & (2.432) & (2.937) & (2.997) \\
    $p=100$ & 0.282  & 0.730  & 0.942  & 0.948  & 0.991  & 1.000  \\
          & (1.812) & (2.549) & (2.922) & (2.910) & (2.984) & (3.000) \\
    $p=200$ & 0.605  & 0.791  & 0.961  & 0.991  & 0.991  & 0.999  \\
          & (2.341) & (2.649) & (2.940) & (2.983) & (2.983) & (2.999) \\
    $p=400$ & 0.851  & 0.873  & 0.965  & 1.000  & 0.999  & 1.000  \\
          & (2.749) & (2.792) & (2.939) & (3.000)   & (2.998) & (3.000) \\
        \hline
    \end{tabular}}
\end{table}

\begin{table}[ht]
    \small
    \centering
    \caption{The mean and standard deviation (in parentheses) of $\cD\big\{\cC(\widehat{\bA}),\cC(\bA)\big\}$ for model~\eqref{eq.model_ex1} over 1000 simulation runs.} 
    \label{tab.loading}{
    \begin{tabular}{l|ccc|ccc}
            \hline
             & \multicolumn{3}{c|}{$\delta_0=0.75$} & \multicolumn{3}{c}{$\delta_0=1$} \\
          & Cov & Auto & WAuto & Cov & Auto & WAuto \\
          \hline
              $p=50$ & 0.724  & 0.370  & 0.230  & 0.345  & 0.133  & 0.109  \\
          & (0.152) & (0.276) & (0.165) & (0.292) & (0.124) & (0.037) \\
    $p=100$ & 0.606  & 0.341  & 0.223  & 0.137  & 0.117  & 0.109  \\
          & (0.244) & (0.254) & (0.121) & (0.148) & (0.066) & (0.022) \\
    $p=200$ & 0.427  & 0.321  & 0.233  & 0.097  & 0.118  & 0.111  \\
          & (0.267) & (0.227) & (0.107) & (0.068) & (0.068) & (0.026) \\
    $p=400$ & 0.281  & 0.285  & 0.243  & 0.085  & 0.110  & 0.108  \\
          & (0.201) & (0.180) & (0.108) & (0.011) & (0.030) & (0.020) \\
        \hline
    \end{tabular}}
\end{table}

Table~\ref{tab.num} presents the average relative frequencies of $\hat r_0=r_0$ and the average $\hat r_0.$ Table~\ref{tab.loading} reports the numerical summaries of the corresponding estimation errors for $\cC(\bA).$
A few trends are observable from Tables~\ref{tab.num} and \ref{tab.loading}. 
Firstly, our proposed WAuto consistently outperforms Auto across all settings, significantly improving the estimation accuracy for both $r_0$ and $\cC(\bA),$ and thus demonstrating the effectiveness of calibrating weight in the autocovariance-based estimation. For example, when $\delta_0=0.75$ and $p=50$, the average relative frequency of $\hat r_0=r_0$ increases from 0.647 to 0.899, while the average $\cD\big\{\cC(\widehat{\bA}),\cC(\bA)\big\}$ decreases from 0.37 to 0.23. Such good performance of WAuto provides empirical evidence to validate the established asymptotic results in Propositions~\ref{propos.est_r} and \ref{propos.U_A}. 
Secondly, when the signal of common factors is weak relative to the idiosyncratics, corresponding to smaller values of $\delta_0$ and $p,$ the idiosyncratic errors exhibit cross-sectional correlations comparable to those of the factors in finite dimensions. This diminishes the capability of Cov to separate $\bx_t$ from $\be_t$, leading to inferior performance compared to Auto and WAuto, which achieve more effective separation. In contrast, under cases of relatively strong signals with larger values of $\delta_0$ and $p,$ Cov can often successfully identify the factors and result in the best performance, such as when $\delta_0=1$ and $p=200,400.$
Thirdly, we observe the phenomenon of the ``blessing of dimensionality'' when estimating $r_0$ in the sense of improved accuracy as $p$ increases, which is due to the increased information from the added components on the factors.

\subsection{The model with different factor strengths}
\label{sim.gfm}
Consider the factor model
\begin{equation}
    \label{eq.model_ex2}
    \by_t=\widetilde{\bA}\widetilde{\bx}_t+\widetilde{\bB}\widetilde{\bz}_t+\be_t,~~t\in[n],
\end{equation}
where the entries of $\widetilde{\bA}\in\eR^{p\times r_0}$ and $\widetilde{\bB}\in\eR^{p\times r_1}$ are sampled independently from the uniform distributions on $(-p^{-(1-\delta_0)/2},p^{-(1-\delta_0)/2})$ and $(-p^{-(1-\delta_1)/2},p^{-(1-\delta_1)/2})$, respectively. 
The component series of $\bx_t$ and $\bz_t$ respectively follow AR(1) and MA(1) models with independent $\cN(0,0.2^2)$ innovations, and the corresponding autoregressive and moving-average coefficients drawn from the uniform distribution on $(-0.95,-0.85)\cup(0.85,0.95)$. Each component series of $\be_t$ is generated independently from standard normal.
Similar to the formulation in Section~\ref{sim.bfm}, we can decompose $\widetilde{\bA}$ and $\widetilde{\bB}$ as $\widetilde{\bA}=\bA\bC_1$ and $\widetilde{\bB}=\bB\bC_2$ such that $\cC(\bA)=\cC(\widetilde{\bA}),\bA^{\top}\bA=\bI_{r_0},\cC(\bB)=\cC(\widetilde{\bB})$ and $\bB^{\top}\bB=\bI_{r_1}$. By letting $\bx_t=\bC_1\widetilde{\bx}_t$ and $\bz_t=\bC_2\widetilde{\bz}_t$, we can rewrite model~\eqref{eq.model_ex2} as $\by_t=\bA\bx_t+\bB\bz_t+\be_t,$ which takes the same form of factor model \eqref{eq.model_2}. Under our simulation setup, the strong and weak factor series $\{\bx_t\}$ and $\{\bz_t\}$ can be easily verified to satisfy the corresponding factor strength requirements in Conditions~\ref{cond.cov_a}(i), \ref{cond.auto_a00}, and \ref{cond.auto_a}(i). We generate $n=400$ serially dependent observations of $p=50,100,300$ and $500$ variables based on $r_0=3$ strong factors and $r_1=3$ weak factors with the respective strengths $(\delta_0,\delta_1)=(1,1)$ and $(\delta_0,\delta_1)=(1,0.85).$ Notable, under the first scenario $\delta_0=\delta_1,$ Cov fails to separate $\bx_t$ from $\bz_t$ according to Remark~\ref{rmk.cov}. Thus, the relaxed second scenario $\delta_0>\delta_1$ is also considered, where Cov is able to achieve such separation.

To better validate the theoretical results established in Section~\ref{sec.theory_weak}, in addition to three competing methods, we also compare the following methods: the standard autocovariance-based method using lag $1$ only (denoted as Auto1) and its weight-calibrated version (denoted as WAuto1), which estimate $\bA$ and $r_0$ by the corresponding eigenanalysis of $\widehat{\bOmega}_{y}(1)\widehat{\bOmega}_{y}(1)^{\top}$ and $\widehat{\bOmega}_{y}(1)\widehat{\bW}\widehat{\bOmega}_{y}(1)^{\top},$ and, similarly, two autocovariance-based methods using lag $2$ only (respectively denoted as Auto2 and WAuto2). For a fair comparison, we set $m=2$ for Auto and WAuto. All simulation results are based on 1000 replications.

Table~\ref{tab.num_new} reports the relative frequencies of $\hat{r}_{0}=r_0$ and the average $\hat{r}_{0}$. We observe several apparent patterns. 
Firstly, when strong and weak factors exhibit the same strength of cross-sectional correlations (i.e., $\delta_0=\delta_1=1$), Cov completely fails to estimate $r_0$ especially for large $p$, often misidentifying 6 strong factors. This aligns with Theorem~\ref{advsam1}(i). Additionally, Auto1 also fails to distinguish $\bx_t$ from $\bz_t$ effectively, whereas Auto2 performs very well, as guaranteed by Theorem~\ref{advauto}(i). This distinct performance arises from $\delta_{11}=1$ and $\delta_{12}=0$ for $\{\bz_t\}$, which follows a vector MA(1) process. 
Secondly, when $\delta_0=\delta_1=1$, our WAuto methods substantially improve the accuracy of estimating $r_0$ compared to Auto methods, demonstrating their undeniable advantage in distinguishing factors with the same strength of cross-sectional correlations but different strengths of serial correlations. Specifically, compared to Auto1, WAuto1 effectively corrects over-identified $\hat r_0$ to three strong factors with a high proportion. This confirms the result in Theorem~\ref{advour}(i) and the discussion in Remark~\ref{remark.eigengap}. Moreover, our weight-calibrated strategy also enhances the performance of Auto2.
For example, when $p=500,$ our method increases the relative frequency estimates of $\eP(\hat{r}_{0}=r_0)$ from 0.876 to 0.966 for Auto, from 0.532 to 0.913 for Auto1, and from 0.945 to 0.987 for Auto2. Meanwhile, the average $\hat{r}_{0}$ decreases from 3.914 to 2.972 for Auto1, and increases from 2.800 to 2.948 for Auto, and from 2.910 to 2.981 for Auto2. 
Thirdly, under the relaxed scenario where $\delta_0=1$ and $\delta_1=0.85,$ although Cov is able to distinguish $\bx_t$ from $\bz_t$ especially for large $p$, the proposed WAuto still outperforms the competitors across all settings, highlighting its broad empirical superiority.

\begin{table}[H]
    \small
    \centering
    \caption{ The relative frequency estimate of $\eP(\hat{r}_{0}=r_0)$ and the average of $\hat{r}_{0}$ (in parentheses) for model~\eqref{eq.model_ex2} over 1000 simulation runs.}
    \label{tab.num_new}{
    \begin{tabular}{l|ccccccc}
            \hline
          & Cov & Auto & WAuto & Auto1 & WAuto1 & Auto2 & WAuto2 \\
          \hline
          & \multicolumn{7}{c}{$\delta_0=1,~~~\delta_1=1$}\\
           $p=50$ & 0.322  & 0.664  & 0.839  & 0.634  & 0.793  & 0.700  & 0.823  \\
                 & (1.983) & (2.474) & (2.772) & (2.438) & (2.771) & (2.532) & (2.758) \\
           $p=100$ & 0.345  & 0.773  & 0.915  & 0.721  & 0.853  & 0.834  & 0.934  \\
                 & (2.703) & (2.664) & (2.881) & (2.580) & (2.797) & (2.752) & (2.910) \\
           $p=300$ & 0.073  & 0.830  & 0.952  & 0.687  & 0.892  & 0.918  & 0.974  \\
                 & (5.359) & (2.724) & (2.926) & (3.117) & (2.916) & (2.867) & (2.961) \\
           $p=500$ & 0.019  & 0.876  & 0.966  & 0.532  & 0.913  & 0.945  & 0.987  \\
                 & (5.844) & (2.800) & (2.948) & (3.914) & (2.972) & (2.910) & (2.981) \\
          \hline
          & \multicolumn{7}{c}{$\delta_0=1,~~~\delta_1=0.85$} \\
           $p=50$ & 0.580  & 0.795  & 0.924  & 0.821  & 0.913  & 0.752  & 0.835  \\
                 & (2.338) & (2.681) & (2.893) & (2.723) & (2.890) & (2.617) & (2.774) \\
           $p=100$ & 0.793  & 0.920  & 0.969  & 0.925  & 0.973  & 0.908  & 0.936  \\
                 & (2.669) & (2.881) & (2.958) & (2.890) & (2.964) & (2.864) & (2.910) \\
           $p=300$ & 0.945  & 0.977  & 0.990  & 0.976  & 0.992  & 0.980  & 0.984  \\
                 & (2.906) & (2.962) & (2.982) & (2.960) & (2.985) & (2.966) & (2.976) \\
           $p=500$ & 0.979  & 0.989  & 0.995  & 0.989  & 0.995  & 0.990  & 0.993  \\
                 & (2.964) & (2.982) & (2.992) & (2.982) & (2.992) & (2.984) & (2.990) \\
        \hline
    \end{tabular}}
\end{table}

\begin{table}[H]
    \small
    \centering
    \caption{ The mean and standard deviation (in parentheses) of $\cD\big\{\cC(\widehat{\bA}),\cC(\bA)\big\}$ for model~\eqref{eq.model_ex2} over 1000 simulation runs.}
    \label{tab.loading_new}{
    \begin{tabular}{l|ccccccc}
            \hline
          & Cov & Auto & WAuto & Auto1 & WAuto1 & Auto2 & WAuto2 \\
          \hline
        & \multicolumn{7}{c}{$\delta_0=1,~~~\delta_1=1$}\\
        $p=50$ & 0.583  & 0.402  & 0.319  & 0.425  & 0.351  & 0.409  & 0.357  \\
                 & (0.239) & (0.241) & (0.177) & (0.238) & (0.195) & (0.220) & (0.173) \\
           $p=100$ & 0.556  & 0.340  & 0.275  & 0.370  & 0.308  & 0.335  & 0.295  \\
                 & (0.253) & (0.208) & (0.134) & (0.224) & (0.170) & (0.177) & (0.114) \\
           $p=300$ & 0.687  & 0.312  & 0.252  & 0.384  & 0.281  & 0.288  & 0.266  \\
                 & (0.147) & (0.200) & (0.112) & (0.247) & (0.162) & (0.142) & (0.084) \\
           $p=500$ & 0.712  & 0.290  & 0.246  & 0.457  & 0.270  & 0.275  & 0.260  \\
                 & (0.077) & (0.174) & (0.095) & (0.263) & (0.146) & (0.120) & (0.064) \\
          \hline
    & \multicolumn{7}{c}{$\delta_0=1,~~~\delta_1=0.85$}\\
     $p=50$ & 0.438  & 0.335  & 0.275  & 0.332  & 0.291  & 0.382  & 0.349  \\
                 & (0.257) & (0.206) & (0.128) & (0.191) & (0.135) & (0.208) & (0.170) \\
           $p=100$ & 0.320  & 0.267  & 0.245  & 0.267  & 0.248  & 0.298  & 0.290  \\
                 & (0.217) & (0.136) & (0.085) & (0.131) & (0.080) & (0.139) & (0.115) \\
           $p=300$ & 0.230  & 0.234  & 0.230  & 0.233  & 0.227  & 0.255  & 0.257  \\
                 & (0.131) & (0.084) & (0.063) & (0.086) & (0.059) & (0.080) & (0.070) \\
           $p=500$ & 0.210  & 0.228  & 0.227  & 0.225  & 0.224  & 0.249  & 0.252  \\
                 & (0.084) & (0.061) & (0.046) & (0.061) & (0.046) & (0.059) & (0.052) \\
        \hline
    \end{tabular}}
\end{table}

Table~\ref{tab.loading_new} presents numerical summaries for $\cD\big\{\cC(\widehat{\bA}),\cC(\bA)\big\}$, revealing several notable trends. 
Firstly, when $\delta_0=\delta_1=1$, WAuto consistently achieves the most accurate recovery of the factor loading space for all settings we consider.
Secondly, when $\delta_0=1,\delta_1=0.85$ and $p=50,100,300,$ WAuto continues to perform the best, while for $p=500,$ Cov leads to the highest accuracy. This aligns with the findings in Section~\ref{sim.bfm} that Cov achieves the best performance when the signal is relatively strong. 
Thirdly, WAuto, which aggregates the dynamic information across different lags, is superior to both WAuto1 and WAuto2 in most cases, highlighting  the benefit of aggregating more autocovariance information. 
Note that $\widehat{\bA}$ is defined using $\hat{r}_0$. 
For comparison with the estimation results of the loading space with known $r_0$, see Section~\ref{supsubsec.true} of the supplementary material.

\section{Real data analysis}
\label{sec.real}

\subsection{Daily returns for S\&P 500 stocks}
\label{subsec.sp500}

The first dataset consists of the daily returns of S\&P 500 component stocks from January 2, 2002 to December 31, 2007 encompassing $n=1510$ trading days. We removed stocks that were not traded on every trading day during the period, resulting in a total of $p=375$ stocks in our analysis. A similar dataset was previously analyzed using the standard autocovariance-based method in \cite{lam2012factor}. 
Before further analysis, each component series is standardized to have zero mean and unit variance. Denote the $p$-vector of stock return series as $\{\by_t\}_{t=1}^{n}$.
To examine whether $\{\by_t\}$ has a finite fourth moment as in Condition~\ref{cond.x}(ii), we employ the Hill's tail-index estimator \cite[]{hsing1991tail}, which yields a 95\% confidence interval of $(3.847, 7.433)$. Therefore, the null hypothesis that $\{\by_t\}$ has a finite fourth moment cannot be rejected.
Since the returns exhibit very small lag-$k$ autocorrelations beyond $k=2,$ we use $m=2$ in our estimation. With $q$ selected by the generalized BIC in \eqref{eq.bic}, we employ the ratio-based method for both Auto and WAuto, yielding the estimated number of strong factors as $\tilde{r}_0=2$ and $\hat{r}_0=1,$ respectively. 

Figures~\ref{fig.eigen}(a) and (b) plot the ratios of cumulative weighted eigenvalues defined in the form of \eqref{eq.deter_r} for Auto and WAuto, respectively. It can be observed from Figure~\ref{fig.eigen}(b) that there exist two or three weak factors when implementing WAuto, as the 3rd and 4th ratios are significantly larger than the 2nd and 5th ratios. To further estimate the number of weak factors and the corresponding loading matrix, we apply the two-step estimation procedure introduced by \cite{lam2012factor}. 
Specifically for the two-step WAuto, let $\widehat{\bA} \in {\mathbb R}^{p \times \hat r_0}$ be the estimate of $\bA$ in \eqref{eq.model_2} and $\widehat\bx_t=\widehat{\bA}^{\top} \by_t$ be estimated strong factors based on the original time series $\{\by_t\}_{t \in [n]}$ in the first step. Define the remaining time series as $\by_t^*=\by_t - \widehat \bA \widehat \bx_t = \by_t-\widehat{\bA}\widehat{\bA}^{\top}\by_t$ for $t\in[n].$ In the second step, we apply the same estimation procedure to $\{\by_t^*\}_{t=1}^{n}$, thus obtaining the estimated number of weak factors as $\hat r_1$ and the corresponding estimated loading matrix as $\widehat{\bB}\in\eR^{p\times \hat r_1}.$ The resulting ratios of cumulative weighted eigenvalues in the second step for Auto and WAuto are displayed in Figures~\ref{fig.eigen}(c) and (d), respectively, indicating the corresponding estimated number of weak factors as $\tilde{r}_1=1$ and $\hat{r}_1=3.$

\begin{figure}[ht]
    \begin{center}
    \includegraphics[width=1\linewidth]{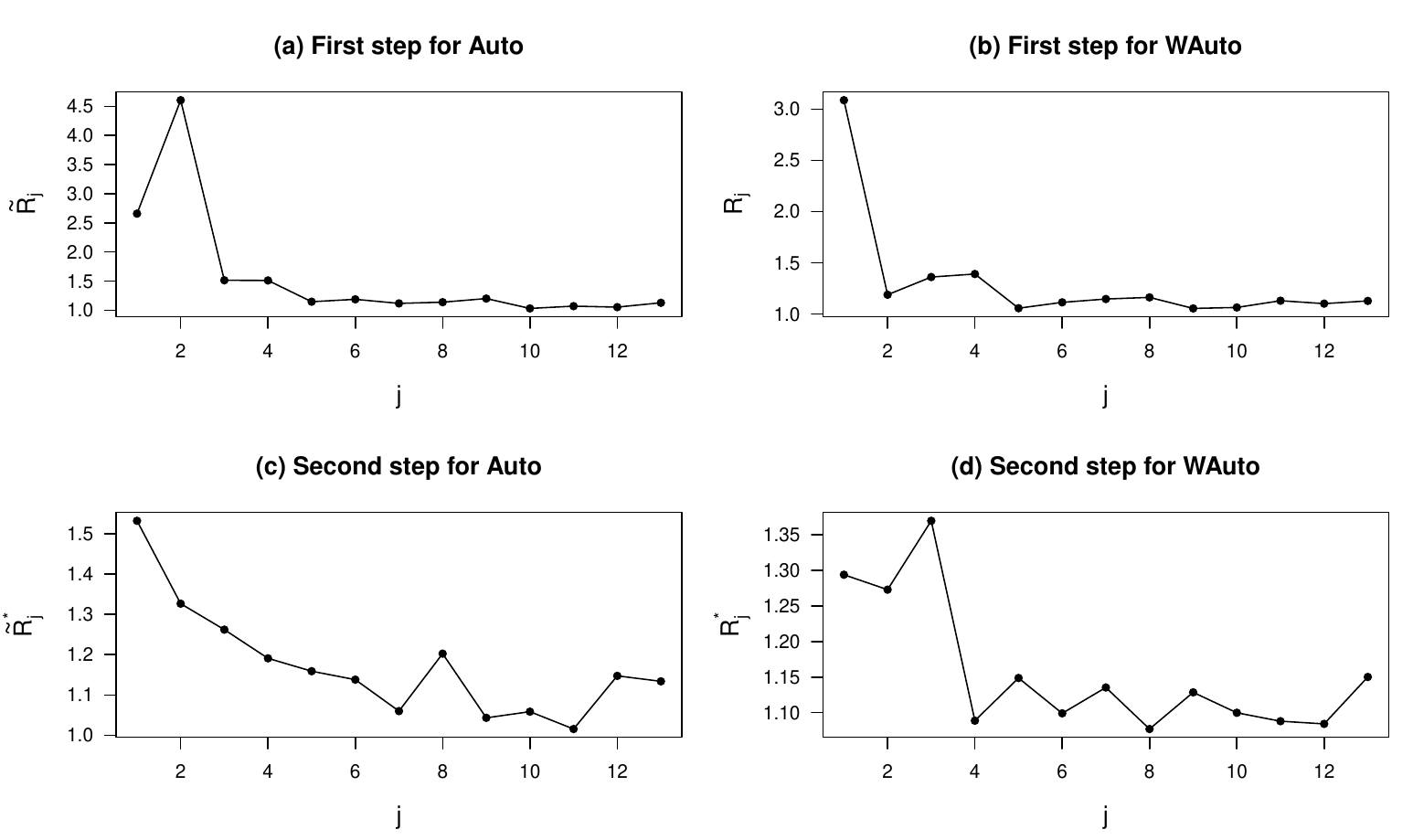}
    \end{center}
    \vspace{-2em}
    \caption{Plots of ratios of adjacent cumulative weighted eigenvalues in the first and second steps for Auto and WAuto. With such ratio proposed in \eqref{eq.deter_r} for WAuto, $\widetilde{R}_j,R_j^*$ and $\widetilde{R}_j^*$ can be defined analogously using $\{\hat{\mu}_{kj}\}$, $\{\hat{\lambda}_{kj}^*\}$ and $\{\hat{\mu}_{kj}^*\}$, respectively, where $\{\hat{\lambda}_{kj}^*\}$ and $\{\hat{\mu}_{kj}^*\}$ represent the corresponding eigenvalues obtained in the second step when applying WAuto or Auto to $\{\by_t^*\}_{t=1}^{n}$.}
    \label{fig.eigen}
\end{figure}

\begin{figure}[ht]
    \begin{center}
    \includegraphics[width=0.93\linewidth]{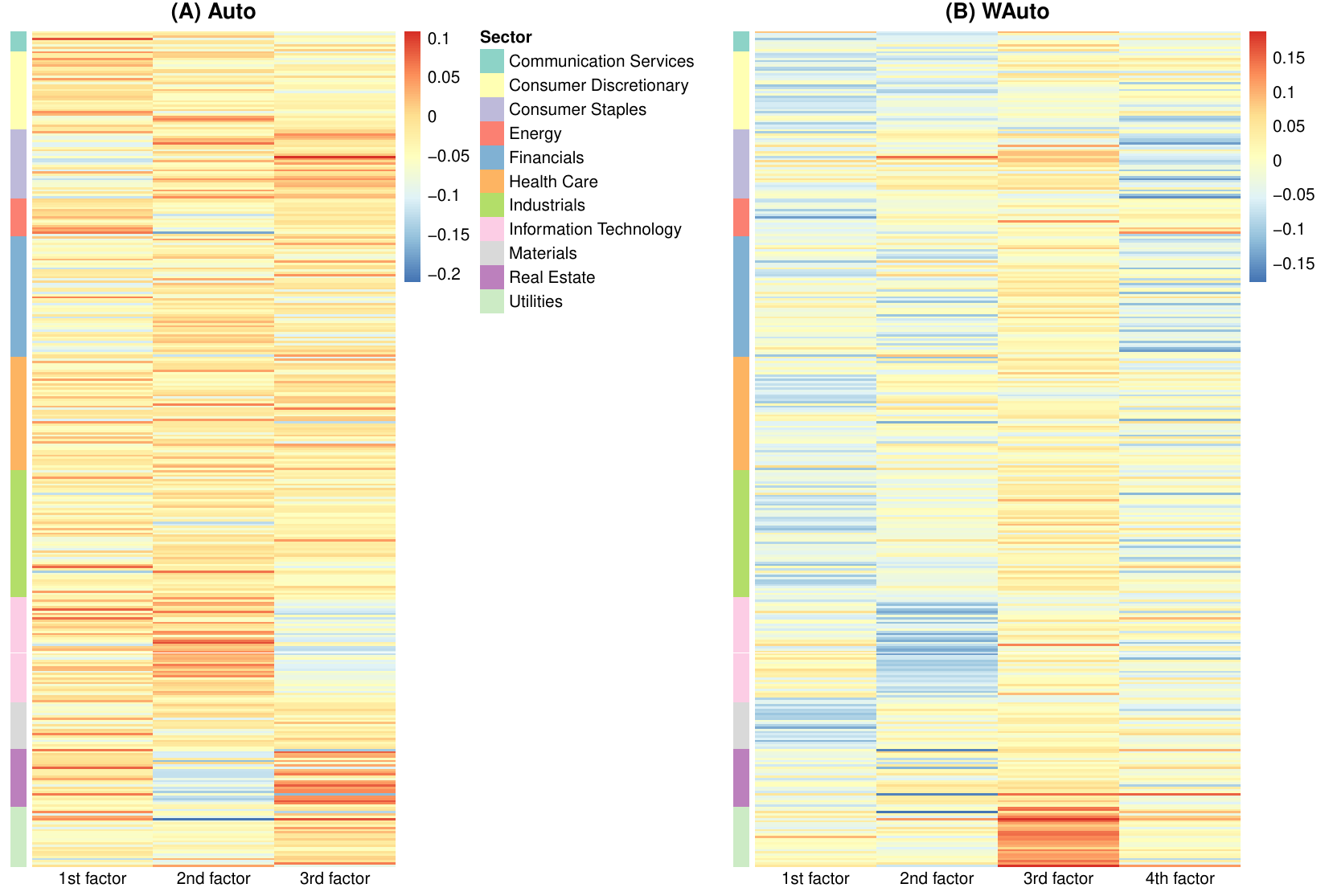}
    \end{center}
    \caption{ The left and right heatmaps display the varimax-rotated loadings of the first $(\tilde r_0 + \tilde r_1)$ and $(\hat r_0 + \hat r_1)$ identified factors for Auto and WAuto, respectively.}
    \label{fig.heat}
\end{figure}

Figure~\ref{fig.heat} displays the heatmaps of the loading matrices for the strong and weak factors identified by Auto and WAuto, using the varimax procedure to maximize the sum of the variances of the squared loadings. The varimax rotation enhances the color distinction in the heatmaps for better interpretation. The stocks are sorted according to their respective industrial sectors.
It is apparent that the four factors identified by WAuto serve as the main driving force for the dynamics of certain sectors. Specifically, the 1st factor mainly influences the dynamics of
Consumer Discretionary, Energy, Industrials, and Materials, as well as some stocks in Health
Care. The 2nd factor has a significantly impact on the dynamics of Information Technology, while Utilities is predominately loaded on the 3rd factor. Additionally, the dynamics of Financials are driven by both the 2nd and 4th factors, whereas the dynamics of Consumer Staples are influenced by both the 3rd and 4th factors. For Auto, although its three identified factors can also be viewed as the main driving force for the dynamics of specific sectors, its inferior intepretability is evident from the less clear sector-specific clustering structure of the factor loadings compared to WAuto, especially for the 1st factor. Moreover, the 2nd and 3rd factors have opposite impacts on the dynamics of Energy, Information Technology and Real Estate, which could potentially be merged into a single factor. In summary, WAuto provides more intepretable results for the model in presence of factors with different strengths.

\subsection{U.S. monthly macroeconomic data}
\label{subsec.fred}
Our second example considers incorporating the factor model framework into macroeconomic forecasting. The dataset, provided by \cite{mccracken2016fred}, contains $p=119$  U.S. monthly macroeconomic variables covering broad economic categories from January 1960 to September 2024 ($n=777$). Before further analysis, each macro variable is transformed as recommended in \cite{mccracken2016fred} to ensure stationarity and then standardized to have zero mean and unit variance. Denote the $p$-vector of macroeconomic variables series as $\{\by_t\}_{t\in[n]}$.
The corresponding 95\% confidence interval for the Hill's tail-index estimator is $(2.059, 4.554)$, indicating that the null hypothesis that $\{\by_t\}$ has a finite fourth moment cannot be rejected.
Following \cite{fan2013large} and \cite{wang2019}, we try $\hat{r}_0=1,3$ and 5 to check its effect on the forecasting performance of WAuto, and compare it with Cov and Auto. For each method, the $h$-step-ahead forecasting procedure consists of the following four steps:
\begin{itemize}
    \item[1.] For a given $\hat{r}_0$, apply the corresponding method to estimate model~\eqref{eq.model} based on past observations $\{\by_t\}_{t=1}^{T}$, thus obtaining the estimated factor loading matrix $\widehat{\bA}$.
    \item[2.] Compute $\hat{r}_0$ estimated factors by $\widehat{\bx}_t=\widehat{\bA}^{\top}\by_t$ for $t\in[T]$. 
    \item[3.] For each $j\in[\hat{r}_0],$ select the best ARMA model that fits $\{\widehat{x}_{tj}\}_{t=1}^{T}$ using the AIC criterion, and based on which achieves the $h$-step-ahead forecast for the factor series as $\widehat{x}_{T+h,j}$. 
    \item [4.] Convert back to obtain the $h$-step-ahead forecast for the original series as $\widehat{\by}_{T+h}=\widehat{\bA}\widehat\bx_{T+h}$, where $\widehat\bx_{T+h}=(\hat x_{T+h,1},\dots,\hat x_{T+h,\hat{r}_0})^{\top}$.
\end{itemize}

To evaluate the forecasting performance, we employ the expanding window approach. The data is divided into a training set and a test set, consisting of the first $n_1$ and the last $n_2=n-n_1$ observations, respectively. For a positive integer $h,$ we apply each fitting method to the training set, obtain an $h$-step-ahead forecast on the test data according to the fitted model, increase the training size by one, and repeat this procedure $n_2-h+1$ times. The $h$-step-ahead mean absolute forecasting error (MAFE) and mean squared forecasting error (MSFE) are given by ${\rm MAFE}(h)=\{p\times (n_2-h+1)\}^{-1}\sum_{t=n_1+h}^{n}\sum_{j=1}^{p}|\widehat{y}_{tj}-y_{tj}|$ and ${\rm MASE}(h)=\{p\times (n_2-h+1)\}^{-1}\sum_{t=n_1+h}^{n}\sum_{j=1}^{p}(\widehat{y}_{tj}-y_{tj})^2$. 
For comparison, we consider $n_1=727,n_2=50,h=1,2,3,$ and $\hat{r}_0=1,3,5$. 
Auto and WAuto are implemented by using $m=3$, as the corresponding autocorrelations of the cross-sectionally aggregated series are statistically significant. For WAuto, we select the optimal $q$ according to the generalized BIC \eqref{eq.deter_q}. 
The resulting MAFE and MSFE values are summarized in Table~\ref{tab.forecast}. It is obvious that WAuto consistently provides the highest forecasting accuracies among the three competitors across all settings. We perform the corresponding two-sample t-tests and find that, at the $0.05$ significance level, WAuto produces statistically significantly lower forecasting errors than both Cov and Auto.
Moreover, Auto outperforms Cov in many cases, particularly when a relatively small $\hat r_0$ or large $h$ is considered. Overall, our proposed WAuto effectively captures the dynamical information, thereby achieving the best forecasting performance.

\begin{table}[ht]
    \small
    \centering
    \caption{The MAFEs and MSFEs of the proposed WAuto and two competing methods. The lowest values are in bold font.} 
    \label{tab.forecast}{
    \begin{tabular}{ll|ccc|ccc}
            \hline
             & & \multicolumn{3}{c|}{MAFE} & \multicolumn{3}{c}{MSFE} \\
             & & Cov & Auto & WAuto & Cov & Auto & WAuto \\
          \hline
          \multirow{3}{*}{$h=1$}     & $\hat{r}_0=1$     & 0.715   & 0.689  & \textbf{0.685} & 1.303   & 1.177  & \textbf{1.167} \\
          & $\hat{r}_0=3$     & 0.723   & 0.716  & \textbf{0.693} & 1.318   & 1.325  & \textbf{1.204} \\
          & $\hat{r}_0=5$     & 0.712   & 0.717  & \textbf{0.683} & 1.297   & 1.342  & \textbf{1.206} \\
          \hline
          
    \multirow{3}{*}{$h=2$}     & $\hat{r}_0=1$     & 0.693   & 0.678  & \textbf{0.674} & 1.172   & 1.139  & \textbf{1.132} \\
          & $\hat{r}_0=3$     & 0.701  & 0.697  & \textbf{0.679} & 1.184   & 1.236  & \textbf{1.148} \\
          & $\hat{r}_0=5$     & 0.699   & 0.700  & \textbf{0.671} & 1.179  & 1.247  & \textbf{1.152} \\
          \hline
          
    \multirow{3}{*}{$h=3$}     & $\hat{r}_0=1$     & 0.696   & 0.679  & \textbf{0.677} & 1.186  & 1.143  & \textbf{1.140} \\
          & $\hat{r}_0=3$     & 0.701   & 0.678  & \textbf{0.672} & 1.192   & 1.138  & \textbf{1.130} \\
          & $\hat{r}_0=5$     & 0.699   & 0.680  & \textbf{0.666} & 1.189   & 1.143  & \textbf{1.138} \\
        \hline
    \end{tabular}}
\end{table}

\section{Discussion}
\label{sec.disscuss}
We identify several important directions for future research. The first topic involves developing inferential theory for our proposed weight-calibrated autocovariance-based estimation. The existing literature on the limiting distributions of relevant estimators for the factor models primarily focuses on the covariance-based method \citep{bai2003inferential}. However, even for the standard autocovariance-based method, the corresponding inferential theory is largely untouched. Motivated by the formulation of our proposal from a reduced-rank autoregression perspective, we may potentially leverage the existing inferential results in reduced-rank regression to investigate the limiting distributions of relevant estimated quantities. For some heuristic distributional analysis under strong assumptions, see Section~\ref{supsubsec.approximation} of the supplementary material.

The second topic considers developing the weight-calibrated autocovariance-based method to improve the estimation for high-dimensional matrix factor models. In the autocovariance-based estimation \citep[e.g.,][]{wang2019}, non-negative definite matrices can be constructed using respective weight matrices, and their eigenpairs can be used to estimate the numbers of factors and the factor loading spaces. Specifically, consider a stationary matrix-valued time series $\{\bY_t\}_{t\in\eZ}$ of size $p_1\times p_2$ satisfying the factor model $\bY_t=\bR\bX_t\bC^{\top}+\bE_t$ for $t\in [n],$ where the latent matrix-valued factor time series $\{\bX_t\}$ is of size $d_1\times d_2$.
For lag $k \in [m],$ define $\bOmega_{y,ij}^{(1)}(k)=\cov(\by_{t,\cdot i},\by_{t-k,\cdot j})$ for $i,j\in[p_2]$, and $\bOmega_{y,i'j'}^{(2)}(k)=\cov(\by_{t,i'\cdot},\by_{t-k,j'\cdot})$ for $i',j'\in[p_1]$, where $\by_{t,\cdot j}$ and $\by_{t,j'\cdot }$ are the $j$-th column and the $j'$-th row of $\bY_t,$ respectively. Let $\widehat{\bOmega}_{y,ij}^{(1)}(k)$ and $\widehat{\bOmega}_{y,i'j'}^{(2)}(k)$ be the corresponding sample estimates. Define 
$$
\widehat{\bM}^{(1)}=\sum_{k=1}^{m}\sum_{i=1}^{p_2}\sum_{j=1}^{p_2}\widehat{\bOmega}_{y,ij}^{(1)}(k)\widehat{\bW}_{j}^{(1)}\widehat{\bOmega}_{y,ij}^{(1)}(k)^{\top},~~{\rm and}
~~\widehat{\bM}^{(2)}=\sum_{k=1}^{m}\sum_{i'=1}^{p_1}\sum_{j'=1}^{p_1}\widehat{\bOmega}_{y,i'j'}^{(2)}(k)\widehat{\bW}_{j'}^{(2)}\widehat{\bOmega}_{y,i'j'}^{(2)}(k)^{\top},$$ 
where $\widehat{\bW}_{j}^{(1)}=\bQ_{1j}\big\{\bQ_{1j}^{\top}\widehat{\bOmega}_{y,jj}^{(1)}(0)\bQ_{1j}\big\}^{-1}\bQ_{1j}^{\top}$, and $\widehat{\bW}_{j'}^{(2)}=\bQ_{2j'}\big\{\bQ_{2j'}^{\top}\widehat{\bOmega}_{y,j'j'}^{(2)}(0)\bQ_{2j'}\big\}^{-1}\bQ_{2j'}^{\top}$. The expressions of $\widehat{\bM}^{(1)}$ and $\widehat{\bM}^{(2)}$ are derived in a reduced-rank autoregression formulation for matrix-valued time series; see Section~\ref{supsubsec.reduced} of the supplementary material. The row (or column) factor loading space $\cC(\bR)$ (or $\cC(\bC)$) can then be estimated using the $d_1$ (or $d_2$) leading eigenvectors of $\widehat{\bM}^{(1)}$ (or $\widehat{\bM}^{(2)}$). 
Analogous to the formulation in Section~\ref{subsec.reduced}, we can specify $\bQ_{1j}$ (or $\bQ_{2j'}$) as the $q_1$ (or $q_2$) leading eigenvectors of $\widehat{\bOmega}_{y,jj}^{(1)}(0)$ (or $\widehat{\bOmega}_{y,j'j'}^{(2)}(0)$). We also expect our weight-calibrated idea can also be applied to improve the autocovariance-based estimation of factor models for high-dimensional tensor time series \citep[e.g.,][]{chen2022factor}. 

The third topic concerns an iterative extension of our method to further improve the performance. At each iteration, the matrix $\widehat{\bM}=\sum_{k=1}^{m}\widehat\bOmega_{y}(k)\widehat\bW\,\widehat\bOmega_{y}(k)^{\top}$ obtained from the previous step is used to update the projection matrix $\bQ$, and then eigenanalysis is performed in the resulting subspace. Here, $\cC(\bQ)$ can be regarded as an ``overfitted'' candidate estimator of the loading space, and the iterative procedure aims to progressively refine its accuracy. The detailed algorithm and conducted simulations are provided in Section~\ref{supsubsec.iterative} of the supplementary material.
Lastly, the proposed calibrating weight is expected to have broader applications beyond the eigenvalue-ratio estimator. For example, it can be incorporated into the information criterion for estimating the number of factors. The detailed methodology and numerical evidence are provided in Section~\ref{supsubsec.IC} of the supplementary material.

\section*{Acknowledgments}
We are grateful to the editor, the associate editor and three referees for their insightful comments and suggestions, which have led to significant improvement of our paper.

\section*{Disclosure Statement}
The authors report there are no competing interests to declare.


\linespread{1}\selectfont
\bibliographystyle{dcu}
\bibliography{main}

\newpage
\linespread{1.7}\selectfont
\begin{center}
	{\noindent \bf \large Supplementary Material to ``Weight-calibrated estimation for factor models of high-dimensional time series''}\\
\end{center}
\begin{center}
	{\noindent Xinghao Qiao, Zihan Wang, Qiwei Yao, and Bo Zhang}
\end{center}
\bigskip

\setcounter{page}{1}
\setcounter{section}{0}
\renewcommand{\thesection}{\Alph{section}}
\renewcommand{\theHsection}{\Alph{section}}
\setcounter{lemma}{0}
\renewcommand{\thelemma}{S.\arabic{lemma}}
\renewcommand{\theHlemma}{S.\arabic{lemma}}
\setcounter{equation}{0}
\renewcommand{\theequation}{S.\arabic{equation}}
\renewcommand{\theHequation}{S.\arabic{equation}}
\setcounter{theorem}{0}
\renewcommand{\thetheorem}{S.\arabic{theorem}}
\renewcommand{\theHtheorem}{S.\arabic{theorem}}
\setcounter{remark}{0}
\renewcommand{\theremark}{S.\arabic{remark}}
\renewcommand{\theHremark}{S.\arabic{remark}}
\setcounter{definition}{0}
\renewcommand{\thedefinition}{S.\arabic{definition}}
\renewcommand{\theHdefinition}{S.\arabic{definition}}
\setcounter{proposition}{0}
\renewcommand{\theproposition}{S.\arabic{proposition}}
\renewcommand{\theHproposition}{S.\arabic{proposition}}
\setcounter{figure}{0}
\renewcommand{\thefigure}{S.\arabic{figure}}
\renewcommand{\theHfigure}{S.\arabic{figure}}
\setcounter{table}{0}
\renewcommand{\thetable}{S.\arabic{table}}
\renewcommand{\theHtable}{S.\arabic{table}}
\setcounter{condition}{0}
\renewcommand{\thecondition}{S.\arabic{condition}}
\renewcommand{\theHcondition}{S.\arabic{condition}}
\setcounter{assumption}{0}
\renewcommand{\theassumption}{S.\arabic{assumption}}
\renewcommand{\theHassumption}{S.\arabic{assumption}}
\setcounter{algorithm}{0}
\renewcommand{\thealgorithm}{S.\arabic{algorithm}}
\renewcommand{\theHalgorithm}{S.\arabic{algorithm}}

This supplementary material contains technical proofs of theoretical results in Sections~\ref{supsec.A} and \ref{supsec.B}, further methodology and derivations in Section~\ref{supsec.C}, and additional simulation results in Section~\ref{supsec.D}. 

Throughout, let $n^{-1/2}\bY=n^{-1/2}(\by_1,\dots,\by_n)^{\top}=\widetilde{\bXi}\widehat{\bTheta}^{1/2}\bXi^{\top}=\sum_{j=1}^{p} \hat{\theta}_j^{1/2} \widetilde{\bxi}_j \widehat{\bxi}_j^{\top}$ be the $n \times p$ matrix, where $\hat{\theta}_1^{1/2}\ge \cdots\ge\hat{\theta}_p^{1/2}\ge 0$ are the singular values of $n^{-1/2}\bY.$ Without loss of generality, we assume that $\{\by_t\}_{t\in\eZ},\{\bx_t\}_{t\in\eZ}$ and $\{\bz_t\}_{t\in\eZ}$ have been centered to have mean zero.
Then,
$\widehat{\bOmega}_{y}=n^{-1}\bY^{\top}\bY=\sum_{j=1}^{p} \hat{\theta}_j \widehat{\bxi}_j \widehat{\bxi}_j^{\top},$
and
$n^{-1}\bY\bY^{\top}=\sum_{j=1}^{p} \hat{\theta}_j \widetilde{\bxi}_j \widetilde{\bxi}_j^{\top}.$
Moreover, we define a $n \times n$ matrix $\bD_k=\{D_{ksj}\}_{n\times n}$ with $D_{ksj}=I(s=k+j)$ for $s,j\in[n]$, where $I(\cdot)$ is the indicator function.
Then, it follows that $\widehat{\bOmega}_{y}(k)=(n-k)^{-1}\bY^{\top}\bD_k \bY,$ and
\begin{eqnarray}\label{eq1}
\widehat{\bOmega}_{y}(k)\widehat{\bW}\widehat{\bOmega}_{y}(k)^{\top}=\widehat{\bOmega}_{y}(k) \widehat{\bXi}_q\widehat{\bTheta}_{q}^{-1}\widehat{\bXi}_q^{\top} \widehat{\bOmega}_{y}(k)^{\top},
\end{eqnarray}
where $\widehat{\bXi}_q=(\widehat{\bxi}_1,\dots,\widehat{\bxi}_q)$ is a $p \times q$ matrix, and $\widehat{\bTheta}_{q}=\diag(\hat{\theta}_1,\dots,\hat{\theta}_q)$.

\section{Proofs of theoretical results in Section~\ref{sec.theory_basic}}
\label{supsec.A}
 
Let $\bX=(\bx_1,\dots,\bx_n)^{\top}$ and $\bE=(\be_1,\dots,\be_n)^{\top}$. Then, model~\eqref{eq.model} can be rewritten in the matrix form as $\bY=\bX\bA^{\top}+\bE.$ We first introduce two useful lemmas that provide upper bounds on the perturbation of matrix eigenvalues and eigenvectors, and are used multiple times in the proofs of this paper. Let $\{\lambda_j\}_{j \in [p]}$ be the eigenvalues of $\bSigma\in\eR^{p \times p}$ in a descending order and $\{\bxi_j\}_{j \in [p]}$ be the corresponding eigenvectors. Similarly, $\{\widetilde{\lambda}_j\}_{j \in [p]}$ and $\{\widetilde{\bxi}_j\}_{j \in [p]}$ are the corresponding eigenvalues and eigenvectors of $\widetilde{\bSigma} \in {\mathbb R}^{p \times p},$ respectively. 
\begin{lemma}
\label{lem.weyl}
(Weyl's theorem; \textcolor{blue}{Weyl} (\textcolor{blue}{1912})). $|\widetilde{\lambda}_j-\lambda_j|\le\Vert\widetilde{\bSigma}-\bSigma\Vert$ for $j \in [p].$
\end{lemma}

\begin{lemma}
\label{lem.sintheta}
(A useful variant of ${\rm sin}(\theta)$ theorem; \textcolor{blue}{Yu et al.} (\textcolor{blue}{2015})). If $\widetilde{\bxi}_j^{\top}\bxi_j\ge0$ for $j \in [p]$, then we have
$$\Vert\widetilde{\bxi}_{j}-\bxi_j\Vert\le\frac{\Vert\widetilde{\bSigma}-\bSigma\Vert/\sqrt{2}}{|\widetilde{\lambda}_{j-1}-\lambda_{j}|\wedge|\lambda_{j}-\widetilde{\lambda}_{j+1}|}.$$
\end{lemma}

\subsection{Proof of Theorem \ref{theoforbasic}}
Let $\hat{\lambda}_{k1}\ge\cdots \ge\hat{\lambda}_{kq} \ge 0$ be the eigenvalues of $\widehat{\bOmega}_{y}(k)\widehat{\bW} \widehat{\bOmega}_{y}(k)^{\top}$, with the corresponding spectral decomposition
\begin{equation}
    \label{basiceq000anew}
    \widehat{\bOmega}_{y}(k)\widehat{\bW}\widehat{\bOmega}_{y}(k)^{\top}=\sum_{i=1}^{q} \hat{\lambda}_{ki}\widehat{\bphi}_{ki} \widehat{\bphi}_{ki}^{\top}=\widehat{\bPhi}_{k,A} \widehat{\bLambda}_{k,A} \widehat{\bPhi}_{k,A}^{\top}+\widehat{\bPhi}_{k,E} \widehat{\bLambda}_{k,E} \widehat{\bPhi}_{k,E}^{\top},
\end{equation}
where $\widehat{\bLambda}_{k,A}=\diag(\hat{\lambda}_{k1},\dots,\hat{\lambda}_{kr_0})$, $\widehat{\bLambda}_{k,E}=\diag(\hat{\lambda}_{k,r_0+1},\dots,\hat{\lambda}_{kq}),\widehat{\bPhi}_{k,A}=(\widehat{\bphi}_{k1},\dots,\widehat{\bphi}_{kr_0})$, and $\widehat{\bPhi}_{k,E}=(\widehat{\bphi}_{k,r_0+1},\dots,\widehat{\bphi}_{kq}).$ 
To show Theorem~\ref{theoforbasic}, we first introduce two technical lemmas.

\begin{lemma}\label{basiclemma1}
For $r_0<q\le c(p\wedge n)$, where $c$ is specified in Condition~\ref{cond.E}, recall that $n^{-1/2}\bY=\sum_{j=1}^{p} \hat{\theta}_j^{1/2} \widetilde{\bxi}_j \widehat{\bxi}_j^{\top}$, and $\hat{\theta}_1 \ge\cdots \ge\hat{\theta}_q \ge 0$ are the eigenvalues of $\widehat{\bOmega}_y$. Let $\widehat{\bXi}_A=(\widehat{\bxi}_1,\dots,\widehat{\bxi}_{r_0})$. Then, under Conditions~\ref{cond.A}--\ref{cond.E}, $\hat{\theta}_1 \asymp \hat{\theta}_{r_0} \asymp p^{\delta_0}$ with probability tending to 1, $\hat{\theta}_{r_0+1} \asymp \hat{\theta}_{q} \asymp n^{-1}p$ with probability tending to 1, and
 \begin{eqnarray}\label{vjAeq0}
\big\|\widehat{\bXi}_A\widehat{\bXi}_A^{\top}-\bA\bA^{\top}\big\|=O_p(n^{-1/2}p^{(1-\delta_{0})/2}).
\end{eqnarray}
\end{lemma}
\begin{proof}
    Note that $\hat{\theta}_1^{1/2}\ge\cdots\ge\hat{\theta}_q^{1/2}$ are the $q$ largest singular values of $n^{-1/2}\bY=n^{-1/2}\bX\bA^{\top}+n^{-1/2}\bE$. By Conditions~\ref{cond.x}--\ref{cond.cov}, $n^{-1/2}\bX\bA^{\top}$ has $r_0$ singular values with the order $p^{\delta_0/2}$ with probability tending to 1. By Condition~\ref{cond.E} and Remark~\ref{rmk.cond}, the singular values of $\bE$ satisfy that $\sigma_1(\bE)\asymp\sigma_{\lfloor c(p\wedge n)\rfloor}(\bE)\asymp n^{1/2}+p^{1/2}$ with probability tending to 1. Since $\rank(\bY-\bE)=r_0$ and $r_0<q\le c(p\wedge n)$, we have $\hat{\theta}_1 \asymp \hat{\theta}_{r_0} \asymp p^{\delta_0},$ and $\hat{\theta}_{r_0+1} \asymp \hat{\theta}_{q} \asymp n^{-1}p$ with probability tending to 1. 

    By Condition~\ref{cond.x}(i), $\{\bx_t\}$ and $\{\be_t\}$ are uncorrelated, then $\Vert n^{-1}\bA\bX^{\top}\bE\Vert=O_p(n^{-1/2}p^{1/2+\delta_0/2})$ by Conditions~\ref{cond.A}--\ref{cond.E}. Considering that the $r_0$ largest eigenvalues of $n^{-1}\bA\bX^{\top}\bX\bA^{\top}$ are of order $p^{\delta_0}$ with probability tending to 1. Thus, we conclude \eqref{vjAeq0}, which completes the proof. 
\end{proof}

\begin{lemma}\label{basiclemma2}
Let $\hat{\mu}_{k1} \ge \cdots \ge \hat{\mu}_{kp} \ge 0$ be the eigenvalues of $\widehat{\bOmega}_{y}(k)\widehat{\bOmega}_{y}(k)^{\top}$.
Rewrite $\widehat{\bOmega}_{y}(k)$ as
$\widehat{\bOmega}_{y}(k)=\sum_{i=1}^{p} \hat{\mu}_{ki}^{1/2}\widehat{\bpsi}_{ki} \widetilde{\bpsi}_{ki}^{\top}=\widehat{\bPsi}_{k,A} \widehat{\bN}_{k,A}^{1/2}\widetilde{\bPsi}_{k,A}^{\top}+\widehat{\bPsi}_{k,E} \widehat{\bN}_{k,E}^{1/2} \widetilde{\bPsi}_{k,E}^{\top},$
where $\widehat{\bN}_{k,A}=\diag(\hat{\mu}_{k1},\dots,\hat{\mu}_{kr_0}),$ and $\widehat{\bN}_{k,E}=\diag(\hat{\mu}_{k,r_0+1},\dots,\hat{\mu}_{kp})$. Then, under Conditions~\ref{cond.A}--\ref{cond.auto}, for $r_0<q\le c(p\wedge n)$ and $k\in[m],$ we have $\hat{\mu}_{k1}  \asymp \hat{\mu}_{kr_0} \asymp p^{2\delta_0},\hat{\mu}_{k,r_0+1}  \asymp \hat{\mu}_{kq} \asymp n^{-2}p^2$ with probability tending to 1, and
 \begin{eqnarray}
&\big\|\widehat{\bPsi}_{k,A} \widehat{\bPsi}_{k,A}^{\top}-\bA\bA^{\top}\big\|=O_p(n^{-1/2}p^{(1-\delta_{0})/2}),\label{vjAeqq}\\
&\big\|\widetilde{\bPsi}_{k,A} \widetilde{\bPsi}_{k,A}^{\top}-\bA\bA^{\top}\big\|=O_p(n^{-1/2}p^{(1-\delta_{0})/2}).\label{vjAeq2}
\end{eqnarray}
\end{lemma}
\begin{proof}
We only prove the results $\hat{\mu}_{k,r_0+1}  \asymp \hat{\mu}_{kq} \asymp n^{-2}p^2$ with probability tending to 1, as the others can be obtained by similar procedures to those in the proof of Lemma~\ref{basiclemma1}. Since $\|\bE\|=O_p(n^{1/2}+p^{1/2})$, it can be shown that $\hat{\mu}_{k,r_0+1}=O_p(n^{-2}p^2)$.
Recall that $\bD_k=\{D_{ksj}\}_{n\times n}$ and $\widehat{\bOmega}_{y}(k)=(n-k)^{-1}\bY^{\top}\bD_k \bY.$ Then, we have
$$
\begin{aligned}
\sum_{i=r_0+1}^p\hat{\mu}_{ki}
=&\tr\{(\bI_p-\widehat{\bPsi}_{k,A} \widehat{\bPsi}_{k,A}^{\top})\widehat{\bOmega}_{y}(k)\widehat{\bOmega}_{y}(k)^{\top}(\bI_p-\widehat{\bPsi}_{k,A} \widehat{\bPsi}_{k,A}^{\top})\}\\
=&(n-k)^{-2}\tr\{(\bI_p-\widehat{\bPsi}_{k,A} \widehat{\bPsi}_{k,A}^{\top})\bY^{\top}\bD_k \bY\bY^{\top}\bD_k^{\top}\bY(\bI_p-\widehat{\bPsi}_{k,A} \widehat{\bPsi}_{k,A}^{\top})\}\\
=&(n-k)^{-2}\tr\{(\bI_p-\widehat{\bPsi}_{k,A} \widehat{\bPsi}_{k,A}^{\top})(\bA\bX^{\top}+\bE^{\top})\bD_k(\bX\bA^{\top}+\bE)\\
&\times(\bA\bX^{\top}+\bE^{\top})\bD_k^{\top}(\bX\bA^{\top}+\bE)(\bI_p-\widehat{\bPsi}_{k,A} \widehat{\bPsi}_{k,A}^{\top})\}\\
\ge&(n-k)^{-2}\|(\bI_p-\widehat{\bPsi}_{k,A} \widehat{\bPsi}_{k,A}^{\top})(\bA\bX^{\top}+\bE^{\top})\bD_k\bE(\bI_p-\bA\bA^{\top})\|_{\F}^2
\end{aligned}
$$
with probability tending to 1. Note that $\bD_k$ can be rewritten as $\bD_k=\sum_{i=1}^{n-k}\bs_{k+i}\bs_i^{\top}$, where $\bs_i$ is the $n$-dimensional vector with the $j$-th element $I\{i=j\}$. Then, the singular values of $\bR$ and $\bG^{\top}\bD_k\bG$ in Condition~\ref{cond.E} imply that
$(n-k)^{-2}\|(\bI_p-\widehat{\bPsi}_{k,A} \widehat{\bPsi}_{k,A}^{\top})\bE^{\top}\bD_k\bE(\bI_p-\bA\bA^{\top})\|_{\F}^2 \asymp n^{-1}p^2$ with probability tending to 1, and
$
(n-k)^{-2}\|(\bI_p-\widehat{\bPsi}_{k,A} \widehat{\bPsi}_{k,A}^{\top})\bA\bX^{\top}\bD_k\bE(\bI_p-\bA\bA^{\top})\|_{\F}^2 
\le r_0(n-k)^{-2}\|(\bI_p-\widehat{\bPsi}_{k,A} \widehat{\bPsi}_{k,A}^{\top})\bA\|_{\F}^2\|\bX^{\top}\bD_k\bE\|_{\F}^2=o_p(n^{-1}p^2).
$
Thus, there exists some constant $C_0>0$ such that $\sum_{i=r_0+1}^p\hat{\mu}_{ki} \ge C_0n^{-1}p^2 $ with probability tending to 1, which, together with $n=O(p)$ by Condition~\ref{cond.cov}(ii) and $\hat{\mu}_{k,r_0+1}=O_p(n^{-2}p^2)$, completes the proofs.
\end{proof}

Now we are ready to prove Theorem~\ref{theoforbasic}. 

(i) To study the asymptotic properties of $\hat{\lambda}_{k1} \ge \hat{\lambda}_{k2}\ge \cdots \ge \hat{\lambda}_{kq} \ge 0$, we consider the singular values of $\widehat{\bOmega}_{y}(k)\widehat{\bW}^{1/2}$. Note that
\begin{eqnarray*}
\widehat{\bOmega}_{y}(k)\widehat{\bW}^{1/2}=(\widehat{\bPsi}_{k,A} \widehat{\bN}_{k,A}^{1/2} \widetilde{\bPsi}_{k,A}^{\top}+\widehat{\bPsi}_{k,E} \widehat{\bN}_{k,E}^{1/2} \widetilde{\bPsi}_{k,E}^{\top})\Big(\sum_{j=1}^{q} \hat{\theta}_j^{-1/2} \widehat{\bxi}_j \widehat{\bxi}_j^{\top}\Big),
\end{eqnarray*}
where $\hat{\theta}_j,\widehat{\bxi}_j,\widehat{\bPsi}_{k,A},\widehat{\bN}_{k,A}^{1/2}$, and $\widetilde{\bPsi}_{k,A}$ are defined in Lemmas~\ref{basiclemma1} and \ref{basiclemma2}.
It can be shown that $\rank\big\{\widehat{\bPsi}_{k,A}\widehat{\bN}_{k,A}^{1/2}\widetilde{\bPsi}_{k,A}^{\top}(\sum_{j=1}^{q}\hat{\theta}_j^{-1/2} \widehat{\bxi}_j \widehat{\bxi}_j^{\top})\big\}=r_0$ with probability tending to 1, and the eigenvalues of $\widehat{\bN}_{k,A}^{1/2} \widetilde{\bPsi}_{k,A}^{\top}(\sum_{j=1}^{q} \hat{\theta}_j^{-1/} \widehat{\bxi}_j \widehat{\bxi}_j^{\top}) \widetilde{\bPsi}_{k,A} \widehat{\bN}_{k,A}^{1/2}$ are not smaller than the corresponding ordered eigenvalues of $\widehat{\bN}_{k,A}^{1/2} \widetilde{\bPsi}_{k,A}^{\top}(\sum_{j=1}^{r_0} \hat{\theta}_j^{-1} \widehat{\bxi}_j \widehat{\bxi}_j^{\top}) \widetilde{\bPsi}_{k,A} \widehat{\bN}_{k,A}^{1/2}.$
Then, by combining \eqref{vjAeq0} and \eqref{vjAeq2} in Lemmas \ref{basiclemma1} and \ref{basiclemma2}, it follows that 
$$
\Big\|\widetilde{\bPsi}_{k,A}^{\top}\Big(\sum_{j=1}^{r_0} \hat{\theta}_j^{-1} \widehat{\bxi}_j \widehat{\bxi}_j^{\top}\Big)\widetilde{\bPsi}_{k,A}\Big\| \asymp \Big\|\widetilde{\bPsi}_{k,A}^{\top}\Big(\sum_{j=1}^{r_0} \hat{\theta}_j^{-1} \widehat{\bxi}_j \widehat{\bxi}_j^{\top}\Big)\widetilde{\bPsi}_{k,A}\Big\|_{\min} \asymp p^{-\delta_0}
$$
with probability tending to 1, and
$$
\Big\|\widetilde{\bPsi}_{k,A}^{\top}\Big(\sum_{j=r_0+1}^{q} \hat{\theta}_j^{-1} \widehat{\bxi}_j \widehat{\bxi}_j^{\top}\Big)\widetilde{\bPsi}_{k,A}\Big\| =O_p(n^{-1}p^{1-\delta_0}\cdot np^{-1})=O_p(p^{-\delta_0}).
$$
Thus, we have
$$\Big\|\widehat{\bN}_{k,A}^{1/2} \widetilde{\bPsi}_{k,A}^{\top}\Big(\sum_{j=1}^{r_0} \hat{\theta}_j^{-1} \widehat{\bxi}_j \widehat{\bxi}_j^{\top}\Big) \widetilde{\bPsi}_{k,A} \widehat{\bN}_{k,A}^{1/2}\Big\| \asymp\Big\|\widehat{\bN}_{k,A}^{1/2} \widetilde{\bPsi}_{k,A}^{\top}\Big(\sum_{j=1}^{r_0} \hat{\theta}_j^{-1} \widehat{\bxi}_j \widehat{\bxi}_j^{\top}) \widetilde{\bPsi}_{k,A} \widehat{\bN}_{k,A}^{1/2}\Big\|_{\min} \asymp p^{\delta_0}$$
with probability tending to 1, and
$$
\Big\|\widehat{\bN}_{k,A}^{1/2} \widetilde{\bPsi}_{k,A}^{\top}\Big(\sum_{j=r_0+1}^{q} \hat{\theta}_j^{-1} \widehat{\bxi}_j \widehat{\bxi}_j^{\top}\Big) \widetilde{\bPsi}_{k,A} \widehat{\bN}_{k,A}^{1/2}\Big\|  =O_p(p^{2\delta_0}\cdot p^{-\delta_0})=O_p(p^{\delta_0}),
$$
which together shows that
\begin{eqnarray}\label{eq.jless_r0}
\Big\|\widehat{\bN}_{k,A}^{1/2} \widetilde{\bPsi}_{k,A}^{\top}\Big(\sum_{j=1}^{q} \hat{\theta}_j^{-1} \widehat{\bxi}_j \widehat{\bxi}_j^{\top}\Big) \widetilde{\bPsi}_{k,A} \widehat{\bN}_{k,A}^{1/2}\Big\| \asymp\Big\|\widehat{\bN}_{k,A}^{1/2} \widetilde{\bPsi}_{k,A}^{\top}\Big(\sum_{j=1}^{q} \hat{\theta}_j^{-1} \widehat{\bxi}_j \widehat{\bxi}_j^{\top}\Big) \widetilde{\bPsi}_{k,A} \widehat{\bN}_{k,A}^{1/2}\Big\|_{\min} \asymp p^{\delta_0}
\end{eqnarray}
with probability tending to 1. For $\hat{\lambda}_{kj}$ with $j>r_0,$ by Lemmas~\ref{basiclemma1} and \ref{basiclemma2}, we have
\begin{eqnarray}\label{eq.jlarge_r0}
\Big\|\widehat{\bN}_{k,E}^{1/2} \widetilde{\bPsi}_{k,E}^{\top} \Big(\sum_{j=1}^{q} \hat{\theta}_j^{-1} \widehat{\bxi}_j \widehat{\bxi}_j^{\top}\Big)\widetilde{\bPsi}_{k,E}\widehat{\bN}_{k,E}^{1/2}\Big\|=O_p( \hat{\theta}_{k,r_0+1}\hat{\theta}_q^{-1})=O_p(n^{-1}p).
\end{eqnarray}
Hence, it follows that $\hat{\lambda}_{k1} \asymp \hat{\lambda}_{kr_0} \asymp p^{\delta_0}$ with probability tending to 1 by combining \eqref{eq.jless_r0} and \eqref{eq.jlarge_r0}, and $\hat{\lambda}_{k,r_0+1}=O_p(n^{-1}p)=o_p(p^{\delta_0})$ since $p^{1-\delta_0}=o(n)$ by Condition~\ref{cond.cov}(ii).

(ii) Notice that 
$$
\begin{aligned}
    &\|\widehat{\bOmega}_y(k)\widehat{\bW}\widehat{\bOmega}_y(k)^{\top}-\widehat{\bPsi}_{k,A}\widehat{\bN}_{k,A}^{1/2}\widetilde{\bPsi}_{k,A}^{\top}\widehat{\bW}\widetilde{\bPsi}_{k,A}\widehat{\bN}_{k,A}^{1/2}\widehat{\bPsi}_{k,A}^{\top}\| \\
    \le&\|\widehat{\bPsi}_{k,E}\widehat{\bN}_{k,E}^{1/2}\widetilde{\bPsi}_{k,E}^{\top}\widehat{\bW}\widetilde{\bPsi}_{k,E}\widehat{\bN}_{k,E}^{1/2}\widehat{\bPsi}_{k,E}^{\top}\|+2\|\widehat{\bPsi}_{k,A}\widehat{\bN}_{k,A}^{1/2}\widetilde{\bPsi}_{k,A}^{\top}\widehat{\bW}\widetilde{\bPsi}_{k,E}\widehat{\bN}_{k,E}^{1/2}\widehat{\bPsi}_{k,E}^{\top}\|\\
    =&O_p(n^{-1}p+n^{-1/2}p^{(1+\delta_0)/2})=O_p(n^{-1/2}p^{(1+\delta_0)/2})
\end{aligned}
$$ 
by using \eqref{eq.jless_r0} and \eqref{eq.jlarge_r0}. Then by Lemma~\ref{lem.sintheta}, we have
\begin{eqnarray*}
\|\widehat{\bPhi}_{k,A}\widehat{\bPhi}_{k,A}^{\top}-\widehat{\bPsi}_{k,A} \widehat{\bPsi}_{k,A}^{\top}\|=O_p(n^{-1/2}p^{(1-\delta_{0})/2}),
\end{eqnarray*}
which, together with \eqref{vjAeqq} in Lemma \ref{basiclemma2}, obtains that $\big\|\widehat{\bPhi}_{k,A}\widehat{\bPhi}_{k,A}^{\top}-\bA\bA^{\top}\big\|=O_p(n^{-1/2}p^{(1-\delta_{0})/2})$ and completes the proof of Theorem~\ref{theoforbasic}.

\subsection{Proof of Proposition \ref{propos.est_r}}
By Theorem \ref{theoforbasic}(i), and the conditions $\vartheta_np^{-\delta_0}\to0$ and $\vartheta_n\gtrsim n^{-1}p$ as imposed in Proposition~\ref{propos.est_r}, we have $R_j\asymp1$ with probability tending to 1 for $j\le r_0-1$ and $r_0+1\le j\le q-1$, where $R_j$ is defined in \eqref{eq.deter_r}. Moreover, $R_{r_0}\asymp p^{\delta_0}\vartheta_n^{-1}\to\infty$ with probability tending to 1. Thus, $R_{r_0}>\max_{1\le j\neq r_0\le q-1}R_j$ with probability tending to 1, which implies that $\eP(\hat{r}_{0}=r_0)\to1$ as $p,n\to\infty.$

\subsection{Proof of Proposition \ref{propos.U_A}}
Let $\hat{\lambda}_{1}\ge\cdots \ge\hat{\lambda}_{q} \ge 0$ be the eigenvalues of $\widehat{\bM}=\sum_{k=1}^{m}\widehat{\bOmega}_{y}(k)\widehat{\bW} \widehat{\bOmega}_{y}(k)^{\top}$ with the corresponding eigenvectors $\widehat{\bphi}_{1},\dots,\widehat{\bphi}_{q},$ and $\widehat{\bA}=(\widehat{\bphi}_{1},\dots,\widehat{\bphi}_{\hat{r}_0})$. Denote $\widetilde{\bA}=(\widehat{\bphi}_{1},\dots,\widehat{\bphi}_{r_0})$ as the estimator of $\bA$ with the known value of $r_0.$ Note that $\widehat{\bM}=\sum_{k=1}^{m}\widehat{\bOmega}_{y}(k)\widehat{\bW}\widehat{\bOmega}_{y}(k)^{\top},$ where each term on the right hand is non-negative. Thus, by Theorem~\ref{theoforbasic}(i), we immediately obtain that $\hat{\lambda}_{1}\asymp\hat{\lambda}_{r_0}\asymp p^{\delta_0}$ with probability tending to 1, and $\hat{\lambda}_{r_0+1}=O_p(n^{-1}p)$. By Lemma~\ref{lem.sintheta} and similar procedures to the proof of Theorem~\ref{theoforbasic}(ii), it can be shown that there exists an orthogonal matrix $\bU\in\eR^{r_0\times r_0}$ such that $\|\widetilde{\bA}-\bA\bU\|=O_p(n^{-1/2}p^{(1-\delta_{0})/2})$, and thus 
\begin{equation}
    \label{eq.tilde_A}
    \big\|\widetilde{\bA}\widetilde{\bA}^{\top}-\bA\bA^{\top}\big\|=O_p(n^{-1/2}p^{(1-\delta_{0})/2}).
\end{equation}

With the aid of weak consistency of $\hat{r}_{0}$ in Proposition~\ref{propos.est_r}, the convergence rate of $\widehat{\bA}$ with the estimated $\hat{r}_0$ is the same as $\widetilde{\bA}$ with the known $r_0$. To see this, let $\varrho_{n}=n^{-1/2}p^{(1-\delta_{0})/2}$. It holds that for any constant $\tilde{c}>0$ that 
\begin{equation}
\label{eq.rho_n}
\begin{aligned}
\eP\big(\varrho_{n}^{-1}\|\widehat{\bA}\widehat{\bA}^{\top}-\bA\bA^{\top}\|>\tilde{c}\big)\le&\eP\big(\varrho_{n}^{-1}\|\widehat{\bA}\widehat{\bA}^{\top}-\bA\bA^{\top}\|>\tilde{c}\mid\hat{r}_{0}=r_0\big)+\eP(\hat{r}_{0}\neq r_0)\\
\le&\eP\big(\varrho_{n}^{-1}\|\widetilde{\bA}\widetilde{\bA}^{\top}-\bA\bA^{\top}\|>\tilde{c}\big)+o(1),
\end{aligned}
\end{equation}
which, together with \eqref{eq.tilde_A}, shows that $\big\|\widehat{\bA}\widehat{\bA}^{\top}-\bA\bA^{\top}\big\|=O_p(\varrho_{n})=O_p(n^{-1/2}p^{(1-\delta_{0})/2}).$

For the measure of the distance between two columns spaces, it follows from the orthogonality of columns of $\widetilde{\bA}$ and $\bA$ that
$$
\begin{aligned}
    \cD^2\big\{\cC(\widetilde{\bA}),\cC(\bA)\big\}=&1-r_0^{-1}\tr(\bA\bA^{\top}\widetilde{\bA}\widetilde{\bA}^{\top})=r_0^{-1}\tr(\bA^{\top}\bA-\bA\bA^{\top}\widetilde{\bA}\widetilde{\bA}^{\top})\\
    =&r_0^{-1}\tr\{\bU^{\top}\bA^{\top}(\bI_p-\widetilde{\bA}\widetilde{\bA}^{\top})\bA\bU\}\\
    =&r_0^{-1}\left[\tr\{\bU^{\top}\bA^{\top}(\bI_p-\widetilde{\bA}\widetilde{\bA}^{\top})\bA\bU\}-\tr\{\bU^{\top}\bA^{\top}(\bI_p-\bA\bA^{\top})\bA\bU\}\right]\\
    =&r_0^{-1}\tr\{\bU^{\top}\bA^{\top}(\bA\bA^{\top}-\widetilde{\bA}\widetilde{\bA}^{\top})\bA\bU\}\\
    \le&\Vert\bU^{\top}\bA^{\top}(\bA\bA^{\top}-\widetilde{\bA}\widetilde{\bA}^{\top})\bA\bU\Vert\\
    =&\Vert(\bA\bU-\widetilde{\bA})^{\top}(\bA\bU-\widetilde{\bA})-\bU^{\top}\bA^{\top}(\bA\bU-\widetilde{\bA})(\bA\bU-\widetilde{\bA})^{\top}\bA\bU\Vert\\
    \le&2\Vert\widetilde{\bA}-\bA\bU\Vert^2=O_p(n^{-1}p^{1-\delta_{0}}),
\end{aligned}
$$
which implies that $\cD\big\{\cC(\widetilde{\bA}),\cC(\bA)\big\}=O_p(n^{-1/2}p^{(1-\delta_{0})/2})$. By using the similar arguments to \eqref{eq.rho_n}, we have $\cD\big\{\cC(\widehat{\bA}),\cC(\bA)\big\}=O_p(n^{-1/2}p^{(1-\delta_{0})/2})$. The proof is complete.

\section{Proofs of theoretical results in Section~\ref{sec.theory_weak}}
\label{supsec.B}

Define $\bX=(\bx_1,\dots,\bx_n)^{\top}$, $\bZ=(\bz_1,\dots,\bz_n)^{\top}$ and $\bE=(\be_1,\dots,\be_n)^{\top}$. Then model~\eqref{eq.model_2} can be expressed as $\bY=\bX\bA^{\top}+\bZ\bB^{\top}+\bE.$

\subsection{Proof of Theorem \ref{advsam1}}
Let $\hat{\theta}_{1} \ge \cdots \ge \hat{\theta}_{p} \ge 0$ be the eigenvalues of $\widehat{\bOmega}_{y}$, with the corresponding spectral decomposition written as
$$
\widehat{\bOmega}_{y}=\sum_{i=1}^{p}\hat{\theta}_{i}\widehat{\bxi}_{i}\widehat{\bxi}_{i}^{\top}=\widehat{\bXi}_{A}\widehat{\bTheta}_{A}\widehat{\bXi}_{A}^{\top}+\widehat{\bXi}_{B}\widehat{\bTheta}_{B}\widehat{\bXi}_{B}^{\top}+\widehat{\bXi}_{E}\widehat{\bTheta}_{E}\widehat{\bXi}_{E}^{\top},
$$
where $\widehat{\bTheta}_{A}=\diag(\hat{\theta}_{1},\dots,\hat{\theta}_{r_0}),\widehat{\bTheta}_{B}=\diag(\hat{\theta}_{r_0+1},\dots,\hat{\theta}_{r}),\widehat{\bTheta}_{E}=\diag(\hat{\theta}_{r+1},\dots,\hat{\theta}_{p}),\widehat{\bXi}_{A}=(\widehat{\bxi}_1,\dots,\widehat{\bxi}_{r_0}),\widehat{\bXi}_{B}=(\widehat{\bxi}_{r_0+1},\dots,\widehat{\bxi}_{r}),$ and $\widehat{\bXi}_{E}=(\widehat{\bxi}_{r+1},\dots,\widehat{\bxi}_{p})$.

(i) Firstly, $\hat{\theta}_1^{1/2}\ge \cdots \ge \hat{\theta}_{\lfloor c'(p\wedge n)\rfloor}^{1/2} \ge 0$ are the $\lfloor c'(p\wedge n)\rfloor$ largest singular values of $n^{-1/2}\bY=n^{-1/2}\bX\bA^{\top}+n^{-1/2}\bZ\bB^{\top}+n^{-1/2}\bE$. By Conditions \ref{cond.x_a}--\ref{cond.E_a}, $n^{-1/2}\bX\bA^{\top}$ has $r_0$ singular values of order $p^{\delta_0/2}$ with probability tending to 1, $\|n^{-1/2}\bZ\bB^{\top}\|=O_p(p^{\delta_1/2})=o_p(p^{\delta_0/2}),$ and $\|n^{-1/2}\bE\|=O_p(n^{-1/2}p^{1/2})=o_p(p^{\delta_0/2}),$ which together derive $\hat{\theta}_{1} \asymp \hat{\theta}_{r_0} \asymp p^{\delta_{0}}$ with probability tending to 1.

Secondly, consider that
$n^{-1/2}\bY(\bI_p-\widehat{\bXi}_A\widehat{\bXi}_A^{\top})=n^{-1/2}(\bX\bA^{\top}+\bZ\bB^{\top}+\bE)(\bI_p-\widehat{\bXi}_A\widehat{\bXi}_A^{\top})
    =n^{-1/2}(\bZ\bB^{\top}+\bE)(\bI_p-\bA\bA^{\top})+n^{-1/2}(\bX\bA^{\top}+\bZ\bB^{\top}+\bE)(\bA\bA^{\top}-\widehat{\bXi}_A\widehat{\bXi}_A^{\top}).
$
Condition \ref{cond.A_a} and Lemma~\ref{lem.weyl} together ensure that $\|\bB^{\top}(\bI_p-\bA\bA^{\top})\|_{\min}\ge\Vert\bB\Vert_{\min}-\|\bA\bA^{\top}\bB\|>0$, which further implies that $$\Vert n^{-1/2}(\bZ\bB^{\top}+\bE)(\bI_p-\bA\bA^{\top})\Vert\asymp\Vert n^{-1/2}(\bZ\bB^{\top}+\bE)(\bI_p-\bA\bA^{\top})\Vert_{\min}\asymp p^{\delta_1/2}$$ with probability tending to 1. 
By Theorem \ref{advsam1}(ii), which will be proved later, and Conditions~\ref{cond.A_a}--\ref{cond.E_a}, we have $\|n^{-1/2}(\bX\bA^{\top}+\bZ\bB^{\top}+\bE)(\bA\bA^{\top}-\widehat{\bXi}_A\widehat{\bXi}_A^{\top})\|=O_p(n^{-1/2}p^{1/2}+p^{\delta_1-\delta_0/2})=o_p(p^{\delta_1/2}),$ which implies that $\|n^{-1/2}\bY(\bI_p-\widehat{\bXi}_A\widehat{\bXi}_A^{\top})-n^{-1/2}(\bZ\bB^{\top}+\bE)(\bI_p-\bA\bA^{\top})\|=o_p(p^{\delta_1/2}).$ Note that $\hat{\theta}_{r_0+1}^{1/2}\ge \cdots \ge \hat{\theta}_r^{1/2} \ge 0$ are the $r_1$ largest singular values of $n^{-1/2}\bY(\bI_p-\widehat{\bXi}_A\widehat{\bXi}_A^{\top})$.
Thus, we have $\hat{\theta}_{r_0+1} \asymp \hat{\theta}_{r} \asymp p^{\delta_{1}}$ with probability tending to 1.

Thirdly, by Condition \ref{cond.E_a} and Remark~\ref{rmk.cond}, the singular values of $\bE$ satisfy that and $\sigma_1(\bE) \asymp\sigma_{\lfloor c'(p\wedge n)\rfloor}(\bE) \asymp n^{1/2}+p^{1/2}$ with probability tending to 1. Since $\rank(\bY-\bE)=r$, we conclude that $\hat{\theta}_{r+1} \asymp \hat{\theta}_{\lfloor c'(p\wedge n)\rfloor} \asymp n^{-1}p$ with probability tending to 1. Then, the proof is complete. 

(ii) By Condition \ref{cond.x_a}(i), $\{\bx_t\}_{t\in\eZ},\{\bz_t\}_{t\in\eZ}$ and $\{\be_t\}_{t\in\eZ}$ are mutually uncorrelated. Then, by Conditions~\ref{cond.A_a}--\ref{cond.E_a}, we have
$$
\begin{aligned}
    &\|n^{-1}\bA\bX^{\top}\bZ\bB^{\top}\|\le \|n^{-1}\bX^{\top} \bZ\| =O_p(n^{-1/2}p^{\delta_1/2+\delta_0/2})=O_p(n^{-1/2}p^{1/2+\delta_0/2}),\\
    & \|n^{-1}\bA\bX^{\top} \bE\|\le \|n^{-1}\bX^{\top} \bE\| =O_p(n^{-1/2}p^{1/2+\delta_0/2}),\\
    & \Vert n^{-1}\bB\bZ^{\top}\bZ\bB^{\top}\Vert=O_p(p^{\delta_1}),~~\Vert n^{-1}\bE^{\top}\bE\Vert=O_p(n^{-1}p),
\end{aligned}
$$
which together imply that $\|n^{-1}\bY^{\top} \bY-n^{-1}\bA\bX^{\top} \bX\bA^{\top}\|=O_p(n^{-1/2}p^{1/2+\delta_0/2}+p^{\delta_1}).$
Consider that the $r_0$ largest eigenvalues of $n^{-1}\bA\bX^{\top}\bX\bA^{\top}$ are of order $p^{\delta_0}$ with probability tending to 1. Thus, by using Lemma~\ref{lem.sintheta}, we conclude $\big\|\widehat{\bXi}_{A}\widehat{\bXi}_{A}^{\top}-\bA\bA^{\top}\big\|=O_p(n^{-1/2}p^{(1-\delta_0)/2}+p^{\delta_1-\delta_0}).$

\subsection{Proof of Theorem \ref{advauto}}
Let $\hat{\mu}_{k1} \ge \cdots \ge \hat{\mu}_{kp} \ge 0$ be the eigenvalues of $\widehat{\bOmega}_{y}(k)\widehat{\bOmega}_{y}(k)^{\top}$, with the corresponding spectral decomposition
 \begin{eqnarray}\label{eq000a}
\nonumber&&\widehat{\bOmega}_{y}(k)\widehat{\bOmega}_{y}(k)^{\top}=\sum_{i=1}^p \hat{\mu}_{ki}\widehat{\bpsi}_{ki} \widehat{\bpsi}_{ki}^{\top}
=\widehat{\bPsi}_{k,A} \widehat{\bN}_{k,A} \widehat{\bPsi}_{k,A}^{\top}+\widehat{\bPsi}_{k,B} \widehat{\bN}_{k,B} \widehat{\bPsi}_{k,B}^{\top}+\widehat{\bPsi}_{k,E} \widehat{\bN}_{k,E} \widehat{\bPsi}_{k,E}^{\top},
\end{eqnarray}
where $\widehat{\bN}_{k,A}=\diag(\hat{\mu}_{k1},\dots,\hat{\mu}_{kr_0}),\widehat{\bN}_{k,B}=\diag(\hat{\mu}_{k,r_0+1},\dots,\hat{\mu}_{kr}),\widehat{\bN}_{k,E}=\diag(\hat{\mu}_{k,r+1},\dots,\hat{\mu}_{kp}),$ $\widehat{\bPsi}_{k,A}=(\widehat{\bpsi}_{k1},\dots,\widehat{\bpsi}_{kr_{0}}),\widehat{\bPsi}_{k,B}=(\widehat{\bpsi}_{k,r_0+1},\dots,\widehat{\bpsi}_{kr})$, and $\widehat{\bPsi}_{k,E}=(\widehat{\bpsi}_{k,r+1},\dots,\widehat{\bpsi}_{kp})$.
To show Theorem~\ref{advauto}, we first introduce a technical lemma.
\begin{lemma}
    \label{lem.auto}
    Suppose that Conditions~\ref{cond.A_a}--\ref{cond.auto_a} hold and $\hat{\mu}_{k1} \asymp \hat{\mu}_{kr_0} \asymp p^{2\delta_{0}}$ with probability tending to 1. Then, for $k\in[m]$,
    \begin{equation}
        \label{eq.lem_auto}
        \Vert(\bI_p-\widehat{\bPsi}_{k,A} \widehat{\bPsi}_{k,A}^{\top})\widehat{\bOmega}_{y}(k)\widehat{\bPsi}_{k,A} \widehat{\bPsi}_{k,A}^{\top}\widehat{\bOmega}_{y}(k)^{\top}(\bI_p-\widehat{\bPsi}_{k,A} \widehat{\bPsi}_{k,A}^{\top})\Vert=o_p(p^{2\delta_{1k}}).
    \end{equation}
\end{lemma}
\begin{proof}
    Consider that
    \begin{equation}
        \label{ssf0}
        \begin{aligned}
            &(\bI_p-\widehat{\bPsi}_{k,A} \widehat{\bPsi}_{k,A}^{\top})\widehat{\bOmega}_{y}(k)\widehat{\bPsi}_{k,A} \widehat{\bPsi}_{k,A}^{\top}\widehat{\bOmega}_{y}(k)^{\top}(\bI_p-\widehat{\bPsi}_{k,A} \widehat{\bPsi}_{k,A}^{\top})\\
            =&\widehat{\bPsi}_{k,B} \widehat{\bPsi}_{k,B}^{\top}\widehat{\bOmega}_{y}(k)\widehat{\bPsi}_{k,A} \widehat{\bPsi}_{k,A}^{\top}\widehat{\bOmega}_{y}(k)^{\top}\widehat{\bPsi}_{k,B} \widehat{\bPsi}_{k,B}^{\top}
            +\widehat{\bPsi}_{k,B} \widehat{\bPsi}_{k,B}^{\top}\widehat{\bOmega}_{y}(k)\widehat{\bPsi}_{k,A} \widehat{\bPsi}_{k,A}^{\top}\widehat{\bOmega}_{y}(k)^{\top}\widehat{\bPsi}_{k,E} \widehat{\bPsi}_{k,E}^{\top}\\
            &+\widehat{\bPsi}_{k,E} \widehat{\bPsi}_{k,E}^{\top}\widehat{\bOmega}_{y}(k)\widehat{\bPsi}_{k,A} \widehat{\bPsi}_{k,A}^{\top}\widehat{\bOmega}_{y}(k)^{\top}\widehat{\bPsi}_{k,B} \widehat{\bPsi}_{k,B}^{\top}
            +\widehat{\bPsi}_{k,E} \widehat{\bPsi}_{k,E}^{\top}\widehat{\bOmega}_{y}(k)\widehat{\bPsi}_{k,A} \widehat{\bPsi}_{k,A}^{\top}\widehat{\bOmega}_{y}(k)^{\top}\widehat{\bPsi}_{k,E} \widehat{\bPsi}_{k,E}^{\top}.
        \end{aligned}
    \end{equation}
    We first derive the order of the first term in \eqref{ssf0}, i.e., the order of $\widehat{\bPsi}_{k,B}^{\top}\widehat{\bOmega}_{y}(k)\widehat{\bPsi}_{k,A}$. To this end, by the properties of eigenvectors, we have
    $$
    \begin{aligned}
    \bzero=&\widehat{\bPsi}_{k,B}^{\top}\widehat{\bOmega}_{y}(k)\widehat{\bOmega}_{y}(k)^{\top}\widehat{\bPsi}_{k,A}\\
    =&\widehat{\bPsi}_{k,B}^{\top}\widehat{\bOmega}_{y}(k)\widehat{\bPsi}_{k,A} \widehat{\bPsi}_{k,A}^{\top}\widehat{\bOmega}_{y}(k)^{\top}\widehat{\bPsi}_{k,A}+\widehat{\bPsi}_{k,B}^{\top}\widehat{\bOmega}_{y}(k)\widehat{\bPsi}_{k,B} \widehat{\bPsi}_{k,B}^{\top}\widehat{\bOmega}_{y}(k)^{\top}\widehat{\bPsi}_{k,A}\\
    &+\widehat{\bPsi}_{k,B}^{\top}\widehat{\bOmega}_{y}(k)\widehat{\bPsi}_{k,E} \widehat{\bPsi}_{k,E}^{\top}\widehat{\bOmega}_{y}(k)^{\top}\widehat{\bPsi}_{k,A}.
    \end{aligned}
    $$
    The orders of each term are calculated as follows: 
    \begin{itemize}
        \item [(i)] $\|\widehat{\bPsi}_{k,A}^{\top}\widehat{\bOmega}_{y}(k)^{\top}\widehat{\bPsi}_{k,A}\| \asymp \|\widehat{\bPsi}_{k,A}^{\top}\widehat{\bOmega}_{y}(k)^{\top}\widehat{\bPsi}_{k,A}\|_{\min} \asymp p^{\delta_0}$ with probability tending to 1, by noticing that $\hat{\mu}_{k1} \asymp \hat{\mu}_{kr_0} \asymp p^{2\delta_{0}}$, $\rank\big(\widehat{\bPsi}_{k,A}^{\top}\widehat{\bOmega}_{y}(k)^{\top}\widehat{\bPsi}_{k,A}\big)=r_0$, and $\widehat{\bPsi}_{k,A}^{\top}\widehat{\bOmega}_{y}(k)^{\top}\widehat{\bPsi}_{k,A}=\widehat{\bPsi}_{k,A}^{\top}\widetilde{\bPsi}_{k,A}\widehat{\bN}_{k,A}^{1/2}=\widehat{\bPsi}_{k,A}^{\top}\widehat{\bN}_{k,A}^{1/2}\widetilde{\bPsi}_{k,A}.$
    
        \item [(ii)] $\|\widehat{\bPsi}_{k,B}^{\top}\widehat{\bOmega}_{y}(k)\widehat{\bPsi}_{k,B}\| \asymp p^{\delta_{1k}}$ with probability tending to 1, by noticing that
        $$
        \begin{aligned}
        &\widehat{\bPsi}_{k,B}^{\top}\widehat{\bOmega}_{y}(k)\widehat{\bPsi}_{k,B}=\widehat{\bPsi}_{k,B}^{\top}(\bA\bX^{\top}+\bB\bZ^{\top}+\bE^{\top})\bD_k(\bX\bA^{\top}+\bZ\bB^{\top}+\bE)\widehat{\bPsi}_{k,B}\\
        =&\widehat{\bPsi}_{k,B}^{\top}(\bI_p-\widehat{\bPsi}_{k,A} \widehat{\bPsi}_{k,A}^{\top})(\bA\bX^{\top}+\bB\bZ^{\top}+\bE^{\top})\bD_k(\bX\bA^{\top}+\bZ\bB^{\top}+\bE)(\bI_p-\widehat{\bPsi}_{k,A} \widehat{\bPsi}_{k,A}^{\top})\widehat{\bPsi}_{k,B}\\
        =&\widehat{\bPsi}_{k,B}^{\top}(\bB\bZ^{\top}+\bE^{\top})\bD_k(\bZ\bB^{\top}+\bE)\widehat{\bPsi}_{k,B}\\
        &+\widehat{\bPsi}_{k,B}^{\top}(\bA\bA^{\top}-\widehat{\bPsi}_{k,A} \widehat{\bPsi}_{k,A}^{\top})\bA\bX^{\top}\bD_k\bX\bA^{\top}(\bA\bA^{\top}-\widehat{\bPsi}_{k,A} \widehat{\bPsi}_{k,A}^{\top})\widehat{\bPsi}_{k,B}\\
        &+\widehat{\bPsi}_{k,B}^{\top}(\bA\bA^{\top}-\widehat{\bPsi}_{k,A} \widehat{\bPsi}_{k,A}^{\top})\bA\bX^{\top}\bD_k(\bZ\bB^{\top}+\bE)\widehat{\bPsi}_{k,B}\\
        &+\widehat{\bPsi}_{k,B}^{\top}(\bB\bZ^{\top}+\bE^{\top})\bD_k\bX\bA^{\top}(\bA\bA^{\top}-\widehat{\bPsi}_{k,A} \widehat{\bPsi}_{k,A}^{\top})\widehat{\bPsi}_{k,B}.
        \end{aligned}
        $$
    
        \item [(iii)] $\|\widehat{\bPsi}_{k,B}^{\top}\widehat{\bOmega}_{y}(k)^{\top}\widehat{\bPsi}_{k,A}\|=O_p(p^{\delta_{1k}}+n^{-1/2}p^{(1+\delta_0)/2})$, by noticing that
        $$
        \begin{aligned}
            &\widehat{\bPsi}_{k,A}^{\top}\widehat{\bOmega}_{y}(k)\widehat{\bPsi}_{k,B}=\widehat{\bPsi}_A^{\top}(\bA\bX^{\top}+\bB\bZ^{\top}+\bE^{\top})\bD_k(\bX\bA^{\top}+\bZ\bB^{\top}+\bE)\widehat{\bPsi}_{k,B}\\
            =&\widehat{\bPsi}_{k,A}^{\top}(\bA\bX^{\top}+\bB\bZ^{\top}+\bE^{\top})\bD_k(\bX\bA^{\top}+\bZ\bB^{\top}+\bE)(\bI_p-\widehat{\bPsi}_{k,A} \widehat{\bPsi}_{k,A}^{\top})\widehat{\bPsi}_{k,B}\\
            =&\widehat{\bPsi}_{k,A}^{\top}(\bA\bX^{\top}+\bB\bZ^{\top}+\bE^{\top})\bD_k(\bZ\bB^{\top}+\bE)\widehat{\bPsi}_{k,B}\\
            &+\widehat{\bPsi}_{k,A}^{\top}(\bA\bX^{\top}+\bB\bZ^{\top}+\bE^{\top})\bD_k\bX\bA^{\top}(\bA\bA^{\top}-\widehat{\bPsi}_{k,A} \widehat{\bPsi}_{k,A}^{\top})\widehat{\bPsi}_{k,B}.
        \end{aligned}
        $$
    
        \item [(iv)] $\|\widehat{\bPsi}_{k,B}^{\top}\widehat{\bOmega}_{y}(k)\widehat{\bPsi}_{k,E}\|=O_p(p^{\delta_{1k}}),$ by noticing that
        $$
        \begin{aligned}
            &\widehat{\bPsi}_{k,B}^{\top}\widehat{\bOmega}_{y}(k)\widehat{\bPsi}_{k,E}=\widehat{\bPsi}_{k,B}^{\top}(\bA\bX^{\top}+\bB\bZ^{\top}+\bE^{\top})\bD_k(\bX\bA^{\top}+\bZ\bB^{\top}+\bE)\widehat{\bPsi}_{k,E}\\
            =&\widehat{\bPsi}_{k,B}^{\top}(\bI_p-\widehat{\bPsi}_{k,A} \widehat{\bPsi}_{k,A}^{\top})(\bA\bX^{\top}+\bB\bZ^{\top}+\bE^{\top})\bD_k(\bX\bA^{\top}+\bZ\bB^{\top}+\bE)(\bI_p-\widehat{\bPsi}_{k,A} \widehat{\bPsi}_{k,A}^{\top})\widehat{\bPsi}_{k,E}\\
            =&\widehat{\bPsi}_{k,B}^{\top}(\bB\bZ^{\top}+\bE^{\top})\bD_k(\bZ\bB^{\top}+\bE)\widehat{\bPsi}_{k,E}\\
            &+\widehat{\bPsi}_{k,B}\bB^{\top}(\bA\bA^{\top}-\widehat{\bPsi}_{k,A} \widehat{\bPsi}_{k,A}^{\top})\bA\bX^{\top}\bD_k\bX\bA^{\top}(\bA\bA^{\top}-\widehat{\bPsi}_{k,A} \widehat{\bPsi}_{k,A}^{\top})\widehat{\bPsi}_{k,E}\\
            &+\widehat{\bPsi}_{k,B}^{\top}(\bA\bA^{\top}-\widehat{\bPsi}_{k,A} \widehat{\bPsi}_{k,A}^{\top})\bA\bX^{\top}\bD_k(\bZ\bB^{\top}+\bE)\widehat{\bPsi}_{k,E}\\
            &+\widehat{\bPsi}_{k,B}^{\top}(\bB\bZ^{\top}+\bE^{\top})\bD_k\bX\bA^{\top}(\bA\bA^{\top}-\widehat{\bPsi}_{k,A} \widehat{\bPsi}_{k,A}^{\top})\widehat{\bPsi}_{k,E}.
        \end{aligned}
        $$
    
        \item [(v)] $\|\widehat{\bPsi}_{k,E}^{\top}\widehat{\bOmega}_{y}(k)^{\top}\widehat{\bPsi}_{k,A}\|=O_p(p^{\delta_{1k}}+n^{-1/2}p^{(1+\delta_0)/2}),$ by noticing that
        $$
        \begin{aligned}
            &\widehat{\bPsi}_{k,A}^{\top}\widehat{\bOmega}_{y}(k)\widehat{\bPsi}_{k,E}=\widehat{\bPsi}_{k,A}^{\top}(\bA\bX^{\top}+\bB\bZ^{\top}+\bE^{\top})\bD_k(\bX\bA^{\top}+\bZ\bB^{\top}+\bE)\widehat{\bPsi}_{k,E}\\
            =&\widehat{\bPsi}_{k,A}^{\top}(\bA\bX^{\top}+\bB\bZ^{\top}+\bE^{\top})\bD_k(\bX\bA^{\top}+\bZ\bB^{\top}+\bE)(\bI_p-\widehat{\bPsi}_{k,A} \widehat{\bPsi}_{k,A}^{\top})\widehat{\bPsi}_{k,E}\\
            =&\widehat{\bPsi}_{k,A}^{\top}(\bA\bX^{\top}+\bB\bZ^{\top}+\bE^{\top})\bD_k(\bZ\bB^{\top}+\bE)\widehat{\bPsi}_{k,B}\\
            &+\widehat{\bPsi}_{k,A}^{\top}(\bA\bX^{\top}+\bB\bZ^{\top}+\bE^{\top})\bD_k\bX\bA^{\top}(\bA\bA^{\top}-\widehat{\bPsi}_{k,A} \widehat{\bPsi}_{k,A}^{\top})\widehat{\bPsi}_{k,E}.
        \end{aligned}
        $$
    \end{itemize}
    Combining above, we have
    \begin{equation}
        \label{eq.BOA}
        \|\widehat{\bPsi}_{k,B}^{\top}\widehat{\bOmega}_{y}(k)\widehat{\bPsi}_{k,A} \|=O_p(p^{2\delta_{1k}-\delta_0}+n^{-1/2}p^{(1-\delta_0+2\delta_{1k})/2})=o_p(p^{\delta_{1k}}).
    \end{equation}
    Similar procedures can be used to prove the other three terms of \eqref{ssf0}. Thus, we have
    $$
    \|(\bI_p-\widehat{\bPsi}_{k,A} \widehat{\bPsi}_{k,A}^{\top})\widehat{\bOmega}_{y}(k)\widehat{\bPsi}_{k,A} \widehat{\bPsi}_{k,A}^{\top}\widehat{\bOmega}_{y}(k)^{\top}(\bI_p-\widehat{\bPsi}_{k,A} \widehat{\bPsi}_{k,A}^{\top})\|=o_p(p^{2\delta_{1k}}),
    $$
    which concludes \eqref{eq.lem_auto}. The proof is complete.
\end{proof}

Now we are ready to prove Theorem~\ref{advauto}. The proofs of $\hat{\mu}_{k1} \asymp \hat{\mu}_{kr_0} \asymp p^{2\delta_{0}}$ in Theorem~\ref{advauto}(i) and $\big\|\widehat{\bPsi}_{k,A}\widehat{\bPsi}_{k,A}^{\top}-\bA\bA^{\top}\big\|=O_p(n^{-1/2}p^{(1-\delta_{0})/2}+p^{\delta_{1k}-\delta_{0}})$ in Theorem~\ref{advauto}(ii) are similar to those of Theorem 1 in \textcolor{blue}{Zhang et al.} (\textcolor{blue}{2024}), and thus are omitted. It suffices to show that $\hat{\mu}_{k,r_0+1} \asymp \hat{\mu}_{kr} \asymp p^{2\delta_{1k}}$, and $ \hat{\mu}_{k,r+1} \asymp \hat{\mu}_{k,\lfloor c'(p\wedge n)\rfloor} \asymp n^{-2}p^{2}$ with probability tending to 1. By definitions, we have
\begin{eqnarray}
&&\widehat{\bPsi}_{k,B} \widehat{\bN}_{k,B} \widehat{\bPsi}_{k,B}^{\top}+\widehat{\bPsi}_{k,E} \widehat{\bN}_{k,E} \widehat{\bPsi}_{k,E}^{\top}
=(\bI_p-\widehat{\bPsi}_{k,A} \widehat{\bPsi}_{k,A}^{\top})\widehat{\bOmega}_{y}(k) \widehat{\bOmega}_{y}(k)^{\top}(\bI_p-\widehat{\bPsi}_{k,A} \widehat{\bPsi}_{k,A}^{\top})\nonumber\\
&=&(\bI_p-\widehat{\bPsi}_{k,A} \widehat{\bPsi}_{k,A}^{\top})\widehat{\bOmega}_{y}(k) (\bI_p-\widehat{\bPsi}_{k,A} \widehat{\bPsi}_{k,A}^{\top})\widehat{\bOmega}_{y}(k)^{\top}(\bI_p-\widehat{\bPsi}_{k,A} \widehat{\bPsi}_{k,A}^{\top})\label{eq000a1ss01}\\
&&+(\bI_p-\widehat{\bPsi}_{k,A} \widehat{\bPsi}_{k,A}^{\top})\widehat{\bOmega}_{y}(k) \widehat{\bPsi}_{k,A} \widehat{\bPsi}_{k,A}^{\top}\widehat{\bOmega}_{y}(k)^{\top}(\bI_p-\widehat{\bPsi}_{k,A} \widehat{\bPsi}_{k,A}^{\top})\label{eq000a1ss00}.
\end{eqnarray}
Since both \eqref{eq000a1ss01} and \eqref{eq000a1ss00} are non-negative definite, the rank of each term is not larger than $p-r_0$. Recall that $\widehat{\bOmega}_y(k)=(n-k)^{-1}\bY^{\top}\bD_k\bY$. Then we rewrite \eqref{eq000a1ss01} as
\begin{equation}
    \label{eq.thm_auto_1}
    \begin{aligned}
        &(\bI_p-\widehat{\bPsi}_{k,A} \widehat{\bPsi}_{k,A}^{\top})\widehat{\bOmega}_{y}(k) (\bI_p-\widehat{\bPsi}_{k,A} \widehat{\bPsi}_{k,A}^{\top})\widehat{\bOmega}_{y}(k)^{\top}(\bI_p-\widehat{\bPsi}_{k,A} \widehat{\bPsi}_{k,A}^{\top})\\
        =&(n-k)^{-2}(\bI_p-\widehat{\bPsi}_{k,A} \widehat{\bPsi}_{k,A}^{\top})\bY^{\top}\bD_k\bY (\bI_p-\widehat{\bPsi}_{k,A} \widehat{\bPsi}_{k,A}^{\top})(\bY^{\top}\bD_k\bY)^{\top}(\bI_p-\widehat{\bPsi}_{k,A} \widehat{\bPsi}_{k,A}^{\top})\\
        :=&\sum_{i=r_0+1}^p \widebar{\mu}_{ki}\widebar{\bpsi}_{ki} \widebar{\bpsi}_{ki}^{\top}=\widebar{\bPsi}_{k,B} \widebar{\bN}_{k,B} \widebar{\bPsi}_{k,B}^{\top}+\widebar{\bPsi}_{k,E} \widebar{\bN}_{k,E} \widebar{\bPsi}_{k,E}^{\top}
    \end{aligned}
\end{equation}
with $\widebar{\bN}_{k,B}=\diag(\widebar{\mu}_{k,r_0+1},\dots,\widebar{\mu}_{kr}),\widebar{\bN}_{k,E}=\diag(\widebar{\mu}_{k,r+1},\dots,\widebar{\mu}_{kp}),\widebar{\bPsi}_{k,B}=(\widebar{\bpsi}_{k,r_0+1},\dots,\widebar{\bpsi}_{kr}),$ and $\widebar{\bPsi}_{k,E}=(\widebar{\bpsi}_{k,r+1},\dots,\widebar{\bpsi}_{kp})$, where $\widebar{\mu}^{1/2}_{k,r_0+1}\ge \dots\ge\widebar{\mu}^{1/2}_{kp}$ are the singular values of $(n-k)^{-1}(\bI_p-\widehat{\bPsi}_{k,A} \widehat{\bPsi}_{k,A}^{\top})\bY^{\top}\bD_k\bY(\bI_p-\widehat{\bPsi}_{k,A} \widehat{\bPsi}_{k,A}^{\top}).$ 
Define $\bP_{A,B}=(\bA,\bB)\{(\bA,\bB)^{\top}(\bA,\bB)\}^{-1}(\bA,\bB)^{\top}$, and rewrite it as $\bP_{A,B}-\bA\bA^{\top}=\bB_1\bB_1^{\top}$ with $\bB_1^{\top}\bB_1=\bI_{r_1}$ (see Lemma~\ref{advsam0} below for the well-definedness of $\bB_1$).
Note that
$$
\begin{aligned}
    &(\bI_p-\widehat{\bPsi}_{k,A}\widehat{\bPsi}_{k,A}^{\top})\mathbf{Y}^{\top}=(\bI_p-\widehat{\bPsi}_{k,A}\widehat{\bPsi}_{k,A}^{\top})(\bA\bX^{\top}+\bB\bZ^{\top}+\bE^{\top})\\
    =&(\bA\bA^{\top}-\widehat{\bPsi}_{k,A}\widehat{\bPsi}_{k,A}^{\top})(\bA\bX^{\top}+\bB\bZ^{\top})+\bB_1\bB_1^{\top}\bB\bZ^{\top}+(\bI_p-\widehat{\bPsi}_{k,A}\widehat{\bPsi}_{k,A}^{\top})\bE^{\top}.
    \end{aligned}
$$
The orders of each term are calculated as follows:
\begin{itemize}
    \item [(i)] $(n-k)^{-1}\|\bE\|^2=O_p(n^{-1}p)=o_p(p^{\delta_{1k}})$.

    \item [(ii)] $\|(n-k)^{-1}\bB_1\bB_1^{\top}\bB\bZ^{\top}\bD_k\bE\|=O_p(n^{-1/2}p^{(1+\delta_{1})/2})=o_p(p^{\delta_{1k}}).$

    \item [(iii)] $\|(n-k)^{-1}\bB_1\bB_1^{\top}\bB\bZ^{\top}\bD_k\bZ\bB^{\top}\bB_1\bB_1^{\top}\| \asymp \|(n-k)^{-1}\bB_1\bB_1^{\top}\bB\bZ^{\top}\bD_k\bZ\bB^{\top}\bB_1\bB_1^{\top}\|_{\min} \asymp p^{\delta_{1k}}$ with probability tending to 1.

    \item [(iv)] $\|(n-k)^{-1}(\bA\bA^{\top}-\widehat{\bPsi}_{k,A}\widehat{\bPsi}_{k,A}^{\top})(\bA\bX^{\top}+\bB\bZ^{\top})\bD_k(\bX\bA^{\top}+\bZ\bB^{\top})(\bA\bA^{\top}-\widehat{\bPsi}_{k,A}\widehat{\bPsi}_{k,A}^{\top})\|=O_p(n^{-1}p+p^{2\delta_{1k}-\delta_{0}})=o_p(p^{\delta_{1k}})$ by using Theorem~\ref{advauto}(ii).

    \item [(v)] $\|(n-k)^{-1}\bB_1\bB_1^{\top}\bB\bZ^{\top}\bD_k(\bX\bA^{\top}+\bZ\bB^{\top})(\bA\bA^{\top}-\widehat{\bPsi}_{k,A}\widehat{\bPsi}_{k,A}^{\top})\|=O_p(n^{-1/2}p^{(1-\delta_0+2\delta_{1k})/2}+p^{2\delta_{1k}-\delta_{0}})=o_p(p^{\delta_{1k}})$ by using Theorem~\ref{advauto}(ii).

    \item [(vi)] $\|(n-k)^{-1}(\bA\bA^{\top}-\widehat{\bPsi}_{k,A}\widehat{\bPsi}_{k,A}^{\top})(\bA\bX^{\top}+\bB\bZ^{\top})\bD_k\bE\|=O_p(n^{-1}p+n^{-1/2}p^{(1-\delta_0+2\delta_{1k})/2})=o_p(p^{\delta_{1k}})$ by using Theorem~\ref{advauto}(ii).
\end{itemize}
Combining above, we have $\|\widebar{\bPsi}_{k,B}\widebar{\bPsi}_{k,B}^{\top}-\bB_1\bB_1^{\top}\|=O_p(n^{-1/2}p^{(1+\delta_1-2\delta_{1k})/2}+p^{\delta_{1k}-\delta_{0}}),$ and
$$
\widebar{\mu}_{k,r_0+1} \asymp \widebar{\mu}_{kr} \asymp p^{2\delta_{1k}}
$$
with probability tending to 1, which together with Lemmas~\ref{lem.weyl} and~\ref{lem.auto}, imply $\hat{\mu}_{k,r_0+1} \asymp \hat{\mu}_{kr} \asymp p^{2\delta_{1k}}$ with probability tending to 1.
The assertion $\hat{\mu}_{k,r+1} \asymp \hat{\mu}_{k,\lfloor c'(p\wedge n)\rfloor} \asymp n^{-2}p^{2}$ with probability tending to 1 can be proved by using similar procedures. 
The proof is complete.

\subsection{Proof of Corollary \ref{coro.auto}}
(i) Note that $\widehat{\bOmega}_{y}(k)=(\bI_p-\bA\bA^{\top})\widehat{\bOmega}_{y}(k)+\bA\bA^{\top}\widehat{\bOmega}_{y}(k)$ with $\|\bA^{\top}\widehat{\bOmega}_{y}(k)\| \asymp \|\bA^{\top}\widehat{\bOmega}_{y}(k)\|_{\min} \asymp p^{\delta_0}$ with probability tending to 1. Firstly, by $\|\Omega_z(k)\|=O(n^{-1/2}p^{(1+\delta_1)/2})$, we have
\begin{equation}
    \label{eq.coro_1}
    \begin{aligned}
        &\|(\bI_p-\bA\bA^{\top})\widehat{\bOmega}_{y}(k)\|=\|(n-k)^{-1}(\bI_p-\bA\bA^{\top})(\bB\bZ^{\top}+\bE^{\top})\bD_k(\bX\bA^{\top}+\bZ\bB^{\top}+\bE)\|\\
    \le&\|(n-k)^{-1}(\bB\bZ^{\top}+\bE^{\top})\bD_k(\bX\bA^{\top}+\bZ\bB^{\top}+\bE)\| \\
    \le& \|(n-k)^{-1}(\bB\bZ^{\top}+\bE^{\top})\bD_k(\bZ\bB^{\top}+\bE)\|+\|(n-k)^{-1}(\bB\bZ^{\top}+\bE^{\top})\bD_k \bX\bA^{\top})\|\\
    =&O_p(n^{-1/2}p^{(1+\delta_0)/2}+n^{-1/2}p^{(1+\delta_1)/2}+n^{-1/2}p^{(\delta_0+\delta_1)/2}+n^{-1}p)\\
    =&O_p(n^{-1/2}p^{(1+\delta_0)/2})=o_p(p^{\delta_0}),
    \end{aligned}
\end{equation}
which implies $\hat{\mu}_{k1} \asymp \hat{\mu}_{kr_0} \asymp p^{2\delta_{0}}$ with probability tending to 1. 

Note that $\hat{\mu}_{k,r_0+1}^{1/2}$ is the largest singular value of $(\bI_p-\widehat{\bPsi}_{k,A}\widehat{\bPsi}_{k,A}^{\top})\widehat{\bOmega}_{y}(k)$. It follows that
$$
\begin{aligned}
    &(\bI_p-\widehat{\bPsi}_{k,A} \widehat{\bPsi}_{k,A}^{\top})\widehat{\bOmega}_{y}(k)(\bI_p-\widehat{\bPsi}_{k,A} \widehat{\bPsi}_{k,A}^{\top})\\
=&(n-k)^{-1}(\bI_p-\widehat{\bPsi}_{k,A} \widehat{\bPsi}_{k,A}^{\top})(\bA\bX^{\top}+\bB\bZ^{\top}+\bE^{\top})\bD_k(\bX\bA^{\top}+\bZ\bB^{\top}+\bE)(\bI_p-\widehat{\bPsi}_{k,A} \widehat{\bPsi}_{k,A}^{\top})\\
=&(n-k)^{-1}(\bI_p-\widehat{\bPsi}_{k,A} \widehat{\bPsi}_{k,A}^{\top})(\bB\bZ^{\top}+\bE^{\top})\bD_k(\bZ\bB^{\top}+\bE)(\bI_p-\widehat{\bPsi}_{k,A} \widehat{\bPsi}_{k,A}^{\top})\\
&+(n-k)^{-1}(\bI_p-\widehat{\bPsi}_{k,A} \widehat{\bPsi}_{k,A}^{\top})(\bB\bZ^{\top}+\bE^{\top})\bD_k\bX\bA^{\top}(\bA\bA^{\top}-\widehat{\bPsi}_{k,A} \widehat{\bPsi}_{k,A}^{\top})\\
&+(n-k)^{-1}(\bA\bA^{\top}-\widehat{\bPsi}_{k,A} \widehat{\bPsi}_{k,A}^{\top})\bA\bX^{\top}\bD_k(\bZ\bB^{\top}+\bE)(\bI_p-\widehat{\bPsi}_{k,A} \widehat{\bPsi}_{k,A}^{\top})\\
&+(n-k)^{-1}(\bA\bA^{\top}-\widehat{\bPsi}_{k,A} \widehat{\bPsi}_{k,A}^{\top})\bA\bX^{\top}\bD_k\bX\bA^{\top}(\bA\bA^{\top}-\widehat{\bPsi}_{k,A} \widehat{\bPsi}_{k,A}^{\top})\\
=&O_p(n^{-1/2}p^{(1+\delta_1)/2}+n^{-1}p)=O_p(n^{-1/2}p^{(1+\delta_1)/2}).
\end{aligned}
$$
Similarly, we can show that
$\|(\bI_p-\widehat{\bPsi}_{k,A} \widehat{\bPsi}_{k,A}^{\top})\widehat{\bOmega}_{y}(k)\widehat{\bPsi}_{k,A} \widehat{\bPsi}_{k,A}^{\top}\|=O_p(n^{-1/2}p^{(1+\delta_1)/2}).$ Thus, we have $\hat{\mu}_{k,r_0+1}=O_p(n^{-1}p^{1+\delta_1})$. 

(ii) Noticing that $\|(\bI_p-\bA\bA^{\top})\widehat{\bOmega}_{y}(k)\widehat{\bOmega}_{y}(k)^{\top}\|\le \|(n-k)^{-1}(\bI_p-\bA\bA^{\top})(\bB\bZ^{\top}+\bE^{\top})\bD_k(\bX\bA^{\top}+\bZ\bB^{\top}+\bE)\|\|\widehat{\bOmega}_{y}(k)\|=O_p(n^{-1/2}p^{(1+3\delta_0)/2})$ by \eqref{eq.coro_1}, which, together with the fact $\hat{\mu}_{k1} \asymp \hat{\mu}_{kr_0} \asymp p^{2\delta_{0}}$ with probability tending to 1, derives that $\big\|\widehat{\bPsi}_{k,A}\widehat{\bPsi}_{k,A}^{\top}-\bA\bA^{\top}\big\|=O_p(n^{-1/2}p^{(1-\delta_{0})/2})$. The proof is complete.

\subsection{Proof of Theorem \ref{advour}}
Let $\hat{\lambda}_{k1} \ge  \cdots\ge \hat{\lambda}_{kq} \ge 0$ be the eigenvalues of $\widehat{\bOmega}_{y}(k)\widehat{\bW} \widehat{\bOmega}_{y}(k)^{\top}$, with the corresponding spectral decomposition
$$
\widehat{\bOmega}_{y}(k)\widehat{\bW}\widehat{\bOmega}_{y}(k)^{\top}=\sum_{i=1}^{q} \hat{\lambda}_{ki}\widehat{\bphi}_{ki}\widehat{\bphi}_{ki}^{\top}
=\widehat{\bPhi}_{k,A}\widehat{\bLambda}_{k,A}\widehat{\bPhi}_{k,A}^{\top}+\widehat{\bPhi}_{k,B}\widehat{\bLambda}_{k,B}\widehat{\bPhi}_{k,B}^{\top}+\widehat{\bPhi}_{k,E}\widehat{\bLambda}_{k,E}\widehat{\bPhi}_{k,E}^{\top},
$$
where $\widehat{\bLambda}_{k,A}=\diag(\hat{\lambda}_{k1},\dots,\hat{\lambda}_{kr_0}),\widehat{\bLambda}_{k,B}=\diag(\hat{\lambda}_{k,r_0+1},\dots,\hat{\lambda}_{kr}),\widehat{\bLambda}_{k,E}=\diag(\hat{\lambda}_{k,r+1},\dots,\hat{\lambda}_{kq})$, $\widehat{\bPhi}_{k,A}=(\widehat{\bphi}_{k1},\dots,\widehat{\bphi}_{kr_0}),\widehat{\bPhi}_{k,B}=(\widehat{\bphi}_{k,r_0+1},\dots,\widehat{\bphi}_{kr}),$ and $\widehat{\bPhi}_{k,E}=(\widehat{\bphi}_{k,r+1},\dots,\widehat{\bphi}_{kq})$.
To show Theorem~\ref{advour}, we first introduce the following technical lemmas.

\begin{lemma}\label{advsam0}
Define $\bP_{A,B}=(\bA,\bB)\{(\bA,\bB)^{\top}(\bA,\bB)\}^{-1}(\bA,\bB)^{\top}$. Under Conditions~\ref{cond.A_a}--\ref{cond.E_a}, 
\begin{align}
& \|\widehat{\bXi}_{A}\widehat{\bXi}_{A}^{\top}+\widehat{\bXi}_{B}\widehat{\bXi}_{B}^{\top}-\bP_{A,B}\|=O_p(n^{-1/2}p^{(1-\delta_1)/2}). \label{eq3}
\end{align}
Moreover, we rewrite $\bP_{A,B}-\bA\bA^{\top}=\bB_1\bB_1^{\top}$ with $\bB_1^{\top}\bB_1=\bI_{r_1}$.
Then, for any $i\in(r,q]$,
\begin{equation}
    \label{eq4ca}
    \|\widehat{\bxi}_{i}^{\top}\bA\|^2=O_p(n^{-1}p^{1-\delta_0}),~~\|\widehat{\bxi}_{i}^{\top}\bB_1\|^2=O_p(n^{-1}p^{1-\delta_1}).
\end{equation}
\end{lemma}
\begin{proof}
    Note that Condition \ref{cond.A_a} and Lemma~\ref{lem.weyl} together ensure that $\|\bB^{\top}(\bI_p-\bA\bA^{\top})\|_{\min}\ge\Vert\bB\Vert_{\min}-\|\bA\bA^{\top}\bB\|>0,$ which implies that $\bB_1$ is well defined with $\rank(\bB_1)=r_1$ and $\bB_1^{\top}\bB_1=\bI_{r_1}.$
    The proof of \eqref{eq3} is similar to that of our Theorem~\ref{advsam1}(ii), or Theorem 1 in \textcolor{blue}{Zhang et al.} (\textcolor{blue}{2024}), and thus is omitted. It suffices to prove \eqref{eq4ca}. 
    In the following, we consider two cases, i.e., $\delta_0=\delta_1$ and $\delta_0>\delta_1$, separately.
    On the one hand, when $\delta_0=\delta_1$, we can rewrite $\bY$ as $\bY=\bT\bC^{\top}+\bE$ with $\bC\bC^{\top}=\bP_{A,B}$ and $\|n^{-1}\bT^{\top}\bT\|\asymp \|n^{-1}\bT^{\top}\bT\|_{\min}=p^{\delta_0}$ with probability tending to 1.
    For each $i\in(r,q]$, with probability tending to 1,
    $$
    n^{-1}p \asymp \hat{\theta}_i=n^{-1}\widehat{\bxi}_i^{\top}(\bC\bT^{\top}+\bE^{\top})(\bT\bC^{\top}+\bE)\widehat{\bxi}_i
    =n^{-1}(\widehat{\bxi}_i^{\top}\bC)(\bT^{\top}\bT)(\bC^{\top}\widehat{\bxi}_i)+O_p(n^{-1}p),
    $$
    which implies that $\|\widehat{\bxi}_i^{\top}\bC\|^2=O_p(n^{-1}p^{1-\delta_0})$, and thus \eqref{eq4ca} holds.
    
    On the other hand, when $\delta_0>\delta_1$, define $\dot{\bX}^{\top}=\bX^{\top}+\bA^{\top}\bB\bZ^{\top}$ and $\dot{\bZ}^{\top}=\bB_1^{\top}\bB\bZ^{\top}$.
    Then $\bY=\dot{\bX}\bA^{\top}+\dot{\bZ}\bB_1^{\top}+\bE$.
    For each $i\in(r,q]$, with probability tending to 1,
    $$
    \begin{aligned}
    n^{-1}p \asymp \hat{\theta}_i=&n^{-1}\widehat{\bxi}_i^{\top}(\bA\dot{\bX}^{\top}+\bB_1\dot{\bZ}^{\top}+\bE^{\top})(\dot{\bX}\bA^{\top}+\dot{\bZ}\bB_1^{\top}+\bE)\widehat{\bxi}_i\\
    =&n^{-1}(\widehat{\bxi}_i^{\top}\bA)(\dot{\bX}^{\top}\dot{\bX})(\bA^{\top}\widehat{\bxi}_i)+2n^{-1}(\widehat{\bxi}_i^{\top}\bA)(\dot{\bX}^{\top}\dot{\bZ})(\bB_1^{\top}\widehat{\bxi}_i)+n^{-1}(\widehat{\bxi}_i^{\top}\bB_1)(\dot{\bZ}^{\top}\dot{\bZ})(\bB_1^{\top}\widehat{\bxi}_i)\\
    &+2n^{-1}\widehat{\bxi}_i^{\top}(\bA\dot{\bX}^{\top}+\bB_1\dot{\bZ}^{\top})\bE\widehat{\bxi}_i+O_p(n^{-1}p),
    \end{aligned}
    $$
    which implies that $\Vert n^{-1}\widehat{\bxi}_i^{\top}(\bA\dot{\bX}^{\top}+\bB_1\dot{\bZ}^{\top})\bE\widehat{\bxi}_i\Vert=O_p(n^{-1}p)$. Since $\{\bx_t\}_{t\in\eZ},\{\bz_t\}_{t\in\eZ}$ and $\{\be_t\}_{t\in\eZ}$ are mutually uncorrelated, by Condition~\ref{cond.cov_a}, it follows that
    $\|n^{-1}\bX^{\top}\bZ\|=O_p(n^{-1/2}p^{\delta_1/2+\delta_0/2})=o_p(p^{\delta_1}),$
    and thus $$\|n^{-1}\dot{\bX}^{\top}\dot{\bZ}\|=\|n^{-1}\bX^{\top}\bZ\bB^{\top}\bB_1+n^{-1}\bA^{\top}\bB\bZ^{\top}\bZ\bB^{\top}\bB_1\|=O_p(p^{\delta_1}).$$
    Then, with probability tending to 1, we have
    $$
    \begin{aligned}
        &n^{-1}(\widehat{\bxi}_i^{\top}\bA)(\dot{\bX}^{\top}\dot{\bX})(\bA^{\top}\widehat{\bxi}_i) \asymp p^{\delta_0}\|\widehat{\bxi}_i^{\top}\bA\|^2,\\
        &n^{-1}(\widehat{\bxi}_i^{\top}\bA)(\dot{\bX}^{\top}\dot{\bZ})(\bB_1^{\top}\widehat{\bxi}_i)=O_p(p^{\delta_1}\|\widehat{\bxi}_i^{\top}\bA\|\|\bB_1^{\top}\widehat{\bxi}_i\|),\\
        &n^{-1}(\widehat{\bxi}_i^{\top}\bB_1)(\dot{\bZ}^{\top}\dot{\bZ})(\bB_1^{\top}\widehat{\bxi}_i) \asymp p^{\delta_1}\|\widehat{\bxi}_i^{\top}\bB_1\|^2,\\
        &n^{-1}\widehat{\bxi}_i^{\top}(\bA\dot{\bX}^{\top}+\bB_1\dot{\bZ}^{\top})\bE\widehat{\bxi}_i=O_p(n^{-1}p).
    \end{aligned}
    $$
    For analysis, we consider the following two cases separately: \\
    (i) If $\|\widehat{\bxi}_i^{\top}\bA\|=o_p(\|\widehat{\bxi}_i^{\top}\bB_1\|)$, we have $n^{-1}(\widehat{\bxi}_i^{\top}\bA)(\dot{\bX}^{\top}\dot{\bZ})(\bB_1^{\top}\widehat{\bxi}_i)=o_p(p^{\delta_1}\|\widehat{\bxi}_i^{\top}\bB_1\|^2)$. Then, $p^{\delta_0}\|\widehat{\bxi}_i^{\top}\bA\|^2+p^{\delta_1}\|\widehat{\bxi}_i^{\top}\bB_1\|^2=O_p(n^{-1}p)$, which implies that \eqref{eq4ca}. \\
    (ii) If $\|\widehat{\bxi}_i^{\top}\bB_1\|=O_p(\|\widehat{\bxi}_i^{\top}\bA\|)$  and $\delta_0>\delta_1$, we have $
        n^{-1}(\widehat{\bxi}_i^{\top}\bA)(\dot{\bX}^{\top}\dot{\bZ})(\bB_1^{\top}\widehat{\bxi}_i)=o_p(p^{\delta_0}\|\widehat{\bxi}_i^{\top}\bA\|^2),$ and $
        n^{-1}(\widehat{\bxi}_i^{\top}\bB_1)(\dot{\bZ}^{\top}\dot{\bZ})(\bB_1^{\top}\widehat{\bxi}_i)=o_p(p^{\delta_0}\|\widehat{\bxi}_i^{\top}\bA\|^2).$
    Then, we have $p^{\delta_0}\|\widehat{\bxi}_i^{\top}\bA\|^2=O_p(n^{-1}p)$, which implies that \eqref{eq4ca}. The proof is complete. 
\end{proof}

\begin{lemma}
\label{lem.S1}
Under Conditions~\ref{cond.A_a}--\ref{cond.E_a} and $r<q\le c'(p\wedge n)$, we have
\begin{eqnarray}\label{middle1}
\|\bA^{\top}\widehat{\bW}\bA\|=O_p(p^{-\delta_0}),~~\|\bB_1^{\top}\widehat{\bW}\bA\|=O_p(p^{-(\delta_0+\delta_1)/2}),~~\|\bB_1^{\top}\widehat{\bW}\bB_1\|=O_p(p^{-\delta_1}),
\end{eqnarray}
and
\begin{eqnarray}\label{middle2}
\|\bA^{\top}\widehat{\bW}(\bI_p-\bP_{A,B})\|=O_p(n^{1/2}p^{-(1+\delta_0)/2}),~~\|\bB_1^{\top}\widehat{\bW}(\bI_p-\bP_{A,B})\|=O_p(n^{1/2}p^{-(1+\delta_1)/2}),
\end{eqnarray}
where $\bB_1$ and $\bP_{A,B}$ are defined in Lemma~\ref{advsam0}.
\end{lemma}
\begin{proof}
Recall that $\widehat{\bTheta}_{A}=\diag(\hat{\theta}_{1},\dots,\hat{\theta}_{r_0}),\widehat{\bTheta}_{B}=\diag(\hat{\theta}_{r_0+1},\dots,\hat{\theta}_{r}),\widehat{\bTheta}_{q-r}=\diag(\hat{\theta}_{r+1},\dots,\hat{\theta}_{q})$, $\widehat{\bXi}_{A}=(\widehat{\bxi}_1,\dots,\widehat{\bxi}_{r_0}),\widehat{\bXi}_{B}=(\widehat{\bxi}_{r_0+1},\dots,\widehat{\bxi}_{r})$, and $\widehat{\bXi}_{q-r}=(\widehat{\bxi}_{r+1},\dots,\widehat{\bxi}_{q})$. In the following, we consider two cases, i.e., $\delta_0=\delta_1$ and $\delta_0>\delta_1$, separately. 

On the one hand, when $\delta_0>\delta_1$, we have $\|\widehat{\bXi}_B^{\top}\bA\|=O_p(n^{-1/2}p^{(1-\delta_0)/2}+p^{\delta_1-\delta_0})$ by Theorem~\ref{advsam1}(ii), and $\|\widehat{\bTheta}_{q-r}^{-1}\|=O_p(np^{-1})$ by Theorem~\ref{advsam1}(i) and the condition $q\le c'(p\wedge n)$. Then, there exist some constants $C_1,C_2>0$ such that, with probability tending to 1,
$$
\begin{aligned}
    \|\bA^{\top}\widehat{\bW}\bA\|\le&\|\bA^{\top}\widehat{\bXi}_{A}\widehat{\bTheta}_{A}^{-1}\widehat{\bXi}_{A}^{\top}\bA\|+\|\bA^{\top}\widehat{\bXi}_{B}\widehat{\bTheta}_{B}^{-1}\widehat{\bXi}_{B}^{\top}\bA\|+\|\bA^{\top}\widehat{\bXi}_{q-r}\widehat{\bTheta}_{q-r}^{-1}\widehat{\bXi}_{q-r}^{\top}\bA\|\\
    \le&C_1\big(p^{-\delta_0}\|\widehat{\bXi}_{A}^{\top}\bA\|^2+p^{-\delta_1}\|\widehat{\bXi}_B^{\top}\bA\|^2+np^{-1}\|\bA^{\top}\widehat{\bXi}_{q-r}\|^2\big)\le C_2p^{-\delta_0},\\
    \|\bA^{\top}\widehat{\bW}\bB_1\|\le&\|\bA^{\top}\widehat{\bXi}_{A}\widehat{\bTheta}_{A}^{-1}\widehat{\bXi}_{A}^{\top}\bB_1\|+\|\bA^{\top}\widehat{\bXi}_{B}\widehat{\bTheta}_{B}^{-1}\widehat{\bXi}_{B}^{\top}\bB_1\|+\|\bA^{\top}\widehat{\bXi}_{q-r}\widehat{\bTheta}_{q-r}^{-1}\widehat{\bXi}_{q-r}^{\top}\bB_1\|\\
    \le&C_1\big(p^{-\delta_0}\|\widehat{\bXi}_{A}^{\top}\bB_1\|+p^{-\delta_1}\|\widehat{\bXi}_B^{\top}\bA\|+np^{-1}\|\bB_1^{\top}\widehat{\bXi}_{q-r}\|\|\widehat{\bXi}_{q-r}^{\top}\bA\|\big)\le C_2p^{-(\delta_0+\delta_1)/2},\\
    \|\bB_1^{\top}\widehat{\bW}\bB_1\|\le&\|\bB_1^{\top}\widehat{\bXi}_{A}\widehat{\bTheta}_{A}^{-1}\widehat{\bXi}_{A}^{\top}\bB_1\|+\|\bB_1^{\top}\widehat{\bXi}_{B}\widehat{\bTheta}_{B}^{-1}\widehat{\bXi}_{B}^{\top}\bB_1\|+\|\bB_1^{\top}\widehat{\bXi}_{q-r}\widehat{\bTheta}_{q-r}^{-1}\widehat{\bXi}_{q-r}^{\top}\bB_1\|\\
    \le&C_1\big(p^{-\delta_0}\|\widehat{\bXi}_{A}^{\top}\bB_1\|^2+p^{-\delta_1}\|\widehat{\bXi}_{B}^{\top}\bB_1\|^2+np^{-1}\|\bB_1^{\top}\widehat{\bXi}_{q-r}\|^2\big)\le C_2p^{-\delta_1},
\end{aligned}
$$
which together imply \eqref{middle1}. Consequently, $\|\bP_{A,B}\widehat{\bW}\bP_{A,B}\|\le\|\bA^{\top}\widehat{\bW}\bA\|+\|\bB_1^{\top}\widehat{\bW}\bB_1\|+2\|\bA^{\top}\widehat{\bW}\bB_1\|=O_p(p^{-\delta_1}).$ Moreover, by \eqref{eq4ca}, there exist constants $C_3,C_4>0$ such that
$$
\begin{aligned}
    \|\bA^{\top}\widehat{\bW}(\bI_p-\bP_{A,B})\|\le&\|\bA^{\top}\widehat{\bXi}_{A}\widehat{\bTheta}_{A}^{-1}\widehat{\bXi}_{A}^{\top}(\bI_p-\bP_{A,B})\|+\|\bA^{\top}\widehat{\bXi}_{B}\widehat{\bTheta}_{B}^{-1}\widehat{\bXi}_{B}^{\top}(\bI_p-\bP_{A,B})\|\\
    &+\|\bA^{\top}\widehat{\bXi}_{q-r}\widehat{\bTheta}_{q-r}^{-1}\widehat{\bXi}_{q-r}^{\top}(\bI_p-\bP_{A,B})\|\\
    \le&C_3\big(p^{-\delta_0}\|\widehat{\bXi}_{A}^{\top}(\bI_p-\bP_{A,B})\|+p^{-\delta_1}\|\bA^{\top}\widehat{\bXi}_B\|+np^{-1}\|\bA^{\top}\widehat{\bXi}_{q-r}\|\big)\\
    \le&C_4n^{1/2}p^{-(1+\delta_0)/2},\\
    \|\bB_1^{\top}\widehat{\bW}(\bI_p-\bP_{A,B})\|
    \le&\|\bB_1^{\top}\widehat{\bXi}_{A}\widehat{\bTheta}_{A}^{-1}\widehat{\bXi}_{A}^{\top}(\bI_p-\bP_{A,B})\|+\|\bB_1^{\top}\widehat{\bXi}_B\widehat{\bTheta}_{B}^{-1}\widehat{\bXi}_B^{\top}(\bI_p-\bP_{A,B})\|\\
    &+\|\bB_1^{\top}\widehat{\bXi}_{q-r}\widehat{\bTheta}_{q-r}^{-1}\widehat{\bXi}_{q-r}^{\top}(\bI_p-\bP_{A,B})\|\\
    \le&C_3\big(p^{-\delta_0}\|\widehat{\bXi}_{A}^{\top}(\bI_p-\bP_{A,B})\|+p^{-\delta_1}\|\widehat{\bXi}_B^{\top}(\bI_p-\bP_{A,B})\|+np^{-1}\|\bB_1^{\top}\widehat{\bXi}_{q-r}\|\big)\\
    \le&C_4n^{1/2}p^{-(1+\delta_1)/2},
\end{aligned}
$$
which together imply \eqref{middle2}. 

On the other hand, when $\delta_0=\delta_1,$ we can rewrite $\bY$ as $\bY=\bT\bC^{\top}+\bE$ with $\bC\bC^{\top}=\bP_{A,B}$ and $\|n^{-1}\bT^{\top}\bT\|\asymp \|n^{-1}\bT^{\top}\bT\|_{\min}\asymp p^{\delta_0}$ with probability tending to 1. Since $\bB_1$ is proved to be full-ranked in Lemma~\ref{advsam0}, some standard arguments similar to the proof of Theorem~\ref{advsam1}(i) can show that $\bC\in\eR^{p\times r},\bT\in\eR^{n\times r},\hat{\theta}_1^{1/2}\ge\cdots\ge\hat{\theta}_r^{1/2}\ge0$ are the $r$ largest singular values of $n^{-1}\bT\bC^{\top}$ with $\hat{\theta}_1\asymp\hat{\theta}_r\asymp p^{\delta_0}$, and also $\hat{\theta}_{r+1}\asymp\hat{\theta}_q\asymp n^{-1}p$ with probability tending to 1. Thus, $\|\bA^{\top}\widehat{\bW}\bA\|\le C_1(p^{-\delta_0}\|\widehat{\bXi}_{A}^{\top}\bA\|^2+p^{-\delta_0}\|\widehat{\bXi}_B^{\top}\bA\|^2+np^{-1}\|\bA^{\top}\widehat{\bXi}_{q-r}\|^2)\le C_2p^{-\delta_0},\|\bA^{\top}\widehat{\bW}(\bI_p-\bP_{A,B})\|\le C_3(p^{-\delta_0}\|\widehat{\bXi}_{A}^{\top}(\bI_p-\bP_{A,B})\|+p^{-\delta_0}\|\bA^{\top}\widehat{\bXi}_B\|+np^{-1}\|\bA^{\top}\widehat{\bXi}_{q-r}\|)\le C_4n^{1/2}p^{-(1+\delta_0)/2}$, and the other terms can be similarly proved. The proof is complete.
\end{proof}

\begin{lemma}
    \label{lem.our_1}
    Under the conditions of Theorem \ref{advour}, we have
    \begin{equation}
        \label{yux1}
        \|\widehat{\bOmega}_{y}(k)(\bI_p-\bP_{A,B}) \widehat{\bXi}_q\widehat{\bTheta}_{q}^{-1/2}\|=O_p(n^{-1/2}p^{1/2})=o_p(p^{\delta_{1k}-\delta_1/2}),
    \end{equation}
    where $\bP_{A,B}$ is defined in Lemma~\ref{advsam0}, and $\widehat{\bXi}_q$ and $\widehat{\bTheta}_{q}$ are defined in \eqref{eq1}.
\end{lemma}
\begin{proof}
    Notice that the entries of $\bX^{\top}\bD_k\bE(\bI_p-\bP_{A,B}) \widehat{\bXi}_q$ and
    $\bZ^{\top}\bD_k\bE(\bI_p-\bP_{A,B}) \widehat{\bXi}_q$ are all $O_p(n^{1/2}p^{\delta_0/2})$ and $O_p(n^{1/2}p^{\delta_1/2})$, respectively. Thus,
    $$
    \begin{aligned}
        &\|(n-k)^{-1}(\bA\bX^{\top}+\bB\bZ^{\top}+\bE^{\top})\bD_k\bE(\bI_p-\bP_{A,B}) \widehat{\bXi}_q\|\\
        \le &\|(n-k)^{-1}\bA\bX^{\top}\bD_k\bE(\bI_p-\bP_{A,B}) \widehat{\bXi}_q\|+\|(n-k)^{-1}\bB\bZ^{\top}\bD_k\bE(\bI_p-\bP_{A,B}) \widehat{\bXi}_q\|\\
        &+\|(n-k)^{-1}\bE^{\top}\bD_k\bE(\bI_p-\bP_{A,B}) \widehat{\bXi}_q\|\\
        =&O_p(q^{1/2}n^{-1/2}p^{\delta_0/2})+O_p(q^{1/2}n^{-1/2}p^{\delta_1/2})+O_p(n^{-1}p)=O_p(n^{-1}p),
    \end{aligned}
    $$
    where the last line holds since $q=O(n^{-1}p^{2-\delta_0})$ and $\delta_0\ge\delta_1$.
    This, together with the fact $\|\widehat{\bTheta}_{q}^{-1/2}\|=O_p(n^{1/2}p^{-1/2})$, which can be shown to hold for both cases, i.e., $\delta_0=\delta_1$ and $\delta_0>\delta_1$ in a similar fashion to the proof of Theorem~\ref{advsam1}(i), implies that
    $$
        \begin{aligned}
            \|\widehat{\bOmega}_{y}(k)(\bI_p-\bP_{A,B}) \widehat{\bXi}_q\widehat{\bTheta}_{q}^{-1/2}\|
            =&\|(n-k)^{-1}(\bA\bX^{\top}+\bB\bZ^{\top}+\bE^{\top})\bD_k\bE(\bI_p-\bP_{A,B}) \widehat{\bXi}_q\widehat{\bTheta}_{q}^{-1/2}\|\\
            \le &\|\widehat{\bTheta}_{q}^{-1/2}\|\|(n-k)^{-1}(\bA\bX^{\top}+\bB\bZ^{\top}+\bE^{\top})\bD_k\bE(\bI_p-\bP_{A,B}) \widehat{\bXi}_q\|\\
            =&O_p(n^{1/2}p^{-1/2})\times O_p(n^{-1}p)=O_p(n^{-1/2}p^{1/2})=o_p(p^{\delta_{1k}-\delta_1/2}),
        \end{aligned}
    $$
    which concludes \eqref{yux1}. The proof is complete.
\end{proof}

\begin{lemma}
    \label{lem.our_2}
    Under the conditions of Theorem \ref{advour}, for $k\in[m],$
    \begin{equation}
        \label{eq.lem_our_2}
        \Vert\widehat{\bPsi}_{k,B}\widehat{\bPsi}_{k,B}^{\top}-\bB_1\bB_1^{\top}\Vert=O_p(n^{-1/2}p^{(1+\delta_1-2\delta_{1k})/2}+p^{\delta_{1k}-\delta_0})=o_p(1),
    \end{equation}
    where $\bB_1$ is defined in Lemma~\ref{advsam0}.
\end{lemma}
\begin{proof}
    By \eqref{eq.BOA}, it can be shown that
    $$
    \begin{aligned}
        &\Vert\widehat{\bPsi}_{k,B}\widehat{\bPsi}_{k,B}^{\top}\widehat{\bOmega}_{y}(k)\widehat{\bPsi}_{k,A}\widehat{\bPsi}_{k,A}^{\top}\widehat{\bOmega}_{y}(k)^{\top}\widehat{\bPsi}_{k,E}\widehat{\bPsi}_{k,E}^{\top}\Vert\\
        =&O_p(p^{2\delta_{1k}-\delta_0}+n^{-1/2}p^{(1-\delta_0+2\delta_{1k})/2})\times O_p(p^{\delta_{1k}}+n^{-1/2}p^{(1+\delta_0)/2})\\
        =&O_p(n^{-1/2}p^{(1-\delta_0+4\delta_{1k})/2}+n^{-1}p^{1+\delta_{1k}}+p^{3\delta_{1k}-\delta_0}),
    \end{aligned}
    $$
    which, together with Lemma~\ref{lem.sintheta}, the second assertion in Theorem~\ref{advauto}(i) and the orders calculated in Lemma~\ref{lem.auto}, implies that $\Vert\widehat{\bPsi}_{k,B}\widehat{\bPsi}_{k,B}^{\top}-\widebar{\bPsi}_{k,B}\widebar{\bPsi}_{k,B}^{\top}\Vert=O_p(n^{-1/2}p^{(1-\delta_0)/2}+n^{-1}p^{1-\delta_{1k}}+p^{\delta_{1k}-\delta_0})=O_p(n^{-1/2}p^{(1+\delta_1-2\delta_{1k})/2}+p^{\delta_{1k}-\delta_0}),$ where $\widebar{\bPsi}_{k,B}$ is defined in \eqref{eq.thm_auto_1}. By the proof of Theorem~\ref{advauto}, we have $\|\widebar{\bPsi}_{k,B}\widebar{\bPsi}_{k,B}^{\top}-\bB_1\bB_1^{\top}\|=O_p(n^{-1/2}p^{(1+\delta_1-2\delta_{1k})/2}+p^{\delta_{1k}-\delta_{0}}),$ which concludes \eqref{eq.lem_our_2}. The proof is complete. 
\end{proof}

\begin{lemma}
    \label{lem.our_3}
    Under the conditions of Theorem \ref{advour}, for $k\in[m],$
    \begin{equation}
        \label{eq.lem_our_3}
        \|(\bI_p-\bA\bA^{\top})\widehat{\bOmega}_{y}(k)\widehat{\bW}\widehat{\bOmega}_{y}(k)^{\top}\|=O_p\big(n^{-1/2}p^{(1+\delta_0)/2}+p^{\delta_{1k}+(\delta_0-\delta_1)/2})=o_p(p^{\delta_0}\big).
    \end{equation}
\end{lemma}
\begin{proof}
    Consider the following decomposition:
    {\small $$
    \begin{aligned}
        &(\bI_p-\bA\bA^{\top})\widehat{\bOmega}_{y}(k)\widehat{\bW}\widehat{\bOmega}_{y}(k)^{\top} \\
        =&(\bI_p-\bA\bA^{\top})\widehat{\bOmega}_{y}(k)\bA\bA^{\top}\widehat{\bW}\bA\bA^{\top}\widehat{\bOmega}_{y}(k)^{\top}+(\bI_p-\bA\bA^{\top})\widehat{\bOmega}_{y}(k)\bB_1\bB_1^{\top}\widehat{\bW}\bB_1\bB_1^{\top}\widehat{\bOmega}_{y}(k)^{\top}\\
        &+(\bI_p-\bA\bA^{\top})\widehat{\bOmega}_{y}(k)\bA\bA^{\top}\widehat{\bW}\bB_1\bB_1^{\top}\widehat{\bOmega}_{y}(k)^{\top}+(\bI_p-\bA\bA^{\top})\widehat{\bOmega}_{y}(k)\bB_1\bB_1^{\top}\widehat{\bW}\bA\bA^{\top}\widehat{\bOmega}_{y}(k)^{\top}\\
        &+(\bI_p-\bA\bA^{\top})\widehat{\bOmega}_{y}(k)\bA\bA^{\top}\widehat{\bW}(\bI_p-\bP_{A,B})\widehat{\bOmega}_{y}(k)^{\top}+(\bI_p-\bA\bA^{\top})\widehat{\bOmega}_{y}(k)\bB_1\bB_1^{\top}\widehat{\bW}(\bI_p-\bP_{A,B})\widehat{\bOmega}_{y}(k)^{\top}\\
        &+(\bI_p-\bA\bA^{\top})\widehat{\bOmega}_{y}(k)(\bI_p-\bP_{A,B})\widehat{\bW}(\bI_p-\bP_{A,B})\widehat{\bOmega}_{y}(k)^{\top}+(\bI_p-\bA\bA^{\top})\widehat{\bOmega}_{y}(k)(\bI_p-\bP_{A,B})\widehat{\bW}\bB_1\bB_1^{\top}\widehat{\bOmega}_{y}(k)^{\top}\\
        &+(\bI_p-\bA\bA^{\top})\widehat{\bOmega}_{y}(k)(\bI_p-\bP_{A,B})\widehat{\bW}\bA\bA^{\top}\widehat{\bOmega}_{y}(k)^{\top}.
    \end{aligned}
    $$}
    The orders of each term are calculated as follows:
    $$
    \begin{aligned}
        & \|(\bI_p-\bA\bA^{\top})\widehat{\bOmega}_{y}(k)\bA\| = \|(n-k)^{-1}(\bI_p-\bA\bA^{\top})(\bB\bZ^{\top}+\bE^{\top})\bD_k(\bX\bA^{\top}+\bZ\bB^{\top}+\bE)\bA\|=O_p(p^{\delta_{1k}}),\\
        & \|\bA^{\top}\widehat{\bOmega}_{y}(k)^{\top}\|=\|(n-k)^{-1}(\bA\bX^{\top}+\bB\bZ^{\top}+\bE^{\top})\bD_k(\bX\bA^{\top}+\bZ\bB^{\top}+\bE)\bA\|=O_p(p^{\delta_{0}}), \\
        & \|(\bI_p-\bA\bA^{\top})\widehat{\bOmega}_{y}(k)\bB_1\| = \|(n-k)^{-1}(\bI_p-\bA\bA^{\top})(\bB\bZ^{\top}+\bE^{\top})\bD_k(\bZ\bB^{\top}+\bE)\bB_1\|=O_p(p^{\delta_{1k}}), \\
        & \|\bB_1^{\top}\widehat{\bOmega}_{y}(k)^{\top}\| =\|(n-k)^{-1}(\bA\bX^{\top}+\bB\bZ^{\top}+\bE^{\top})\bD_k(\bZ\bB^{\top}+\bE)\bB_1\|=O_p(p^{\delta_{1k}}), \\
        & \|(\bI_p-\bP_{A,B})\widehat{\bOmega}_{y}(k)^{\top}\| =\|(n-k)^{-1}(\bA\bX^{\top}+\bB\bZ^{\top}+\bE^{\top})\bD_k\bE(\bI_p-\bP_{A,B})\|=O_p(n^{-1/2}p^{(1+\delta_{0})/2}),\\
        & \|(\bI_p-\bA\bA^{\top})\widehat{\bOmega}_{y}(k)(\bI_p-\bP_{A,B})\| = \|(n-k)^{-1}(\bI_p-\bA\bA^{\top})(\bB\bZ^{\top}+\bE^{\top})\bD_k\bE\|=O_p(n^{-1/2}p^{(1+\delta_{1})/2}), \\
        &\|\widehat{\bOmega}_{y}(k)(\bI_p-\bP_{A,B})\widehat{\bW}(\bI_p-\bP_{A,B})\widehat{\bOmega}_{y}(k)^{\top}\|=O_p(n^{-1}p)~{\rm by}~\eqref{yux1},
    \end{aligned}
    $$
    which, together with \eqref{middle1} and \eqref{middle2} in Lemma~\ref{lem.S1}, imply that
    $$
    \begin{aligned}
        &\|(\bI_p-\bA\bA^{\top})\widehat{\bOmega}_{y}(k)\widehat{\bW}\widehat{\bOmega}_{y}(k)^{\top}\|\\
        =&\|(\bI_p-\bA\bA^{\top})\widehat{\bOmega}_{y}(k)(\bI_p-\bP_{A,B})\widehat{\bW}\bA\bA^{\top}\widehat{\bOmega}_{y}(k)^{\top}\|+O_p(p^{\delta_{1k}})+O_p(p^{\delta_{1k}+(\delta_0-\delta_1)/2})+O_p(n^{-1/2}p^{(1+\delta_0)/2})\\
        \le&\|(\bI_p-\bA\bA^{\top})\widehat{\bOmega}_{y}(k)(\bI_p-\bP_{A,B})\widehat{\bW}\bA\|\|\widehat{\bOmega}_{y}(k)^{\top}\|+O_p(p^{\delta_{1k}+(\delta_0-\delta_1)/2}+n^{-1/2}p^{(1+\delta_0)/2}).
    \end{aligned}
    $$
    For the first term above, we have
    $$
    \begin{aligned}
        &\|(\bI_p-\bA\bA^{\top})\widehat{\bOmega}_{y}(k)(\bI_p-\bP_{A,B})\widehat{\bW}\bA\|\\
        \le&\|(n-k)^{-1}\bZ^{\top}\bD_k\bE(\bI_p-\bP_{A,B})\widehat{\bW}\bA\|+\|(n-k)^{-1}\bE^{\top}\bD_k\bE(\bI_p-\bP_{A,B})\widehat{\bW}\bA\|\\
        =&\|(n-k)^{-1}\bZ^{\top}\bD_k\bE(\bI_p-\bP_{A,B})\widehat{\bW}\bA\|+O_p(n^{-1/2}p^{(1-\delta_0)/2}),\\
        \le&\|(n-k)^{-1}\bZ^{\top}\bD_k\bE(\bI_p-\bP_{A,B})\widehat{\bXi}_{A}\widehat{\bTheta}_{A}^{-1}\widehat{\bXi}_{A}^{\top}\bA\|+\|(n-k)^{-1}\bZ^{\top}\bD_k\bE(\bI_p-\bP_{A,B})\widehat{\bXi}_{B}\widehat{\bTheta}_{B}^{-1}\widehat{\bXi}_{B}^{\top}\bA\|\\
        &+\|(n-k)^{-1}\bZ^{\top}\bD_k\bE(\bI_p-\bP_{A,B})\widehat{\bXi}_{q-r}\widehat{\bTheta}_{q-r}^{-1}\widehat{\bXi}_{q-r}^{\top}\bA\|+O_p(n^{-1/2}p^{(1-\delta_0)/2})\\
        =&\|(n-k)^{-1}\bZ^{\top}\bD_k\bE(\bI_p-\bP_{A,B})\widehat{\bXi}_{q-r}\widehat{\bTheta}_{q-r}^{-1}\widehat{\bXi}_{q-r}^{\top}\bA\|+O_p(n^{-1/2}p^{(1-\delta_0)/2}),
    \end{aligned}
    $$
    where $\widehat{\bXi}_{q-r}:=(\widehat{\bxi}_{r+1},\dots,\widehat{\bxi}_{q})$ and $\widehat{\bTheta}_{q-r}:=\diag(\hat{\theta}_{r+1},\dots,\hat{\theta}_{q})$. Note that $\widehat{\bXi}_{q-r}$ is a $p \times (q-r)$ matrix, and $\|(n-k)^{-1}\bZ^{\top}\bD_k\bE(\bI_p-\bP_{A,B})\widehat{\bXi}_{q-r}\|=O_p(n^{-1/2}p^{\delta_1/2}).$ Then, by \eqref{eq4ca}, 
    $$\|(n-k)^{-1}\bZ^{\top}\bD_k\bE(\bI_p-\bP_{A,B})\widehat{\bW}\bA\|
    =O_p(p^{(\delta_1-\delta_0-1)/2}+n^{-1/2}p^{(1-\delta_0)/2}).$$ 
    Since $p^{\delta_{1}}=o(p^{2\delta_{1k}})$ by combining Conditions~\ref{cond.cov_a}(ii) and \ref{cond.auto_a}(ii), we have $\|(\bI_p-\bA\bA^{\top})\widehat{\bOmega}_{y}(k)(\bI_p-\bP_{A,B})\widehat{\bW}\bA\|\|\widehat{\bOmega}_{y}(k)^{\top}\|=O_p(n^{-1/2}p^{(1+\delta_0)/2}+p^{(\delta_1+\delta_0-1)/2})=O_p(n^{-1/2}p^{(1+\delta_0)/2}+p^{\delta_{1k}+(\delta_0-\delta_1)/2})$. Combining above, it follows that
    $$
    \|(\bI_p-\bA\bA^{\top})\widehat{\bOmega}_{y}(k)\widehat{\bW}\widehat{\bOmega}_{y}(k)^{\top}\|=O_p\big(n^{-1/2}p^{(1+\delta_0)/2}+p^{\delta_{1k}+(\delta_0-\delta_1)/2})=o_p(p^{\delta_0}\big),
    $$
    which concludes \eqref{eq.lem_our_3}. The proof is complete. 
\end{proof}

Now we are ready to prove Theorem~\ref{advour}. 

{\bf Step 1}. Recall that in \eqref{eq1} we have $\widehat{\bOmega}_{y}(k)\widehat{\bW}\widehat{\bOmega}_{y}(k)^{\top}=\widehat{\bOmega}_{y}(k) \widehat{\bXi}_q\widehat{\bTheta}_{q}^{-1}\widehat{\bXi}_q^{\top} \widehat{\bOmega}_{y}(k)^{\top}$, so we focus on the singular values of $\widehat{\bOmega}_{y}(k) \widehat{\bXi}_q\widehat{\bTheta}_{q}^{-1/2}$. Note that
\begin{eqnarray*}
\widehat{\bOmega}_{y}(k) \widehat{\bXi}_q\widehat{\bTheta}_{q}^{-1/2}=\widehat{\bOmega}_{y}(k)\bP_{A,B} \widehat{\bXi}_q\widehat{\bTheta}_{q}^{-1/2}+\widehat{\bOmega}_{y}(k)(\bI_p-\bP_{A,B}) \widehat{\bXi}_q\widehat{\bTheta}_{q}^{-1/2},
\end{eqnarray*}
where $\rank\big\{\widehat{\bOmega}_{y}(k)\bP_{A,B} \widehat{\bXi}_q\widehat{\bTheta}_{q}^{-1/2}\big\}\le r$. Then by Lemma~\ref{lem.our_1}, $\hat{\lambda}_{k,r+1}=O_p(n^{-1}p)$ holds. 

{\bf Step 2}. Next we prove $\hat{\lambda}_{k1} \asymp  \hat{\lambda}_{kr_0} \asymp p^{\delta_0}$ with probability tending to 1, and $\hat{\lambda}_{k,r_0+1} \asymp \hat{\lambda}_{kr}  \asymp p^{2\delta_{1k}-\delta_1}$ with probability tending to 1 when $\delta_0=\delta_1$. Consider the singular values of $\widehat{\bOmega}_{y}(k)\bP_{A,B} \widehat{\bXi}_q\widehat{\bTheta}_{q}^{-1/2}=\widehat{\bOmega}_{y}(k)\bP_{A,B} \bP_{A,B}\widehat{\bXi}_q\widehat{\bTheta}_{q}^{-1/2}$, which can be divided into two parts, $\widehat{\bOmega}_{y}(k)\bP_{A,B}$ and $\bP_{A,B}\widehat{\bXi}_q\widehat{\bTheta}_{q}^{-1/2}$. Firstly, by combining Theorem~\ref{advauto} and \eqref{eq.lem_our_2} in Lemma~\ref{lem.our_2}, it can be shown that $\bP_{A,B}\widehat{\bOmega}_{y}(k)^{\top}\widehat{\bOmega}_{y}(k)\bP_{A,B}$ has $r$ non-negative eigenvalues, where the $r_0$ largest eigenvalues are of order $p^{2\delta_0}$ with probability tending to 1, and the other $r_1=r-r_0$ non-negative eigenvalues are of order $p^{2\delta_{1k}}$ with probability tending to 1.
Then, note that
$$
\bP_{A,B} \widehat{\bXi}_q\widehat{\bTheta}_{q}^{-1}\widehat{\bXi}_q^{\top}\bP_{A,B}\\
=\bP_{A,B} \widehat{\bXi}_{A}\widehat{\bTheta}_{A}^{-1}\widehat{\bXi}_{A}^{\top}\bP_{A,B}+\bP_{A,B} \widehat{\bXi}_{B}\widehat{\bTheta}_{B}^{-1}\widehat{\bXi}_{B}^{\top}\bP_{A,B}+\bP_{A,B} \widehat{\bXi}_{q-r}\widehat{\bTheta}_{q-r}^{-1}\widehat{\bXi}_{q-r}^{\top}\bP_{A,B},
$$
where $\widehat{\bXi}_{q-r}=(\widehat{\bxi}_{r+1},\dots,\widehat{\bxi}_{q})$ and $\widehat{\bTheta}_{q-r}=\diag(\hat{\theta}_{r+1},\dots,\hat{\theta}_{q})$, and the three terms in the right hand are non-negative.
By using \eqref{eq3} in Lemma~\ref{advsam0}, we have
$$
\begin{aligned}
    \|\bP_{A,B} \widehat{\bXi}_{A}\widehat{\bTheta}_{A}^{-1}\widehat{\bXi}_{A}^{\top}\bP_{A,B}\| \asymp& \|\bP_{A,B} \widehat{\bXi}_{A}\widehat{\bTheta}_{A}^{-1}\widehat{\bXi}_{A}^{\top}\bP_{A,B}\|_{\min} \asymp p^{-\delta_{0}},\\
    \|\bP_{A,B} \widehat{\bXi}_{B}\widehat{\bTheta}_{B}^{-1}\widehat{\bXi}_{B}^{\top}\bP_{A,B}\| \asymp& \|\bP_{A,B} \widehat{\bXi}_{B}\widehat{\bTheta}_{B}^{-1}\widehat{\bXi}_{B}^{\top}\bP_{A,B}\|_{\min} \asymp p^{-\delta_{1}}
\end{aligned}
$$
with probability tending to 1, and by \eqref{eq4ca} in Lemma~\ref{advsam0}, $\|\bP_{A,B} \widehat{\bXi}_{q-r}\widehat{\bTheta}_{q-r}^{-1}\widehat{\bXi}_{q-r}^{\top}\bP_{A,B}\|=O_p(p^{-\delta_{1}})$.
In the following, we consider two cases, i.e., $\delta_0=\delta_1$ and $\delta_0>\delta_1$, separately. On the one hand, when $\delta_0=\delta_1$, we can find that
$$
\|\bP_{A,B} \widehat{\bXi}_q\widehat{\bTheta}_{q}^{-1}\widehat{\bXi}_q^{\top}\bP_{A,B}\| \asymp \|\bP_{A,B} \widehat{\bXi}_q\widehat{\bTheta}_{q}^{-1}\widehat{\bXi}_q^{\top}\bP_{A,B}\|_{\min} \asymp p^{-\delta_{0}}
$$
with probability tending to 1, which, together with the eigenvalues of $\bP_{A,B}\widehat{\bOmega}_{y}(k)^{\top}\widehat{\bOmega}_{y}(k)\bP_{A,B}$, implies that 
$$
\|\widehat{\bOmega}_{y}(k)\bP_{A,B} \widehat{\bXi}_q\widehat{\bTheta}_{q}^{-1/2}\| \asymp \|\widehat{\bOmega}_{y}(k)\bP_{A,B} \widehat{\bXi}_q\widehat{\bTheta}_{q}^{-1/2}\|_{\min} \asymp p^{\delta_{0}/2}
$$ 
with probability tending to 1. Then combined with \eqref{yux1} and Theorem~\ref{advauto}(i), we can show that $\hat{\lambda}_{k1} \asymp  \hat{\lambda}_{kr_0} \asymp p^{\delta_0},$ and $\hat{\lambda}_{k,r_0+1} \asymp \hat{\lambda}_{kr}  \asymp p^{2\delta_{1k}-\delta_1}$ with probability tending to 1. On the other hand, when $\delta_0>\delta_1$, we have
\begin{equation}
    \label{eq.Omega_11}
    \widehat{\bOmega}_{y}(k)\bP_{A,B} \widehat{\bXi}_q\widehat{\bTheta}_{q}^{-1/2}=\widehat{\bOmega}_{y}(k)\bA\bA^{\top} \widehat{\bXi}_q\widehat{\bTheta}_{q}^{-1/2}+\widehat{\bOmega}_{y}(k)\bB_1\bB_1^{\top} \widehat{\bXi}_q\widehat{\bTheta}_{q}^{-1/2}.
\end{equation}
For the first term in \eqref{eq.Omega_11}, note that
$$
\bA^{\top} \widehat{\bXi}_q\widehat{\bTheta}_{q}^{-1}\widehat{\bXi}_q^{\top}\bA=\bA^{\top} \widehat{\bXi}_{A}\widehat{\bTheta}_{A}^{-1}\widehat{\bXi}_{A}^{\top}\bA+\bA^{\top} \widehat{\bXi}_{B}\widehat{\bTheta}_{B}^{-1}\widehat{\bXi}_{B}^{\top}\bA+\bA^{\top}\widehat{\bXi}_{q-r}\widehat{\bTheta}_{q-r}^{-1}\widehat{\bXi}_{q-r}^{\top}\bA,
$$
where the three terms in the right hand are non-negative.
By combining \eqref{eq3} and Theorem~\ref{advsam1}(ii), $\|\bA^{\top} \widehat{\bXi}_{B}\widehat{\bTheta}_{B}^{-1}\widehat{\bXi}_{B}^{\top}\bA\|=o_p(p^{-\delta_0})$, and
$\|\bA^{\top} \widehat{\bXi}_{A}\widehat{\bTheta}_{A}^{-1}\widehat{\bXi}_{A}^{\top}\bA\| \asymp \|\bA^{\top} \widehat{\bXi}_{A}\widehat{\bTheta}_{A}^{-1}\widehat{\bXi}_{A}^{\top}\bA\|_{\min} \asymp p^{-\delta_0}$
with probability tending to 1. By combining 
\eqref{eq4ca} and the fact $\hat{\theta}_{r+1} \asymp \hat{\theta}_{q} \asymp n^{-1}p$ with probability tending to 1 by Theorem~\ref{advsam1}(i), we have $\|\bA^{\top} \widehat{\bXi}_{q-r}\widehat{\bTheta}_{q-r}^{-1}\widehat{\bXi}_{q-r}^{\top}\bA\|=o_p(p^{-\delta_0})$. Then, $\|\bA^{\top} \widehat{\bXi}_q\widehat{\bTheta}_{q}^{-1}\widehat{\bXi}_q^{\top}\bA\| \asymp \|\bA^{\top} \widehat{\bXi}_q\widehat{\bTheta}_{q}^{-1}\widehat{\bXi}_q^{\top}\bA\|_{\min} \asymp p^{-\delta_{0}}$ with probability tending to 1.
Thus, it follows that
\begin{eqnarray}\label{z1}
\|\widehat{\bOmega}_{y}(k)\bA\bA^{\top} \widehat{\bXi}_q\widehat{\bTheta}_{q}^{-1/2}\| \asymp \|\widehat{\bOmega}_{y}(k)\bA\bA^{\top} \widehat{\bXi}_q\widehat{\bTheta}_{q}^{-1/2}\|_{\min} \asymp p^{\delta_{0}/2}
\end{eqnarray}
with probability tending to 1. For the second term in \eqref{eq.Omega_11}, note that
$$
\bB_1^{\top} \widehat{\bXi}_q\widehat{\bTheta}_{q}^{-1}\widehat{\bXi}_q^{\top}\bB_1=\bB_1^{\top} \widehat{\bXi}_{A}\widehat{\bTheta}_{A}^{-1}\widehat{\bXi}_{A}^{\top}\bB_1+\bB_1^{\top} \widehat{\bXi}_{B}\widehat{\bTheta}_{B}^{-1}\widehat{\bXi}_{B}^{\top}\bB_1+\bB_1^{\top} \widehat{\bXi}_{q-r}\widehat{\bTheta}_{q-r}^{-1}\widehat{\bXi}_{q-r}^{\top}\bB_1,
$$
where the three terms in the right hand are non-negative. By Theorem~\ref{advsam1}(ii), it follows that $\widehat{\bXi}_{A}^{\top}\bB_1=o_p(1)$, which implies that $\|\bB_1^{\top} \widehat{\bXi}_{A}\widehat{\bTheta}_{A}^{-1}\widehat{\bXi}_{A}^{\top}\bB_1\|=o_p(p^{-\delta_1})$. By \eqref{eq3}, we have $\|\bB_1^{\top} \widehat{\bXi}_{B}\widehat{\bTheta}_{B}^{-1}\widehat{\bXi}_{B}^{\top}\bB_1\| \asymp \|\bB_1^{\top} \widehat{\bXi}_{B}\widehat{\bTheta}_{B}^{-1}\widehat{\bXi}_{B}^{\top}\bB_1\|_{\min} \asymp p^{-\delta_1}$ with probability tending to 1. By combining 
\eqref{eq4ca} and the fact $\hat{\theta}_{r+1} \asymp \hat{\theta}_{q} \asymp n^{-1}p$ with probability tending to 1 by Theorem~\ref{advsam1}(i), we have $\|\bB_1^{\top} \widehat{\bXi}_{q-r}\widehat{\bTheta}_{q-r}^{-1}\widehat{\bXi}_{q-r}^{\top}\bB_1\|=o_p(p^{-\delta_{1}})$. Then, $\|\bB_1^{\top} \widehat{\bXi}_q\widehat{\bTheta}_{q}^{-1}\widehat{\bXi}_q^{\top}\bB_1\| \asymp \|\bB_1^{\top} \widehat{\bXi}_q\widehat{\bTheta}_{q}^{-1}\widehat{\bXi}_q^{\top}\bB_1\|_{\min} \asymp p^{-\delta_{1}}$ with probability tending to 1. Thus, considering the orthogonality of the columns of $\bA$ and $\bB_1,$ we have
\begin{eqnarray}\label{z2}
\|\widehat{\bOmega}_{y}(k)\bB_1\bB_1^{\top} \widehat{\bXi}_q\widehat{\bTheta}_{q}^{-1/2}\| \asymp \|\widehat{\bOmega}_{y}(k)\bB_1\bB_1^{\top} \widehat{\bXi}_q\widehat{\bTheta}_{q}^{-1/2}\|_{\min} \asymp p^{\delta_{1k}-\delta_{1}/2}
\end{eqnarray}
with probability tending to 1. Note that the rank of $\widehat{\bOmega}_{y}(k)\bP_{A,B}\widehat{\bXi}_q\widehat{\bTheta}_{q}^{-1}\widehat{\bXi}_q^{\top}\bP_{A,B}\widehat{\bOmega}_{y}(k)^{\top}$ is not larger than $r$. By combining \eqref{yux1}, \eqref{eq.Omega_11}, \eqref{z1} and \eqref{z2}, we conclude $\hat{\lambda}_{k1} \asymp  \hat{\lambda}_{kr_0} \asymp p^{\delta_0}$ with probability tending to 1, and
\begin{eqnarray}\label{r2a}
\hat{\lambda}_{k,r_0+1} =O_p(p^{2\delta_{1k}-\delta_1}).
\end{eqnarray}
It is noteworthy that there is a gap between \eqref{r2a} and the assertion $\hat{\lambda}_{k,r_0+1} \asymp \hat{\lambda}_{kr}  \asymp p^{2\delta_{1k}-\delta_1}$ with probability tending to 1. We will prove Theorem~\ref{advour}(ii) first and then conclude this assertion by using \eqref{r2a} and Theorem~\ref{advour}(ii).

{\bf Step 3}. Now we prove Theorem~\ref{advour}(ii). Note that the fact $\hat{\lambda}_{k1} \asymp  \hat{\lambda}_{kr_0} \asymp p^{\delta_0}$ with probability tending to 1 and \eqref{eq.lem_our_3} imply $\|\bA^{\top}\widehat{\bOmega}_{y}(k)\widehat{\bW}\widehat{\bOmega}_{y}(k)^{\top}\bA\|_{\min} \asymp p^{\delta_0}$ with probability tending to 1. To show Theorem~\ref{advour}(ii), by \eqref{eq.lem_our_3} in Lemma~\ref{lem.our_3}, we have
\begin{equation}
    \label{eq.thm4_1}
    \begin{aligned}
        &\|\widehat{\bOmega}_{y}(k)\widehat{\bW}\widehat{\bOmega}_{y}(k)^{\top}-\bA\bA^{\top}\widehat{\bOmega}_{y}(k)\widehat{\bW}\widehat{\bOmega}_{y}(k)^{\top}\bA\bA^{\top}\| \\
        \le & \|(\bI_p-\bA\bA^{\top})\widehat{\bOmega}_{y}(k)\widehat{\bW}\widehat{\bOmega}_{y}(k)^{\top}\|+\|\bA\bA^{\top}\widehat{\bOmega}_{y}(k)\widehat{\bW}\widehat{\bOmega}_{y}(k)^{\top}(\bI_p-\bA\bA^{\top})\|\\
        \le &2\|(\bI_p-\bA\bA^{\top})\widehat{\bOmega}_{y}(k)\widehat{\bW}\widehat{\bOmega}_{y}(k)^{\top}\|=O_p\big(n^{-1/2}p^{(1+\delta_0)/2}+p^{\delta_{1k}+(\delta_0-\delta_1)/2}),
    \end{aligned}
\end{equation}
which, together with the fact $\hat{\lambda}_{k1} \asymp  \hat{\lambda}_{kr_0} \asymp p^{\delta_0}$ with probability tending to 1 and Lemma~\ref{lem.sintheta}, concludes $\big\|\widehat{\bPhi}_{k,A}\widehat{\bPhi}_{k,A}^{\top}-\bA\bA^{\top}\big\|=O_p(n^{-1/2}p^{(1-\delta_{0})/2}+p^{\delta_{1k}-(\delta_{0}+\delta_{1})/2})$.

{\bf Step 4}. Finally, we prove the assertion $\hat{\lambda}_{k,r_0+1} \asymp \hat{\lambda}_{kr}  \asymp p^{2\delta_{1k}-\delta_1}$ with probability tending to 1 when $\delta_0>\delta_1$. By the definitions, we have that $\hat{\lambda}_{k,r_0+1} \ge \cdots \ge \hat{\lambda}_{kq}$ are the $q-r_0$ largest eigenvalues of the matrix
$$
\begin{aligned}
    &(\bI_p-\widehat{\bPhi}_{k,A}\widehat{\bPhi}_{k,A}^{\top})\widehat{\bOmega}_{y}(k)\widehat{\bW}\widehat{\bOmega}_{y}(k)^{\top}(\bI_p-\widehat{\bPhi}_{k,A}\widehat{\bPhi}_{k,A}^{\top})\\
    =&(\bI_p-\widehat{\bPhi}_{k,A}\widehat{\bPhi}_{k,A}^{\top})\widehat{\bOmega}_{y}(k)\widehat{\bXi}_{A}\widehat{\bTheta}_{A}^{-1}\widehat{\bXi}_{A}^{\top}\widehat{\bOmega}_{y}(k)^{\top}(\bI_p-\widehat{\bPhi}_{k,A}\widehat{\bPhi}_{k,A}^{\top})\\
    &+(\bI_p-\widehat{\bPhi}_{k,A}\widehat{\bPhi}_{k,A}^{\top})\widehat{\bOmega}_{y}(k)\widehat{\bXi}_{B}\widehat{\bTheta}_{B}^{-1}\widehat{\bXi}_{B}^{\top}\widehat{\bOmega}_{y}(k)^{\top}(\bI_p-\widehat{\bPhi}_{k,A}\widehat{\bPhi}_{k,A}^{\top})\\
    &+(\bI_p-\widehat{\bPhi}_{k,A}\widehat{\bPhi}_{k,A}^{\top})\widehat{\bOmega}_{y}(k)\widehat{\bXi}_{q-r}\widehat{\bTheta}_{q-r}^{-1}\widehat{\bXi}_{q-r}^{\top}\widehat{\bOmega}_{y}(k)^{\top}(\bI_p-\widehat{\bPhi}_{k,A}\widehat{\bPhi}_{k,A}^{\top}),
\end{aligned}
$$
where the three terms in the right hand are non-negative. We focus on the singular values of $(\bI_p-\widehat{\bPhi}_{k,A}\widehat{\bPhi}_{k,A}^{\top})\widehat{\bOmega}_{y}(k)\widehat{\bXi}_{B}\widehat{\bTheta}_{B}^{-1/2}.$
By Theorem~\ref{advsam1}(ii) under $\delta_0>\delta_1$, we have
\begin{eqnarray}\label{vr2a}
\|\bA^{\top}\widehat{\bXi}_{B}\|=O_p(n^{-1/2}p^{(1-\delta_0)/2}+p^{\delta_1-\delta_0})=o_p(1).
\end{eqnarray}
There exists some constant $C_5>0$ such that
\begin{eqnarray}
    \|(\bI_p-\widehat{\bPhi}_{k,A}\widehat{\bPhi}_{k,A}^{\top})\widehat{\bOmega}_{y}(k)\widehat{\bXi}_{B}\widehat{\bTheta}_{B}^{-1/2}\|_{\min}
        &\ge&\|(\bI_p-\widehat{\bPhi}_{k,A}\widehat{\bPhi}_{k,A}^{\top})\widehat{\bOmega}_{y}(k)\widehat{\bXi}_{B}\|_{\min}\|\widehat{\bTheta}_{B}^{-1/2}\|_{\min}\nonumber\\
        &\ge&C_5\|(\bI_p-\widehat{\bPhi}_{k,A}\widehat{\bPhi}_{k,A}^{\top})\widehat{\bOmega}_{y}(k)\widehat{\bXi}_{B}\|_{\min}p^{-\delta_1/2}\label{eq.thm4_2}
\end{eqnarray}
with probability tending to 1. For \eqref{eq.thm4_2}, consider the decomposition
\begin{equation}
    \label{eq.thm_our_3}
    \begin{aligned}
        &(\bI_p-\widehat{\bPhi}_{k,A}\widehat{\bPhi}_{k,A}^{\top})\widehat{\bOmega}_{y}(k)\widehat{\bXi}_{B}\\
        =&(n-k)^{-1}(\bI_p-\widehat{\bPhi}_{k,A}\widehat{\bPhi}_{k,A}^{\top})(\bA\bX^{\top}+\bB\bZ^{\top}+\bE^{\top})\bD_k\bX\bA^{\top}\widehat{\bXi}_{B}\\
        +&(n-k)^{-1}(\bI_p-\widehat{\bPhi}_{k,A}\widehat{\bPhi}_{k,A}^{\top})(\bA\bX^{\top}+\bB\bZ^{\top}+\bE^{\top})\bD_k\bE \widehat{\bXi}_{B}\\
        +&(n-k)^{-1}(\bI_p-\widehat{\bPhi}_{k,A}\widehat{\bPhi}_{k,A}^{\top})(\bA\bX^{\top}+\bB\bZ^{\top}+\bE^{\top})\bD_k\bZ\bB^{\top}\widehat{\bXi}_{B},\\
    \end{aligned}
\end{equation}
where the first and second terms are bounded by \eqref{vr2a} and Theorem~\ref{advour}(ii), following that
$$
\begin{aligned}
    &\|(n-k)^{-1}(\bI_p-\widehat{\bPhi}_{k,A}\widehat{\bPhi}_{k,A}^{\top})(\bA\bX^{\top}+\bB\bZ^{\top}+\bE^{\top})\bD_k\bX\bA^{\top}\widehat{\bXi}_{B}\|\\
    \le&\|\bA^{\top}\widehat{\bXi}_{B}\|\|\bA\bA^{\top}-\widehat{\bPhi}_{k,A}\widehat{\bPhi}_{k,A}^{\top}\|\|(n-k)^{-1}\bA\bX^{\top}\bD_k\bX\|\\
    &+\|\bA^{\top}\widehat{\bXi}_{B}\|\|(n-k)^{-1}(\bB\bZ^{\top}+\bE^{\top})\bD_k\bX\|
    =o_p(p^{\delta_{1k}}),\\
    &\|(n-k)^{-1}(\bI_p-\widehat{\bPhi}_{k,A}\widehat{\bPhi}_{k,A}^{\top})(\bA\bX^{\top}+\bB\bZ^{\top}+\bE^{\top})\bD_k\bE \widehat{\bXi}_{B}\|\\
    \le&\|\bA\bA^{\top}-\widehat{\bPhi}_{k,A}\widehat{\bPhi}_{k,A}^{\top}\|\|(n-k)^{-1}\bA\bX^{\top}\bD_k\bE\|+\|(n-k)^{-1}(\bB\bZ^{\top}+\bE^{\top})\bD_k\bE\|
    =o_p(p^{\delta_{1k}}),
\end{aligned}
$$
and for the third term of \eqref{eq.thm_our_3}, there exists some constant $C_6>0$ such that
\begin{equation}
    \label{eq.thm4_3}
    \begin{aligned}
        &\|(n-k)^{-1}(\bI_p-\widehat{\bPhi}_{k,A}\widehat{\bPhi}_{k,A}^{\top})(\bA\bX^{\top}+\bB\bZ^{\top}+\bE^{\top})\bD_k\bZ\bB^{\top}\widehat{\bXi}_{B}\|_{\min}\\
        \ge&\|(n-k)^{-1}(\bI_p-\widehat{\bPhi}_{k,A}\widehat{\bPhi}_{k,A}^{\top})\bB\bZ^{\top}\bD_k\bZ\bB^{\top}\widehat{\bXi}_{B}\|_{\min}-\|(n-k)^{-1}\bE^{\top}\bD_k\bZ\bB^{\top}\|\\
        &-\|\bA\bA^{\top}-\widehat{\bPhi}_{k,A}\widehat{\bPhi}_{k,A}^{\top}\|\|(n-k)^{-1}\bA\bX^{\top}\bD_k\bZ\bB^{\top}\|
        \ge C_6p^{\delta_{1k}}
    \end{aligned}
\end{equation}
with probability tending to 1. Combining \eqref{eq.thm4_2}, \eqref{eq.thm_our_3} and \eqref{eq.thm4_3}, we have $$\Vert(\bI_p-\widehat{\bPhi}_{k,A}\widehat{\bPhi}_{k,A}^{\top})\widehat{\bOmega}_{y}(k)\widehat{\bXi}_{B}\widehat{\bTheta}_{B}^{-1/2}\Vert_{\min}\ge C_5C_6p^{\delta_{1k}-\delta_1/2}$$
with probability tending to 1, which implies that $\hat{\lambda}_{kr} \ge C_5^2C_6^2p^{2\delta_{1k}-\delta_1}$ with probability tending to 1. Then, by combining \eqref{r2a}, we conclude that $\hat{\lambda}_{k,r_0+1} \asymp \hat{\lambda}_{kr}  \asymp p^{2\delta_{1k}-\delta_1}$ with probability tending to 1 and complete the proof.

\subsection{Proof of  Corollary \ref{coro.new}}
If Condition \ref{cond.auto_a} is replaced by $\|\bOmega_z(k)\|=O(n^{-1/2}p^{(1+\delta_1)/2})$ for $k\in[m]$, then we have
$\|(\bI_p-\bA\bA^{\top})\widehat{\bOmega}_{y}(k)\bA\|=O_p(n^{-1/2}p^{(1+\delta_1)/2}),
\|(\bI_p-\bA\bA^{\top})\widehat{\bOmega}_{y}(k)\bB_1\|=O_p(n^{-1/2}p^{(1+\delta_1)/2}),$ and $
\|\bB_1^{\top}\widehat{\bOmega}_{y}(k)^{\top}\|=O_p(n^{-1/2}p^{(1+\delta_1)/2}).
$
By similar procedures to Lemma~\ref{lem.our_3}, we can show that
$$
\|(\bI_p-\bA\bA^{\top})\widehat{\bOmega}_{y}(k)\widehat{\bW}\widehat{\bOmega}_{y}(k)^{\top}\|=O_p(n^{-1/2}p^{(1+\delta_1)/2})=o_p(p^{\delta_0}),
$$
which implies that $\hat{\lambda}_{k1} \asymp  \hat{\lambda}_{kr_0} \asymp p^{\delta_0}$ with probability tending to 1, and $\|\widehat{\bPhi}_{k,A}\widehat{\bPhi}_{k,A}^{\top}-\bA\bA^{\top}\|=O_p(n^{-1/2}p^{(1-\delta_0)/2})$. Note that $\hat{\lambda}_{k,r_0+1}^{1/2}$ is the largest singular value of $(\bI_p-\bA\bA^{\top})\widehat{\bOmega}_{y}(k)\widehat{\bW}^{1/2}$. By using similar procedures on $\|(\bI_p-\bA\bA^{\top})\widehat{\bOmega}_{y}(k)\widehat{\bW}^{1/2}\|$ to the proof of Theorem~\ref{advour}(i), we obtain that $\hat{\lambda}_{k,r_0+1}=O_p(n^{-1}p)$. The proof is complete.

\section{Further methodology and derivations}
\label{supsec.C}

\subsection{The solution of \eqref{eq.minimize}}
\label{supsubsec.solution}
This section derives the solution of \eqref{eq.minimize} in Section~\ref{subsec.reduced}. We rewrite \eqref{eq.rrr} as
\begin{equation}
    \label{eq.rrr_sup}
    \by_t=\bL_k\widetilde{\by}_{t-k}+\be_{tk}=\bA\bH_k^{\top}\widetilde{\by}_{t-k}+\be_{tk},~~t=k+1,\dots,n,
\end{equation}
where the rank-reduced coefficient matrix $\bL_k=\bA\bH_k^{\top}\in\eR^{p\times q}$ satisfies $\rank(\bL_k)=r_0$ and $\cC(\bL_k)=\cC(\bA)$. 
Let $\bY_k=(\by_{k+1},\dots,\by_n)^{\top}\in\eR^{(n-k)\times p}$ and $\widetilde{\bY}_k=(\widetilde{\by}_1,\dots,\widetilde{\by}_{n-k})^{\top}\in\eR^{(n-k)\times q}.$ The ordinary least squares estimator for $\bL_k$  is $\widehat{\bL}_{k,\ols}=\bY_k^{\top}\widetilde{\bY}_k(\widetilde{\bY}_k^{\top}\widetilde{\bY}_k)^{-1},$ which, however, does not satisfy the rank-reduced constraint. 
To fit model \eqref{eq.rrr_sup}, we consider solving the following constrained minimization problem, which is equivalent to \eqref{eq.minimize}:
\begin{equation}
    \label{eq.minimize_sup}
    \min_{\bH_k,\bA^{\top}\bA=\bI_{r_0}}\tr\{(\bY_k-\widetilde{\bY}_k\bH_k\bA^{\top})(\bY_k-\widetilde{\bY}_k\bH_k\bA^{\top})^{\top}\}.
\end{equation}
The solution of \eqref{eq.minimize_sup} is $\widehat{\bH}_k=(\widetilde{\bY}_k^{\top}\widetilde{\bY}_k)^{-1}\widetilde{\bY}_k^{\top}\bY_k\widehat{\bPhi}_{k,A}$ and $\widehat{\bA}_k=\widehat{\bPhi}_{k,A},$ where the columns of $\widehat{\bPhi}_{k,A}$ are the $r_0$ leading eigenvectors of $\widehat{\bY}_{k,\ols}^{\top}\widehat{\bY}_{k,\ols}$ with $\widehat{\bY}_{k,\ols}=\widetilde{\bY}_k\widehat{\bL}_{k,\ols}^{\top}.$ By noticing that
$$
\begin{aligned}
    \widehat{\bY}_{k,\ols}^{\top}\widehat{\bY}_{k,\ols}=&\widehat{\bL}_{k,\ols}\widetilde{\bY}_k^{\top}\widetilde{\bY}_k\widehat{\bL}_{k,\ols}^{\top}=\bY_k^{\top}\widetilde{\bY}_k(\widetilde{\bY}_k^{\top}\widetilde{\bY}_k)^{-1}\widetilde{\bY}_k^{\top}\bY_k\\
    =&\Big(\sum_{t=k+1}^n\by_t\by_{t-k}^{\top}\Big)\bQ\Big(\bQ^{\top}\sum_{t=1}^{n-k}\by_t\by_t^{\top}\bQ\Big)^{-1}\bQ^{\top}\Big(\sum_{t=k+1}^n\by_{t-k}\by_t^{\top}\Big)\\
    =&(n-k)\widehat{\bOmega}_{y}(k)\bQ\big(\bQ^{\top}\widehat{\bOmega}_{y}\bQ\big)^{-1}\bQ^{\top}\widehat{\bOmega}_{y}(k)^{\top},
\end{aligned}
$$
we obtain the solution of \eqref{eq.minimize}.

\subsection{Justification for the weight matrix}
\label{supsubsec.calibrating}

This section presents explanations for why $\widehat{\bW}$ can improve the estimation performance and why $\bQ$ is constructed as the leading eigenvectors of $\widehat{\bOmega}_y$. Recall that $\widehat{\bOmega}_y=\sum_{j=1}^{p}\hat{\theta}_j\widehat{\bxi}_j\widehat{\bxi}_j^{\top}$ and $\widehat{\bOmega}_{y}(k)=\sum_{j=1}^{p}\hat{\mu}_{kj}^{1/2}\widehat{\bpsi}_{kj}\widetilde{\bpsi}_{kj}^{\top}$ for $k \in [m],$ with eigenvalues $\hat{\theta}_{1}\ge\cdots\ge\hat{\theta}_{p}\ge0$ and $\hat{\mu}_{k1}\ge\cdots\ge\hat{\mu}_{kp}\ge0,$ respectively. Suppose that the weight matrix has the spectral decomposition $\widehat{\bW}=\sum_{i=1}^{q}\hat{\tau}_i\widehat{\bnu}_{i}\widehat{\bnu}_{i}^{\top}$ with eigenvalues $\hat{\tau}_1\ge\cdots\ge\hat{\tau}_q\ge 0$. 
Notice that
\begin{equation}
    \label{eq.cali_under}
    \begin{aligned}
        \widehat{\bOmega}_{y}(k)\widehat{\bOmega}_{y}(k)^{\top}=&\sum_{j=1}^{p}\hat{\mu}_{kj}\widehat{\bpsi}_{kj}\widehat{\bpsi}_{kj}^{\top},\quad{\rm and}\\
        \widehat{\bOmega}_{y}(k)\widehat{\bW}\widehat{\bOmega}_{y}(k)^{\top}=&\Big(\sum_{j=1}^{p}\hat{\mu}_{kj}^{1/2}\widehat{\bpsi}_{kj}\widetilde{\bpsi}_{kj}^{\top}\Big)\Big(\sum_{i=1}^{q}\hat{\tau}_i\widehat{\bnu}_{i}\widehat{\bnu}_{i}^{\top}\Big)\Big(\sum_{j'=1}^{p}\hat{\mu}_{kj'}^{1/2}\widetilde{\bpsi}_{kj'}\widehat{\bpsi}_{kj'}^{\top}\Big)\\
        =&\sum_{j=1}^{p}c_{kj}\hat{\mu}_{kj}\widehat{\bpsi}_{kj}\widehat{\bpsi}_{kj}^{\top}+\bS_k,
    \end{aligned}
\end{equation}
where $\bS_k=\sum_{1\le j\neq j'\le p}\hat{\mu}_{kj}^{1/2}\hat{\mu}_{kj'}^{1/2}\big(\widetilde{\bpsi}_{kj}^{\top}\widehat{\bW}\widetilde{\bpsi}_{kj'}\big)\widehat{\bpsi}_{kj}\widehat{\bpsi}_{kj'}^{\top}$ is the cross term and $c_{kj}=\sum_{i=1}^{q}\hat{\tau}_i\big(\widetilde{\bpsi}_{kj}^{\top}\widehat{\bnu}_{i}\big)^2$ is regarded as the calibrating coefficient. The cross term $\bS_k$ has a small impact on the $r_0$ leading eigenpairs of \eqref{eq.cali_under} asymptotically, since $\bS_k\widehat{\bpsi}_{kj}=\sum_{j'\neq j}\hat{\mu}_{kj'}^{1/2}\hat{\mu}_{kj}^{1/2}\big(\widetilde{\bpsi}_{kj'}^{\top}\widehat{\bW}\widetilde{\bpsi}_{kj}\big)\widehat{\bpsi}_{kj'},$ is orthogonal to $\widehat{\bpsi}_{kj}$ for each $j\in[r_0]$, thus asymptotically orthogonal to $\cC(\bA)$. Hence, $c_{kj}$'s are expected to improve the eigenstructure of $\widehat{\bOmega}_{y}(k)\widehat{\bW}\widehat{\bOmega}_{y}(k)^{\top}$ compared to $\widehat{\bOmega}_{y}(k)\widehat{\bOmega}_{y}(k)^{\top}$ by adjusting the corresponding coefficients in the leading term $\sum_{j=1}^{p}c_{kj}\hat{\mu}_{kj}^{1/2}\widehat{\bpsi}_{kj}\widehat{\bpsi}_{kj}^{\top}$. To this end, we aim for $c_{kj}$'s to be large (or small) for those corresponding $\cC(\widehat{\bpsi}_{kj})$'s that are close to (or far from) $\cC(\bA)$. Hereafter, we say $\cC(\bK_1)$ is close to $\cC(\bK_2)$, if $\cD\{\cC(\bK_1),\cC(\bK_2)\}$ is large and vice versa. This ensures that the eigenvectors of $\widehat{\bOmega}_{y}(k)\widehat{\bW}\widehat{\bOmega}_{y}(k)^{\top}$ that are closer to $\cC(\bA)$ are associated with larger eigenvalues, relative to the eigenpairs of $\widehat{\bOmega}_{y}(k)\widehat{\bOmega}_{y}(k)^{\top}$. Consequently, these eigenvectors are more likely to be selected when identifying the $r_0$ leading eigenvectors to estimate $\cC(\bA).$ Additionally, this often results in an enhanced separation between common and idiosyncratic components, and, more precisely, a relatively faster order of decrease from the $r_0$-th to the $(r_0+1)$-th largest eigenvalues of $\widehat{\bOmega}_{y}(k)\widehat{\bW}\widehat{\bOmega}_{y}(k)^{\top}$ than that of $\widehat{\bOmega}_{y}(k)\widehat{\bOmega}_{y}(k)^{\top}.$

Based on the properties of $c_{kj}$'s discussed above, we present a reasonable choice of $\bQ.$ 
Note that, for each $k \in [m],$ $\{\widehat{\bpsi}_{kj}\}_{j=1}^{r_0}$ and $\{\widetilde{\bpsi}_{kj}\}_{j=1}^{r_0}$ are the $r_0$ leading eigenvectors of $\widehat{\bOmega}_{y}(k)\widehat{\bOmega}_{y}(k)^{\top}$ and $\widehat{\bOmega}_{y}(k)^{\top}\widehat{\bOmega}_{y}(k)$, respectively. According to the standard autocovariance-based method, both $\cC\big(\{\widehat{\bpsi}_{kj}\}_{j=1}^{r_0}\big)$ and $\cC\big(\{\widetilde{\bpsi}_{kj}\}_{j=1}^{r_0}\big)$ are close to $\cC(\bA)$. For $k\in[m]$ and $j\in[p],$ we expand $\widetilde{\bpsi}_{kj}=\sum_{i=1}^{p}g_{kji}\widehat{\bnu}_i,$ where $\{\widehat{\bnu}_i\}_{i=1}^{p}$ forms an orthonormal basis of $\eR^p$ and the basis coefficients satisfy $\sum_{i=1}^{p}g_{kji}^2=1$. Given that $c_{kj}=\sum_{i=1}^{q}\hat{\tau}_ig_{kji}^2,$ we aim for $\{g_{kji}\}_{i\in[q]}$'s to be large for those corresponding $\cC(\widehat{\bpsi}_{kj})$'s that are close to $\cC(\bA),$ which suggests that $\cD\{\cC(\{\widehat{\bnu}_i\}_{i=1}^{q}),\cC(\bA)\}$ should be small.
By the covariance-based estimation, a feasible choice is to guarantee
\begin{equation}
    \label{eq.psi_select}
    \{\widehat{\bnu}_i\}_{i=1}^{q}=\{\widehat{\bxi}_{i}\}_{i=1}^{q}.
\end{equation}
To achieve this, we use $\bQ=(\widehat{\bxi}_{1},\dots,\widehat{\bxi}_{q})$. Then, $\bQ^{\top}\widehat{\bOmega}_{y}\bQ=\diag(\hat{\theta}_{1}^{-1},\dots,\hat{\theta}_{q}^{-1})$, and
\begin{equation}
    \label{eq.hat_W}
    \sum_{i=1}^{q}\hat{\tau}_i\widehat{\bnu}_{i}\widehat{\bnu}_{i}^{\top}=\widehat{\bW}=\bQ\big(\bQ^{\top}\widehat{\bOmega}_y\bQ\big)^{-1}\bQ^{\top}=\sum_{i=1}^{q}\hat{\theta}_{i}^{-1}\widehat{\bxi}_{i}\widehat{\bxi}_{i}^{\top},
\end{equation}
which implies that the non-degenerate eigenvectors of $\widehat{\bW}$ correspond to the $q$ leading eigenvectors of $\widehat{\bOmega}_{y},$ and thus \eqref{eq.psi_select} holds. 

An alternative method is to sample the entries of $\bQ$ independently from some random distribution with zero mean and, e.g., unit variance, such as standard normal. This random calibration ensures that $\widehat{\bW}$ will still contain some information about $\{\widehat{\bxi}_{i}\}_{i=1}^{q}.$ However, it is expected to suffer from the reduced statistical efficiency and empirical performance compared to \eqref{eq.hat_W}.

\subsection{Heuristic distributional analysis}
\label{supsubsec.approximation}
Recall that $\widetilde{\by}_t=\bQ^{\top}\by_t\in\eR^q,\bx_t=\bH_k^{\top}\widetilde{\by}_{t-k}+\widetilde{\be}_{tk}$, and $\by_t=\bA\bH_k^{\top}\widetilde{\by}_{t-k}+\be_{tk}.$ Let
$$
(\widehat{\bH}_k,\widehat{\bA}_k)=\arg\min_{\bH_k,\bA^{\top}\bA=\bI_{r_0}}\tr\{(\bY_k-\widetilde{\bY}_k\bH_k\bA^{\top})(\bY_k-\widetilde{\bY}_k\bH_k\bA^{\top})^{\top}\},
$$
where $\widehat{\bA}_k$ is equivalent to $\widehat{\bPhi}_{k,A}$ in \eqref{basiceq000anew}. For each $k\in[m]$, we assume that $\be_{tk}$'s are i.i.d. and follow a multivariate normal distribution $\cN(\bzero,\bI_p).$ Then, it is shown that $\widehat{\bA}_k=(\widehat{\bA}_{k1},\dots,\widehat{\bA}_{kr_0})$ is the maximum likelihood estimator of $\bA=(\bA_1,\dots,\bA_{r_0}).$ Note that 
$
\cov(\by_t,\widetilde{\by}_{t-k})=\bA\bH_k^{\top}\bOmega_{\widetilde{y}},
$
where $\bOmega_{\widetilde{y}}=\cov(\widetilde{y}_t)$ is positive definite.
Then, for each $k\in[m],$
$$
\bOmega_y(k)\bW\bOmega_y(k)=\cov(\by_t,\widetilde{\by}_{t-k})\{\cov(\widetilde{\by}_{t-k})\}^{-1}\cov(\by_t,\widetilde{\by}_{t-k})^{\top}=\bA\bH_k^{\top}\bOmega_{\widetilde{y}}\bH_k\bA^{\top}
$$
has rank $r_0.$ Given that $\bH_k^{\top}\bOmega_{\widetilde{y}}\bH_k=\bH_k^{\top}\bQ^{\top}\bOmega_{y}\bQ\bH_k\in\eR^{r_0\times r_0}$ is full-ranked, the space spanned by the $r_0$ leading eigenvectors of $\bOmega_y(k)\bW\bOmega_y(k)$ is exactly $\cC(\bA)$. Let $\{\bA_j\}_{j\in[p]}$ be an orthonormal basis. 
Without loss of generality, assume $\widehat{\bA}_{kj}^{\top}\bA_j\ge0.$
By Theorem 2.4 of \textcolor{blue}{Reinsel et al.} (\textcolor{blue}{2022}), for $j\in[r_0],$
\begin{equation}
\label{eq.dist1}
\eE\{n(\widehat{\bA}_{kj}-\bA_j)(\widehat{\bA}_{kj}-\bA_j)^{\top}\}\to\bG_{jj}:=\sum_{1\le i\neq j\le p}\frac{\widetilde{\lambda}_{kj}+\widetilde{\lambda}_{ki}}{(\widetilde{\lambda}_{kj}-\widetilde{\lambda}_{ki})^2}\bA_i\bA_i^{\top},
\end{equation}
as $n\to\infty$, and for $1\le j\neq l\le r_0,$
\begin{equation}
\label{eq.dist2}
\eE\{n(\widehat{\bA}_{kj}-\bA_j)(\widehat{\bA}_{kl}-\bA_l)^{\top}\}\to\bG_{jl}:=-\frac{\widetilde{\lambda}_{kj}+\widetilde{\lambda}_{kl}}{(\widetilde{\lambda}_{kj}-\widetilde{\lambda}_{kl})^2}\bA_l\bA_j^{\top},
\end{equation}
as $n\to\infty$, where $\widetilde{\lambda}_{kj}$'s are the eigenvalues of $\bOmega_y(k)\bW\bOmega_y(k)$ for $j\in[r_0]$, and $\widetilde{\lambda}_{kj}=0$ for $j=r_0+1,\dots,p.$ Thus, $\bG_{jj}$ can be rewritten as
$$
\bG_{jj}=\sum_{1\le i\neq j\le r_0}\frac{\widetilde{\lambda}_{kj}+\widetilde{\lambda}_{ki}}{(\widetilde{\lambda}_{kj}-\widetilde{\lambda}_{ki})^2}\bA_i\bA_i^{\top}+\frac{1}{\widetilde{\lambda}_{kj}}(\bI_p-\bA\bA^{\top}).
$$
Let $\bG=\{\bG_{jl}\}\in\eR^{pr_0\times pr_0}$ be the asymptotic covariance matrix of the vectorization of $\widehat{\bA}_k,$ where $\bG_{jl}$ is the $(j,l)$-th block for $j,l\in[r_0]$. 
Then, by \eqref{eq.dist1} and \eqref{eq.dist2}, we have $\sqrt{n}{\rm vec}(\widehat{\bA}_k-\bA)\to\cN(\bzero,\bG)$ as $n\to\infty,$ which can be used to conduct hypothesis tests on $\bA$. Additionally, similar to Theorem~\ref{theoforbasic}, we can show that $\widetilde{\lambda}_{kj}\asymp p^{\delta_0}$ for $j\in[r_0].$ Thus, $\Vert\bG_{jj}\Vert_{\F}\lesssim p^{1-\delta_0}$ and $\Vert\bG_{jl}\Vert_{\F}\lesssim p^{-\delta_0}$ for $j\neq l.$ Hence, $\Vert\bG\Vert\le\big(\sum_{j=1}^{r_0}\sum_{l=1}^{r_0}\Vert\bG_{jl}\Vert_{\F}^2\big)^{1/2}\lesssim p^{1-\delta_0}$, and thus $\Vert\widehat{\bA}_k\widehat{\bA}_k^{\top}-\bA\bA^{\top}\Vert=O_p(n^{-1/2}p^{(1-\delta_0)/2})$, aligning with the rate in Theorem~\ref{theoforbasic}(ii).

\subsection{Reduced-rank autoregression formulation for matrix-valued time series}
\label{supsubsec.reduced}
This section provides an explanation for the suggested forms of weight matrices for matrix factor models as discussed in Section~\ref{sec.disscuss}. 
Recall that $\{\bY_t\}_{t\in\eZ}$ is a stationary matrix-valued time series of size $p_1\times p_2$ satisfying the factor model $\bY_t=\bR\bX_t\bC^{\top}+\bE_t$. The matrices $\widehat{\bOmega}_{y,ij}^{(1)}(k)$ and $\widehat{\bOmega}_{y,ij}^{(2)}(k)$ denote the sample estimates of $\bOmega_{y,ij}^{(1)}(k)$ and $\bOmega_{y,ij}^{(2)}(k)$, respectively. 
Furthermore, let the $j$-th row of $\bR$ and $\bC$ be $\br_{j\cdot}$ and $\bc_{j\cdot}$, respectively, the $j$-th column of $\bE_t$ be $\be_{t,\cdot j}$, and $\bz_{t,j}=\bX_t\bc_{j\cdot}^{\top}$. Then, we have the following $p_2$ vector factor models:
\begin{equation}
    \label{eq.model_matrix_v}
    \by_{t,\cdot j}=\bR\bX_t\bc_{j\cdot}^{\top}+\be_{t,\cdot j}=\bR\bz_{t,j}+\be_{t,\cdot j},~~j\in[p_2].
\end{equation}
Without loss of generality, we assume that $\eE(\by_{t,\cdot j})=\bzero$.
Letting $\bQ_{1j}$ consists of the $q_1$ leading eigenvectors of $\widehat{\bOmega}_{y,jj}^{(1)}(0)$ with $d_1<q_1\le p_1\wedge n$, we project $\by_{t,\cdot j}$ to $\widetilde{\by}_{t,\cdot j}=\bQ_{1j}^{\top}\by_{t,\cdot j}\in\eR^{q_1}.$ 
By assuming the latent factor is of the form $\bz_{t,i}=\bH_{kij}^{\top}\widetilde{\by}_{t-k,\cdot j}+\widetilde{\be}_{tkij}$, where $\bH_{kij}$ is  $q_1\times d_1$ and full-ranked, model \eqref{eq.model_matrix_v} becomes a reduced-rank autoregression
\begin{equation}
    \label{eq.rrr_matrix}
    \by_{t,\cdot i}=\bL_{kij}\widetilde{\by}_{t-k,\cdot j}+\be_{tkij}=\bR\bH_{kij}^{\top}\widetilde{\by}_{t-k,\cdot j}+\be_{tkij},~~t=k+1,\dots,n,
\end{equation}
where $\bL_{kij}=\bR\bH_{kij}^{\top}\in\eR^{p_1\times q_1}$ satisfies $\rank(\bL_{kij})=d_1$ and $\cC(\bL_{kij})=\cC(\bR)$, and $\be_{tkij}=\bR\widetilde{\be}_{tkij}+\be_{t,\cdot i}$.
Let $\bY_{k,\cdot j}=(\by_{k+1,\cdot j},\dots,\by_{n,\cdot j})^{\top}\in\eR^{(n-k)\times p_1}$, and $\widetilde{\bY}_{k,\cdot j}=(\widetilde{\by}_{1,\cdot j},\dots,\widetilde{\by}_{n-k,\cdot j})^{\top}\in\eR^{(n-k)\times q_1}.$
To fit \eqref{eq.rrr_matrix}, we consider solving the constrained minimization problem:
\begin{equation}
    \label{eq.min_matrix}(\widehat{\bH}_{kij},\widehat{\bR}_{kij})=\arg\min_{\bH_{kij},\bR^{\top}\bR=\bI_{d_1}}\tr\{(\bY_{k,\cdot i}-\widetilde{\bY}_{k,\cdot j}\bH_{kij}\bR^{\top})(\bY_{k,\cdot i}-\widetilde{\bY}_{k,\cdot j}\bH_{kij}\bR^{\top})^{\top}\}.
\end{equation}
Following similar procedures to those in Section~\ref{supsubsec.solution}, the solution of~\eqref{eq.min_matrix} is $\widehat{\bH}_{kij}=(\widetilde{\bY}_{k,\cdot j}^{\top}\widetilde{\bY}_{k,\cdot j})^{-1}\widetilde{\bY}_{k,\cdot j}^{\top}\bY_{k,\cdot i}\widehat{\bPhi}_{kij,R}$ and $\widehat{\bR}_{kij}=\widehat{\bPhi}_{kij,R},$ where the columns of $\widehat{\bPhi}_{kij,R}$ are the $d_1$ leading eigenvectors of 
$$
\begin{aligned}
    \widehat{\bY}_{k,\cdot i,\ols}^{\top}\widehat{\bY}_{k,\cdot i,\ols}
    =&(n-k)\widehat{\bOmega}_{y,ij}^{(1)}(k)\bQ_{1j}\big\{\bQ_{1j}^{\top}\widehat{\bOmega}_{y,jj}^{(1)}(0)\bQ_{1j}\big\}^{-1}\bQ_{1j}^{\top}\widehat{\bOmega}_{y,ij}^{(1)}(k)^{\top}.
\end{aligned}
$$
This implies that the weight matrix to construct $\widehat{\bM}^{(1)}$ is $\widehat{\bW}_{j}^{(1)}=\bQ_{1j}\{\bQ_{1j}^{\top}\widehat{\bOmega}_{y,jj}^{(1)}(0)\bQ_{1j}\}^{-1}\bQ_{1j}^{\top}$ for $j\in[p_2]$, which is non-negative and independent of lag $k$. The weight matrix $\widehat{\bW}_{j}^{(2)}$ to improve the estimation of $d_2$ and $\bC$ can be similarly obtained.

\section{Additional simulation results}
\label{supsec.D}

\subsection{\texorpdfstring{Estimation of the loading space with known $r_0$}{Estimation of the loading space with known r0}}
\label{supsubsec.true}

This section presents the estimation results of the loading space when the number of (strong) factors $r_0$ is known. These results are compared with those based on the estimated $\hat{r}_0$, as reported in Tables~\ref{tab.loading} and \ref{tab.loading_new}. 
For comparison, with known $r_0$, the estimated loading matrix $\widehat{\bA}^*$ is constructed by the $r_0$ leading eigenvectors of matrices $\widehat{\bOmega}_y,\widehat{\bM}_1=\sum_{k=1}^{m}\widehat{\bOmega}_y(k)\widehat{\bOmega}_y(k)^{\top},$ and $\widehat{\bM}=\sum_{k=1}^{m}\widehat{\bOmega}_y(k)\widehat{\bW}\widehat{\bOmega}_y(k)^{\top}$ for Cov, Auto, and WAuto, respectively. The simulation settings follow those in Section~\ref{sec.sim}.

\begin{table}[H]
    \small
    \centering
    \caption{The mean and standard deviation (in parentheses) of $\cD\big\{\cC(\widehat{\bA}^*),\cC(\bA)\big\}$ using the true $r_0$ for model~\eqref{eq.model_ex1} over 1000 simulation runs.} 
    \label{tab.true_loading}{
    \begin{tabular}{l|ccc|ccc}
            \hline
             & \multicolumn{3}{c|}{$\delta_0=0.75$} & \multicolumn{3}{c}{$\delta_0=1$} \\
          & Cov & Auto & WAuto & Cov & Auto & WAuto \\
          \hline
              $p=50$ & 0.477  & 0.183  & 0.181  & 0.149  & 0.109  & 0.108  \\
          & (0.137) & (0.046) & (0.044) & (0.045) & (0.026) & (0.025) \\
    $p=100$ & 0.304  & 0.197  & 0.197  & 0.104  & 0.111  & 0.109  \\
          & (0.092) & (0.041) & (0.041) & (0.021) & (0.023) & (0.022) \\
    $p=200$ & 0.225  & 0.213  & 0.214  & 0.091  & 0.112  & 0.110  \\
          & (0.045) & (0.039) & (0.041) & (0.014) & (0.021) & (0.021) \\
    $p=400$ & 0.201  & 0.223  & 0.225  & 0.085  & 0.109  & 0.108  \\
          & (0.031) & (0.038) & (0.040) & (0.011) & (0.020) & (0.020) \\
        \hline
    \end{tabular}}
\end{table}

\begin{table}[H]
    \small
    \centering
    \caption{The mean and standard deviation (in parentheses) of $\cD\big\{\cC(\widehat{\bA}^*),\cC(\bA)\big\}$ using the true $r_0$ for model~\eqref{eq.model_ex2} over 1000 simulation runs.}
    \label{tab.true_loading_new}{
    \begin{tabular}{l|ccccccc}
            \hline
          & Cov & Auto & WAuto & Auto1 & WAuto1 & Auto2 & WAuto2 \\
          \hline
        & \multicolumn{7}{c}{$\delta_0=1,~~~\delta_1=1$}\\
        $p=50$ & 0.303  & 0.250  & 0.252  & 0.268  & 0.269  & 0.286  & 0.293  \\
          & (0.086) & (0.046) & (0.046) & (0.057) & (0.056) & (0.054) & (0.058) \\
    $p=100$ & 0.265  & 0.238  & 0.239  & 0.246  & 0.246  & 0.265  & 0.270  \\
          & (0.068) & (0.034) & (0.034) & (0.039) & (0.037) & (0.041) & (0.043) \\
    $p=300$ & 0.231  & 0.230  & 0.231  & 0.231  & 0.231  & 0.251  & 0.255  \\
          & (0.049) & (0.031) & (0.032) & (0.033) & (0.033) & (0.036) & (0.039) \\
    $p=500$ & 0.224  & 0.230  & 0.231  & 0.230  & 0.230  & 0.250  & 0.255  \\
          & (0.043) & (0.030) & (0.030) & (0.031) & (0.031) & (0.034) & (0.037) \\
          \hline
    & \multicolumn{7}{c}{$\delta_0=1,~~~\delta_1=0.85$}\\
     $p=50$ & 0.242  & 0.241  & 0.245  & 0.253  & 0.258  & 0.281  & 0.289  \\
          & (0.043) & (0.040) & (0.042) & (0.044) & (0.046) & (0.052) & (0.057) \\
    $p=100$ & 0.217  & 0.230  & 0.233  & 0.234  & 0.237  & 0.260  & 0.265  \\
          & (0.029) & (0.031) & (0.032) & (0.032) & (0.033) & (0.039) & (0.042) \\
    $p=300$ & 0.200  & 0.223  & 0.225  & 0.221  & 0.223  & 0.246  & 0.250  \\
          & (0.022) & (0.028) & (0.030) & (0.028) & (0.029) & (0.035) & (0.038) \\
    $p=500$ & 0.198  & 0.223  & 0.225  & 0.220  & 0.222  & 0.245  & 0.249  \\
          & (0.021) & (0.027) & (0.028) & (0.026) & (0.027) & (0.033) & (0.036) \\
        \hline
    \end{tabular}}
\end{table}

Tables \ref{tab.true_loading} and \ref{tab.true_loading_new} report numerical summaries of the estimation errors for $\cC(\bA)$ under models \eqref{eq.model_ex1} and \eqref{eq.model_ex2}, respectively, revealing several notable trends. Firstly, when the signal is relatively weak, Auto and WAuto outperform Cov, whereas Cov performs the best when the signal is relatively strong. This aligns with the findings in Tables ~\ref{tab.loading} and \ref{tab.loading_new} based on the estimated number of (strong) factors $\hat r_0$. 
Secondly, WAuto achieves the best performance in three cases reported in Table \ref{tab.true_loading}, i.e., $(\delta_0,p)=(0.75,50),(0.75,100)$ and $(1,50)$, which correspond to scenarios with very weak signals. This highlights the advantage of the proposed method in the most challenging settings. Thirdly, in the remaining cases, Auto and WAuto exhibit comparable performance, supporting the discussion in Section \ref{subsec.summary} that WAuto does not attain a faster convergence rate than Auto, and in certain cases both converge at the same rate.

\subsection{\texorpdfstring{Sensitivity analysis with respect to the choice of $q$}{Sensitivity analysis with respect to the choice of q}}
\label{supsubsec.sensitivity}

In this section, we conduct simulations to see how (in)sensitive the proposed method is to the choice of $q.$ Specifically, we use the same data generating process as in Section~\ref{sim.bfm}, fixing $q$ at 5, 10, 15, and 20, and replicating each simulation 1000 times. Tables \ref{tab.sen_num} and \ref{tab.sen_loading} present the average relative frequencies of $\hat{r}_0=r_0$ and the mean estimation errors of $\cC(\bA)$, respectively. Compared with the results in Tables~\ref{tab.num} and \ref{tab.loading}, we observe several apparent patterns. 
Firstly, regardless of the fixed value of $q$, WAuto consistently outperforms Auto, demonstrating the robustness of our proposed calibrating weight to the choice of $q$ in improving the performance of the standard autocovariance-based method. 
Secondly, using the generalized BIC introduced to adaptively select the tuning parameter $q$ outperforms fixing $q$ at any specific value across all simulation settings. This suggests that, despite lacking the theoretical guarantee, the generalized BIC demonstrates good empirical performance.

\begin{table}[H]
    \small
    \centering
    \caption{The relative frequency estimate of $\eP(\hat{r}_{0}=r_0)$ and the average of $\hat{r}_{0}$ (in parentheses) for model~\eqref{eq.model_ex1} with fixed values of $q$ over 1000 simulation runs.}
    \label{tab.sen_num}{
    \begin{tabular}{l|cccc|cccc}
            \hline
            & \multicolumn{4}{c|}{$\delta_0=0.75$} & \multicolumn{4}{c}{$\delta_0=1$} \\
          $q$ & 5     & 10    & 15    & 20    & 5     & 10    & 15    & 20 \\
          \hline
    $p=50$ & 0.736  & 0.855  & 0.803  & 0.744  & 0.995  & 0.990  & 0.983  & 0.964  \\
          & (2.659) & (2.768) & (2.680) & (2.605) & (2.992) & (2.984) & (2.972) & (2.943) \\
    $p=100$ & 0.908  & 0.874  & 0.853  & 0.817  & 1.000  & 0.998  & 0.995  & 0.991  \\
          & (2.866) & (2.801) & (2.759) & (2.694) & (3.000)   & (2.997) & (2.991) & (2.986) \\
    $p=200$ & 0.950  & 0.903  & 0.875  & 0.849  & 0.999  & 0.997  & 0.994  & 0.992  \\
          & (2.923) & (2.840) & (2.797) & (2.755) & (2.999) & (2.995) & (2.989) & (2.985) \\
    $p=400$ & 0.956  & 0.921  & 0.902  & 0.886  & 1.000  & 0.999  & 0.998  & 0.997  \\
          & (2.923) & (2.874) & (2.847) & (2.817) & (3.000)   & (2.998) & (2.997) & (2.995) \\
        \hline
    \end{tabular}}
\end{table}

\begin{table}[H]
    \small
    \centering
    \caption{The mean and standard deviation (in parentheses) of $\cD\big\{\cC(\widehat{\bA}),\cC(\bA)\big\}$ for model~\eqref{eq.model_ex1} with fixed values of $q$ over 1000 simulation runs.}
    \label{tab.sen_loading}{
    \begin{tabular}{l|cccc|cccc}
            \hline
            & \multicolumn{4}{c|}{$\delta_0=0.75$} & \multicolumn{4}{c}{$\delta_0=1$} \\
          $q$ & 5     & 10    & 15    & 20    & 5     & 10    & 15    & 20 \\
           \hline
    $p=50$ & 0.312  & 0.257  & 0.285  & 0.318  & 0.113  & 0.115  & 0.119  & 0.129  \\
          & (0.226) & (0.203) & (0.230) & (0.251) & (0.051) & (0.067) & (0.086) & (0.119) \\
    $p=100$ & 0.245  & 0.262  & 0.273  & 0.293  & 0.111  & 0.111  & 0.112  & 0.114  \\
          & (0.153) & (0.184) & (0.200) & (0.221) & (0.023) & (0.035) & (0.052) & (0.062) \\
    $p=200$ & 0.241  & 0.264  & 0.276  & 0.289  & 0.113  & 0.113  & 0.114  & 0.115  \\
          & (0.118) & (0.164) & (0.181) & (0.196) & (0.026) & (0.041) & (0.056) & (0.064) \\
    $p=400$ & 0.250  & 0.264  & 0.271  & 0.278  & 0.110  & 0.109  & 0.109  & 0.109  \\
          & (0.119) & (0.145) & (0.157) & (0.170) & (0.020) & (0.030) & (0.033) & (0.040) \\
        \hline
    \end{tabular}}
\end{table}

\subsection{Iterative weight-calibrated method}
\label{supsubsec.iterative}

In this section, we develop an iterative extension of our weight-calibrated method, and compare their performance through simulations. The current weight-calibrated method can be viewed as a two-step estimator, where the first step projects the data from the $p$-dimensional column space onto a rank-$q$ subspace, and the second step performs eigenanalysis within this subspace to estimate the number of factors and the loadings. This subspace is constructed to retain the most information of common components while filtering out idiosyncratic noise, thereby enhancing their separation and improving the estimation performance.
Building on this idea, we propose an iterative extension of our method. At each iteration, the matrix $\widehat{\bM}=\sum_{k=1}^{m}\widehat\bOmega_{y}(k)\widehat{\bW}\,\widehat\bOmega_{y}(k)^{\top}$ obtained from the previous step is used to update the projection matrix $\bQ$, and then eigenanalysis is performed in the resulting subspace. Here, $\cC(\bQ)$ can be regarded as an ``overfitted'' candidate estimator of the loading space, and the iterative procedure aims to progressively refine its accuracy. The generalized BIC proposed in Section~\ref{subsec.IC} is applied at each step to determine $q$. We summarize the iterative method in Algorithm \ref{alg.iterative} below. Note that the estimation results for the non-iterative method in Section~\ref{sec.method} correspond to $\hat{r}_{0}^{(0)}$ and $\widehat{\bA}^{(0)}$ in the algorithm.

We next employ the same data generating process as in Section~\ref{sim.bfm} to assess the performance of the proposed iterative weight-calibrated autocovariance-based method (denoted as IWAuto). For implementation, we set $l_{\max}=10,\delta_0=0.5,0.75,1$, and retain all other settings from Section~\ref{sim.bfm}. Tables~\ref{tab.iter_num} and \ref{tab.iter_loading} present the comparison between WAuto and IWAuto in terms of estimation accuracy for $r_0$ and $\cC(\bA)$, respectively. The results show that IWAuto outperforms WAuto substantially in most cases, particularly when the factor strength is weak. This finding highlights the potential benefits of the iterative extension to the proposed method and points to a promising direction for future research.

\begin{algorithm}[H]
\spacingset{1.2}
\caption{Iterative Weight-Calibrated Autocovariance-based Method}
\label{alg.iterative}
\begin{algorithmic}[1]
\Require Data $\{\by_t\}_{t=1}^n$, tuning parameters $m,C,q_0$, maximum number of iterations $l_{\max}$
\Ensure Estimated number of factor $\hat{r}_0$, estimated factor loading matrix $\widehat{\bA}$

\State \textbf{Initialization:}
\State Compute the sample covariance $\widehat{\bOmega}_y=n^{-1}\sum_{t=1}^n \by_t\by_t{^{\top}}$
\State Select $q^{(0)}$ using the generalized BIC in (8)
\State Set $\bQ^{(0)}$ to consist of the $q^{(0)}$ leading eigenvectors of $\widehat{\bOmega}_y$
\State Construct {\small$\widehat{\bW}^{(0)}= \bQ^{(0)}\big\{(\bQ^{(0)})^{\top}\widehat{\bOmega}_{y}\bQ^{(0)}\big\}^{-1}(\bQ^{(0)})^{\top}$}, and {\small$\widehat{\bM}^{(0)} = \sum_{k=1}^{m}\widehat{\bOmega}_{y}(k)\widehat{\bW}^{(0)}\widehat{\bOmega}_{y}(k)^{\top}$}
\State Obtain $\hat{r}_{0}^{(0)}$ by (6) and $\widehat{\bA}^{(0)}$ by the $\hat{r}_{0}^{(0)}$ leading eigenvectors of $\widehat{\bM}^{(0)}$
\Statex

\For{$l=1,2,\dots,l_{\max}$}
    \State \textbf{BIC-based selection of $q^{(l)}$}:
    \State $\bar{r}_0$ is obtained through (6) using $q_0$ by replacing $\widehat{\bW}$ with $\widehat{\bW}^{(l-1)}$
    \For{each $q \in \{\bar{r}_0+1,\dots,q_0\}$}
        \For{$k=1$ to $m$}
            \State Compute $L_k^{(l)}(q)=(pn)^{-1}\sum_{t=k+1}^n\|\by_t-\widehat{\bA}_k^{(l)} (\widehat{\bH}_k^{(l)})^{\top}\widetilde{\by}_{t-k}\|^2$, where $(\widehat{\bH}_k^{(l)},\widehat{\bA}_k^{(l)})$ is the solution of the constrained minimization problem in (5), where $\bQ$ is defined as the $q$ leading eigenvectors of $\widehat{\bM}^{(l-1)}$ and $r_0$ is estimated by (6) using $q_0$ by replacing $\widehat{\bW}$ with $\widehat{\bW}^{(l-1)}$
            \State Set $\bic_k^{(l)}(q)=pn\log L_k^{(l)}(q)+Cd_k(q)\log(pn)$
        \EndFor
    \EndFor
    \State $q^{(l)} = \arg\min_{\bar r_0<q\le q_0}\sum_{k=1}^{m}\bic_k^{(l)}(q)$ 
    \Statex
    
    \State \textbf{Update the projection matrix and perform estimation:}
    \State Set $\bQ^{(l)}$ to consist of the $q^{(l)}$ leading eigenvectors of $\widehat{\bM}^{(l-1)}$
    \State Construct {\small$\widehat{\bW}^{(l)}= \bQ^{(l)}\big\{(\bQ^{(l)})^{\top}\widehat{\bOmega}_{y}\bQ^{(l)}\big\}^{-1}(\bQ^{(l)})^{\top}$}, and {\small$\widehat{\bM}^{(l)}=\sum_{k=1}^{m}\widehat{\bOmega}_{y}(k)\widehat{\bW}^{(l)}\widehat{\bOmega}_{y}(k)^{\top}$}
    \State Obtain $\hat{r}_{0}^{(l)}$ by (6) and $\widehat{\bA}^{(l)}$ by the $\hat{r}_{0}^{(l)}$ leading eigenvectors of $\widehat{\bM}^{(l)}$
    \Statex
    
    \State \textbf{Early stopping:}
    \If{$l\ge 1$ \textbf{ and } $\hat{r}_0^{(l)}=\hat{r}_0^{(l-1)}$}
        \State \textbf{break}
    \EndIf
\EndFor
\Statex

\State \Return $\hat{r}_0=\hat{r}_0^{(l)}$, $\widehat{\bA}=\widehat{\bA}^{(l)}$
\end{algorithmic}
\end{algorithm}

\begin{table}[H]
    \small
    \centering
    \caption{Comparison between WAuto and IWAuto in terms of the relative frequency estimate of $\eP(\hat{r}_{0}=r_0)$ and the average of $\hat{r}_{0}$ (in parentheses) for model~\eqref{eq.model_ex1} over 1000 simulation runs.}
    \label{tab.iter_num}{
    \begin{tabular}{l|cc|cc|cc}
            \hline
            & \multicolumn{2}{c|}{$\delta_0=0.5$} & \multicolumn{2}{c|}{$\delta_0=0.75$} & \multicolumn{2}{c}{$\delta_0=1$} \\
          & WAuto & IWAuto & WAuto & IWAuto & WAuto & IWAuto \\
          \hline
    $p=50$ & 0.444  & 0.593  & 0.899  & 0.950  & 0.998  & 0.998  \\
          & (2.369) & (2.567) & (2.855) & (2.935) & (2.997) & (2.997) \\
    $p=100$ & 0.233  & 0.549  & 0.942  & 0.973  & 1.000  & 1.000  \\
          & (1.832) & (2.355) & (2.922) & (2.963) & (3.000)   & (3.000) \\
    $p=200$ & 0.158  & 0.445  & 0.961  & 0.971  & 0.999  & 0.999  \\
          & (1.726) & (2.175) & (2.940) & (2.955) & (2.999) & (2.999) \\
    $p=400$ & 0.231  & 0.434  & 0.965  & 0.970  & 1.000  & 1.000  \\
          & (1.821) & (2.150) & (2.939) & (2.948) & (3.000)   & (3.000) \\
        \hline
    \end{tabular}}
\end{table}

\begin{table}[H]
    \small
    \centering
    \caption{Comparison between WAuto and IWAuto in terms of the mean and standard deviation (in parentheses) of $\cD\big\{\cC(\widehat{\bA}),\cC(\bA)\big\}$ for model~\eqref{eq.model_ex1} over 1000 simulation runs.}
    \label{tab.iter_loading}{
    \begin{tabular}{l|cc|cc|cc}
            \hline
            & \multicolumn{2}{c|}{$\delta_0=0.5$} & \multicolumn{2}{c|}{$\delta_0=0.75$} & \multicolumn{2}{c}{$\delta_0=1$} \\
          & WAuto & IWAuto & WAuto & IWAuto & WAuto & IWAuto \\
          \hline
    $p=50$ & 0.522  & 0.454  & 0.230  & 0.204  & 0.109  & 0.109  \\
          & (0.234) & (0.218) & (0.165) & (0.118) & (0.037) & (0.037) \\
    $p=100$ & 0.640  & 0.510  & 0.223  & 0.209  & 0.109  & 0.109  \\
          & (0.201) & (0.207) & (0.121) & (0.089) & (0.022) & (0.022) \\
    $p=200$ & 0.682  & 0.579  & 0.233  & 0.228  & 0.111  & 0.111  \\
          & (0.171) & (0.196) & (0.107) & (0.095) & (0.026) & (0.026) \\
    $p=400$ & 0.674  & 0.607  & 0.243  & 0.241  & 0.108  & 0.108  \\
          & (0.173) & (0.182) & (0.108) & (0.101) & (0.020) & (0.020) \\
        \hline
    \end{tabular}}
\end{table}

\subsection{\texorpdfstring{Extension to information criterion for estimating $r_0$}{Extension to information criterion for estimating r0}}
\label{supsubsec.IC}

Our paper demonstrates the usefulness of the calibrating weight in enhancing the accuracy of the eigenvalue-ratio estimator for determining $r_0$ within the autocovariance-based framework. In this section, we further incorporate the calibrating weight into an alternative commonly-adopted information-criterion method for estimating $r_0$ and evaluate its performance through simulations.

Let $\widehat{\bA}_{(\gamma)}$ denote the estimated loading matrix with $\gamma$ factors. For the covariance-based, standard autocovariance-based, and weight-calibrated autocovariance-based methods, $\widehat{\bA}_{(\gamma)}$ is constructed by the $\gamma$ leading eigenvectors of the matrices $\widehat{\bOmega}_y,\widehat{\bM}_1=\sum_{k=1}^{m}\widehat{\bOmega}_y(k)\widehat{\bOmega}_y(k)^{\top},$ and $\widehat{\bM}=\sum_{k=1}^{m}\widehat{\bOmega}_y(k)\widehat{\bW}\widehat{\bOmega}_y(k)^{\top},$ respectively.
Then we propose the following information criterion:
$$
{\rm IC}(\gamma)=\log\{V(\gamma)\}+\gamma g(p,n),
$$ 
which is represented as the sum of the logarithm of the mean squared residual, defined as $V(\gamma)=(pn)^{-1}\sum_{t=1}^{n}\big\|\by_t-\widehat{\bA}_{(\gamma)}\widehat{\bA}_{(\gamma)}^{\top}\by_t\big\|^2,$ and the penalty term with $g(p,n)$ being the penalty function of $(p,n)$ to avoid overparameterization. Following \textcolor{blue}{Bai and Ng} (\textcolor{blue}{2002}), we suggest the following three examples of $g(p,n)$:
{\small\begin{equation}
    \label{eq.penalty}
    (i)\ g(p,n)=\frac{p+n}{pn}\log\left(\frac{pn}{p+n}\right),\
    (ii)\ g(p,n)=\frac{p+n}{pn}\log(p\wedge n),\
    (iii)\ g(p,n)=\frac{\log(p\wedge n)}{p\wedge n}.
\end{equation}
}The number of factors is then estimated as 
\begin{equation}
    \label{eq.IC}
    \hat{r}_0^{\IC}=\arg\min_{\gamma\in[\gamma_{\max}]}{\rm IC}(\gamma),
\end{equation}
where $\gamma_{\max}$ is a prespecified positive integer. Note that, for the covariance-based method, the information criterion is the same as that in \textcolor{blue}{Bai and Ng} (\textcolor{blue}{2002}). Thus, \eqref{eq.IC} naturally extends their criterion to accommodate the autocovariance-based framework. Since the penalty function $g(p,n)$ must dominate the convergence rate of the common component estimator, and the covariance-based and autocovariance-based methods share the same rate under some regularity conditions, adopting the same penalty functions as in \textcolor{blue}{Bai and Ng} (\textcolor{blue}{2002}) is justified.

\begin{table}[H]
    \small
    \centering
    \caption{The relative frequency estimate of $\eP(\hat{r}_0^{\IC}=r_0)$ and the average of $\hat{r}_0^{\IC}$ (in parentheses) for model~\eqref{eq.model_ex1} over 1000 simulation runs.}
    \label{tab.IC_num}{
    \begin{tabular}{l|ccc|ccc}
            \hline
            & \multicolumn{3}{c|}{$\delta_0=0.75$} & \multicolumn{3}{c}{$\delta_0=1$} \\
          & Cov & Auto & WAuto & Cov & Auto & WAuto \\
          \hline
          $p=100$ & 0.000  & 0.000  & 0.320  & 0.000  & 0.000  & 0.510  \\
          & (10.000)  & (9.965) & (4.967) & (10.000)  & (9.954) & (3.982) \\
    $p=200$ & 0.000  & 0.843  & 0.989  & 0.000  & 0.964  & 0.993  \\
          & (10.000)  & (3.206) & (3.011) & (10.000)  & (3.037) & (3.007) \\
    $p=300$ & 0.248  & 1.000  & 1.000  & 0.244  & 1.000  & 1.000  \\
          & (4.805) & (3.000)   & (3.000)   & (4.877) & (3.000)   & (3.000) \\
    $p=400$ & 0.946  & 1.000  & 1.000  & 0.945  & 1.000  & 1.000  \\
          & (3.057) & (3.000)   & (3.000)   & (3.058) & (3.000)   & (3.000) \\
        \hline
    \end{tabular}}
\end{table}

We next employ the same data generating process as in Section~\ref{sim.bfm} to assess the performance of the proposed information criteria. For implementation, we set $\gamma_{\max}=10,p=100,200,300,400$, adopt the second penalty function in \eqref{eq.penalty} when computing the information criteria, and retain all other settings from  Section~\ref{sim.bfm}. The average relative frequencies of $\hat{r}_0^{\IC}=r_0$ and the average $\hat{r}_0^{\IC}$ are reported in Table \ref{tab.IC_num}. A few trends are observable. Firstly, all three methods perform poorly when $p$ is small, with performance improving as $p$ increases. Particularly, for $p=100$ and $200$, Cov completely fails to identify the number of factors, whereas for $p=300$ and $400$, both Auto and WAuto achieve 100\% accuracy, highlighting the effectiveness of the proposed information criteria. Secondly, for $p=100$ and 200, WAuto consistently outperforms Auto, yielding significant gains in the estimation accuracy of $r_0$. This demonstrates the usefulness of the calibrating weight in enhancing the performance of the information-criterion method.

\spacingset{1.0}
\section*{References}
\begin{description}

    \item 
    Bai, J. and Ng, S. (2002). Determining the number of factors in approximate factor models, {\it Econometrica} {\bf 70}(1): 191--221.

    
    \item 
    Reinsel, G. C., Velu, R. P. and Chen, K. (2022). {\it Multivariate Reduced-Rank Regression: Theory, Methods and Applications}, Vol. 225, Springer.

    \item 
    Weyl, H. (1912). Das asymptotische verteilungsgesetz der eigenwerte linearer partieller differentialgleichungen (mit einer anwendung auf die theorie der hohlraumstrahlung), {\it Mathematische Annalen} {\bf 71}(4): 441--479.

    \item 
    Yu, Y., Wang, T. and Samworth, R. J. (2015). A useful variant of the Davis–Kahan theorem for statisticians, {\it Biometrika} {\bf 102}(2): 315--323.
    
    \item 
    Zhang, B., Pan, G., Yao, Q. and Zhou, W. (2024). Factor modeling for clustering high-dimensional time series, {\it Journal of the American Statistical Association} {\bf 119}(546): 1252--1263.
\end{description}

\end{document}